\documentclass[11pt,a4paper]{article}
\usepackage[margin=26mm,headheight=14pt]{geometry}
\usepackage[T1]{fontenc}
\usepackage[utf8]{inputenc}
\usepackage{amsmath,amsthm,mathtools,aliascnt}
\usepackage{lmodern,amssymb,mathrsfs}
\usepackage{microtype,enumitem,booktabs,array,xcolor,fancyhdr}
\usepackage{etoolbox}
\usepackage{tikz,float}
\usepackage{amsmath,amssymb,graphicx,tikz}
\usetikzlibrary{arrows.meta,decorations.pathreplacing}
\usetikzlibrary{arrows.meta,positioning,fit}
\usepackage[unicode,bookmarksnumbered,colorlinks=true,%
linkcolor=blue!40!black,citecolor=blue!40!black,%
urlcolor=blue!40!black]{hyperref}
\usepackage[nameinlink,noabbrev]{cleveref}
\crefname{subsection}{section}{sections}
\Crefname{subsection}{Section}{Sections}
\crefname{subsubsection}{section}{sections}
\Crefname{subsubsection}{Section}{Sections}

\hypersetup{pdftitle={Rapid mixing of Gibbs samplers via quantum Dobrushin--Shlosman conditions},pdfauthor={Cambyse Rouze and Daniel Stilck Franca},pdfsubject={Research
manuscript:
quantum belief propagation, arbitrary external fields, and finite-block rapid mixing}}
\allowdisplaybreaks[2]
\numberwithin{equation}{section}
\newtheorem{theorem}{Theorem}[section]
\newaliascnt{lemma}{theorem}
\newtheorem{lemma}[lemma]{Lemma}
\aliascntresetthe{lemma}
\crefname{lemma}{lemma}{lemmas}
\Crefname{lemma}{Lemma}{Lemmas}
\newaliascnt{proposition}{theorem}
\newtheorem{proposition}[proposition]{Proposition}
\aliascntresetthe{proposition}
\crefname{proposition}{proposition}{propositions}
\Crefname{proposition}{Proposition}{Propositions}
\newaliascnt{corollary}{theorem}

\aliascntresetthe{corollary}
\crefname{corollary}{corollary}{corollaries}
\Crefname{corollary}{Corollary}{Corollaries}
\theoremstyle{definition}
\newaliascnt{definition}{theorem}
\newtheorem{definition}[definition]{Definition}
\aliascntresetthe{definition}
\crefname{definition}{definition}{definitions}
\Crefname{definition}{Definition}{Definitions}
\newaliascnt{remark}{theorem}
\newtheorem{remark}[remark]{Remark}
\newtheorem*{remark*}{Remark}
\aliascntresetthe{remark}
\crefname{remark}{remark}{remarks}
\Crefname{remark}{Remark}{Remarks}
\DeclareMathOperator{\Tr}{Tr}
\DeclareMathOperator{\ad}{ad}
\DeclareMathOperator{\Ad}{Ad}
\DeclareMathOperator{\dist}{dist}
\DeclareMathOperator{\supp}{supp}
\DeclareMathOperator{\diam}{diam}
\DeclareMathOperator{\id}{id}
\DeclareMathOperator{\Lip}{Lip}

\DeclareMathOperator{\Cov}{Cov}
\newcommand{\HH}{\mathcal H}
\newcommand{\BB}{\mathcal B}
\newcommand{\CC}{\mathbb C}
\newcommand{\RR}{\mathbb R}
\newcommand{\NN}{\mathbb N}
\newcommand{\ZZ}{\mathbb Z}
\newcommand{\one}{\mathbf1}
\newcommand{\eps}{\varepsilon}
\newcommand{\norm}[1]{\left\lVert#1\right\rVert}
\newcommand{\nOne}[1]{\left\lVert#1\right\rVert_1}
\newcommand{\nW}[1]{\left\lVert#1\right\rVert_{W_1}}
\newcommand{\nD}[1]{\left\lVert#1\right\rVert_\diamond}
\newcommand{\abs}[1]{\left|#1\right|}
\newcommand{\cK}{\mathcal K}

\newcommand{\cT}{\mathcal T}
\newcommand{\cN}{\mathcal N}

\newcommand{\cJ}{\mathcal J}

\newcommand{\proofstep}[1]{\par\smallskip\noindent\textit{#1.}\ }

\setlist[enumerate]{leftmargin=2em,itemsep=3pt,topsep=3pt}

\crefname{condition}{condition}{conditions}
\Crefname{condition}{Condition}{Conditions}
\begin{document}

\begin{center}
 {\LARGE\bfseries Rapid mixing of Gibbs samplers via\\[4pt]
 quantum Dobrushin--Shlosman conditions\par}
 \vspace{10pt}

 {\large Cambyse Rouz\'e$^1$\qquad Daniel Stilck Fran\c{c}a$^2$\par}
 \vspace{7pt}
 {\small $^1$Inria, T\'el\'ecom Paris, LTCI, Institut Polytechnique de Paris, Palaiseau,
 France\par
 $^2$Department of Mathematical Sciences, University of Copenhagen, Denmark\par}
 \vspace{4pt}
 {\small $^1$\href{mailto:cambyse.rouze@inria.fr}{\texttt{cambyse.rouze@inria.fr}}\qquad
 $^2$\href{mailto:dsfranca@math.ku.dk}{\texttt{dsfranca@math.ku.dk}}\par}
 \vspace{8pt}
\end{center}

\begin{abstract}
\noindent 
Classical Dobrushin–Shlosman theory uses heat-bath block updates to identify sharp temperature thresholds for establishing the rapid mixing of Gibbs samplers. In particular, these mixing guarantees often hold even when stricter single-site conditions fail. We extend this approach to noncommuting quantum lattice systems through a quantum Dobrushin--Shlosman condition for finite-block dynamics, which we call smoothed heat-bath dynamics. We establish this condition in two regimes. For finite-range quantum spin chains at every fixed finite temperature, we prove logarithmic trace-norm mixing in system size, with normalized bounds uniform in on-site fields and single exponential in inverse temperature. We also prove stability under small local perturbations of interacting, possibly noncommuting reference Hamiltonians on graphs of polynomial volume growth. For classical references with bounded local interactions, uniform strong spatial mixing implies our quantum condition. With bounded local terms and fixed physical, geometric and block parameters, the resulting samplers prepare Gibbs states to trace-norm error $\varepsilon$ using $N\operatorname{polylog}(N/\varepsilon)$ gates and classical operations.

Our block updates combine a local Gibbs reset with quantum belief propagation to incorporate interactions across the block boundary, yielding exactly Gibbs-preserving, KMS-reversible channels. Two independently tunable parameters control contraction: an internal evolution time suppresses the influence of sites inside a block, while the block size allows interior contraction to dominate boundary influence. This finite-time relaxation has no counterpart in a classical heat-bath block update.

\end{abstract}

\setcounter{tocdepth}{1}
\begingroup
\small
\setlength{\parskip}{0pt}
\makeatletter
\patchcmd{\l@section}{1.0em}{0.5em}{}{}
\makeatother
\tableofcontents
\endgroup
\clearpage

\section{Introduction}
\label{sec:intro}
Preparing Gibbs states is a central problem in quantum many-body computation and has
motivated a growing number of quantum algorithms~\cite{Temme2011,CKBG2023,CKBG2025}.
A natural approach is to construct local dissipative dynamics whose stationary state
is the desired Gibbs state, following the role of Glauber dynamics in classical
sampling~\cite{Martinelli}. Chen, Kastoryano, and Gily\'en~\cite{CKG} constructed
quasi-local quantum Gibbs samplers satisfying exact KMS detailed balance. Subsequent
work has developed further constructions~\cite{DLL,GCDK} and established rapid mixing
in several regimes~\cite{RFA1,RFA2,BLMT,BC}. A central question is how to use spatial
properties of the Gibbs state to prove convergence of these dynamics, particularly
beyond high temperature.

Classically, the Dobrushin condition~\cite{Dob} provides one such connection: it bounds
the sensitivity of a single-site conditional distribution to changes in the surrounding
spins. Finite-volume Dobrushin--Shlosman (DS) conditions~\cite{DS1985} allow the same
comparison on a block of spins. The block size is then a parameter of the criterion.
Increasing it can make the contraction produced by resampling the block dominate the
influence of its boundary, even when a single-site test fails~\cite{Martinelli,CLW}.
Quantum Dobrushin-type criteria have been obtained from recovery maps and transport
inequalities~\cite{DPR,GR}, as well as from influence estimates for quantum Markov
dynamics~\cite{RFA2,BLMT}. 

Here we develop the finite-block approach for noncommuting
Hamiltonians. Our main application is rapid mixing for quantum spin chains at every
fixed finite temperature, with bounds independent of the norms of the
on-site Hamiltonians. We also prove that, under the stated locality
and geometric assumptions, a strict quantum DS condition for a family of block channels, uniform in the volume, persists throughout a
system-size-independent neighbourhood in the local interaction norm.
These results
use the Wasserstein contraction principle of~\cite{BLMT}; our
contribution is to construct block updates and bound their influences
in regimes where single-site estimates are insufficient.

Let $H$ be a finite-range Hamiltonian on a lattice $\Lambda$ of $N$ spins, with Gibbs
state $\sigma_{\beta,H}=e^{-\beta H}/\Tr(e^{-\beta H})$ at inverse temperature
$\beta>0$. For $B\subseteq\Lambda$, let $H_B$ contain the interaction terms supported
entirely in $B$, and set $H_B^0=H_B+H_{B^c}$. Thus $H-H_B^0$ consists of the
interactions crossing the boundary of $B$. Quantum belief propagation
(QBP)~\cite{HastingsQBP}, applied at inverse temperature $\beta/2$, gives an invertible
operator $O_B$ satisfying
\begin{equation}
 e^{-\beta H/2}=O_Be^{-\beta H_B^0/2}O_B^*,\qquad
 \|O_B^{\pm1}\|\le e^{\beta\|H-H_B^0\|/4}.
 \label{intro:congruence}
\end{equation}
Write $\sigma_\beta^B=e^{-\beta H_B}/\Tr(e^{-\beta H_B})$ for the Gibbs state of the
isolated block. We start from the replacement channel
$\mathcal R_B(X)=\sigma_\beta^B\otimes\Tr_B X$ and incorporate the crossing
interactions through the completely positive map
\begin{equation}
 \mathcal T_B(X)=O_B\left[\sigma_\beta^B\otimes
       \Tr_B(O_B^{-1}XO_B^{-*})\right]O_B^*.
 \label{intro:jumpmap}
\end{equation}
The KMS completion of~\cite{AmorimCarlen,DLL} adds the anticommutator and Hamiltonian
commutator terms needed to obtain a KMS-reversible Lindblad generator $\mathcal K_B$.
Its finite-time channel $\Phi_B(T)=e^{T\mathcal K_B}$ preserves $\sigma_{\beta,H}$
exactly. We use these channels as quantum substitutes for block heat-bath updates.
The norm bound in \eqref{intro:congruence} depends only on the crossing interactions;
combined with locality estimates, it allows us to control how much influence a block
update transmits across its boundary. Both the block size and the evolution time $T$
can then be chosen to obtain contraction.

\subsection{The finite-block criterion and its classical meaning}
For a block-update channel $\Phi_B$ and a site $j\in\Lambda$, define the
\emph{influence of site $j$ under the update $\Phi_B$} by
\begin{equation}\label{equ:def_db}
 D_B(j)=\sup_{\substack{X=X^*\ne0\\\Tr_jX=0}}
       \frac{\|\Phi_B(X)\|_{W_1}}{\|X\|_1/2},
\end{equation}
where $\|\cdot\|_{W_1}$ is the quantum Wasserstein norm~\cite{DPMTL}. The condition
$\Tr_jX=0$ represents a difference between states with the same marginal outside $j$.
Thus $D_B(j)$ measures the size of this difference after the update, including its
effect outside $B$.

Suppose that all updates preserve $\sigma_{\beta,H}$, and assign a positive rate
$a_B$ to each block. Our contraction condition is
$\kappa:=\min_{j\in\Lambda}\sum_Ba_B[1-D_B(j)]>0$. In words, for every initial
site $j$, the weighted sum of the influences must be smaller than the total update
rate. The resulting generator $\mathcal L=\sum_Ba_B(\Phi_B-\id)$ contracts the
quantum Wasserstein distance at rate $\kappa$. In particular, its evolution from any
initial state is within trace norm $\varepsilon$ of $\sigma_{\beta,H}$ after time
$\kappa^{-1}\log(2N/\varepsilon)$, for $0<\varepsilon<1$.

For a classical Hamiltonian, this condition recovers the usual block transport
criterion when the updates are heat-bath conditional expectations. Indeed, on
product-basis diagonal inputs, the quantum Wasserstein distance equals the classical
Wasserstein distance for the Hamming metric. A heat-bath update samples the spins in
$B$ from their Gibbs distribution conditioned on the exterior configuration. Consider
two configurations that differ only at $j$. If $j\in B$, their exterior configurations
agree, and the update produces identical distributions. If $j\notin B$, the difference
at $j$ remains and can also change the distribution sampled inside $B$. The diagonal
restriction of \eqref{equ:def_db} therefore equals $0$ in the first case and $1+c_{Bj}(\pi)$
in the second. Here $\pi$ is the Gibbs distribution, and
$c_{Bj}(\pi)$ is the largest Hamming Wasserstein distance between the
two conditional distributions on $B$, over exterior configurations differing only
at $j$.

If each site belongs to exactly $m$ labelled blocks and $a_B=1/m$, the classical
condition becomes
\begin{align}\label{equ:classical_ds}
 \sum_{B:j\notin B}{c_{Bj}(\pi)}\le(1-\kappa)m,\qquad j\in\Lambda.
\end{align}
This is the finite-block transport form of the DS condition~\cite{DS1985,CLW}.
\Cref{thm:classical} shows that a strict inequality of this form remains valid when
heat-bath updates are replaced by the diagonal restrictions of $e^{T\mathcal K_B}$,
for a sufficiently large finite $T$. The approximation error decays exponentially
in $T$, at a rate controlled by the oscillation of the crossing energy. Although
classical heat-bath updates can be used directly, finite-time evolution is essential
to our quantum construction: it provides enough relaxation inside the block while
retaining quantitative locality bounds. These bounds depend on $T$ and do not give
uniform locality as $T\to\infty$. 

Since these channels converge to heat-bath block updates as $T\to\infty$ in the classical setting, we refer to them as \textit{smoothed heat-bath block updates} and the dynamics they generate as the \textit{smoothed heat-bath dynamics}. Moreover, for classical Hamiltonians with uniformly bounded interactions, satisfying uniform DS strong spatial mixing,
\Cref{qp:classical-qds} proves that smoothed heat-bath updates over sufficiently large blocks satisfy the full quantum DS condition, with a uniform bound
on their evolution times. 

\subsection{Main results}
Our main results are the following.
\begin{enumerate}
\item \emph{Rapid mixing in one dimension at every finite temperature.}
We prove rapid mixing for noncommuting spin chains at every fixed inverse temperature
$\beta<\infty$. Consider a chain of $N$ spins of local dimension at most $d$, with
Hamiltonian $H=\sum_{i=1}^N V_i+\sum_{i=1}^{N-1}W_{i,i+1}$, where
  the on-site terms $V_i$ are arbitrary self-adjoint operators.
We use the interaction strength $J_{\mathrm{int}}$ defined in
\eqref{eq:interaction-strength-gibbs}, which excludes the on-site terms.
\Cref{q1:main} shows that there are a block length $\ell$ and an evolution time $T$,
both independent of $N$, for which the translated interval updates
$\Phi_B=e^{T\mathcal K_B}$ satisfy the quantum DS condition with $\kappa=1/2$.
These parameters can be chosen to satisfy
\begin{equation}
 \ell\le C_d(1+\beta J_{\mathrm{int}})^8e^{8\beta J_{\mathrm{int}}},\qquad
 T\le C_de^{2\beta J_{\mathrm{int}}}\log(2+\beta J_{\mathrm{int}}),
 \label{intro:parameters}
\end{equation}
where $C_d$ depends only on $d$. In particular, these bounds are uniform in the norms
of the on-site terms. The generator
$\mathcal L=\ell^{-1}\sum_B(\Phi_B-\id)$ then has trace-norm mixing time at most
$2\log(2N/\varepsilon)$, for $0<\varepsilon<1$.

The generator $\mathcal L$ assigns rate $1/\ell$ to each block update, independently of
its internal evolution time. To account for this time and the strength of the block
generators, we use the bound
$\|\mathcal K_B\|_\diamond\le q_\beta:=C_de^{2\beta J_{\mathrm{int}}}(1+\beta
J_{\mathrm{int}})^4
\log(2+\beta J_{\mathrm{int}})$. Since
$\Phi_B=e^{q_\beta T(\mathcal K_B/q_\beta)}$, each update takes time $q_\beta T$
under a generator of diamond norm at most one. With this normalization,
rescaling the mixing bound of \Cref{q1:main} gives
\begin{equation}
 t_{\mathrm{mix}}(\mathcal L/(q_\beta T),\varepsilon)
 \le C_de^{4\beta J_{\mathrm{int}}}(1+\beta J_{\mathrm{int}})^4\log^2(2+\beta
J_{\mathrm{int}})\log(2N/\varepsilon).
 \label{intro:normclock}
\end{equation}
The dependence on system size is therefore logarithmic, and the dependence on
$\beta J_{\mathrm{int}}$ is at most exponential. The result extends to chains with any
fixed finite
interaction range by grouping consecutive sites, with constants depending on the range.
As in classical rapid-mixing results~\cite{Zeg1D,SZdiscrete}, decay of correlations is
central to the proof: together with locality of the dynamics, it allows us to verify
the block influence condition. The proof is given in Section~\ref{q1:section}.

\item \emph{Stability under local noncommuting perturbations.}

Let $H^{(0)}$ be a finite-range Hamiltonian whose block channels
satisfy the quantum DS condition with constant $\kappa_0>0$,
independently of the system size. We perturb it to $H=H^{(0)}+Q$,
where $Q=\sum_XQ_X$ has prescribed finite range and local strength
$\delta:=\sup_{i\in\Lambda}\sum_{X\ni i}\|Q_X\|$.
Under the geometric and block assumptions of \Cref{qp:general-main},
the same construction for $H$ satisfies this condition with constant $\kappa_0/2$
whenever $\delta\le\delta_*$, for some $\delta_*>0$ independent of $N$.
The threshold depends on the fixed geometric and intersite-interaction
bounds, inverse temperature, block parameters, update rates, evolution
time, and reference DS constant. With these data fixed, it is independent
of the reference on-site norms. The reference terms need not commute,
and the total perturbation norm may grow as $N\delta$.

As a concrete example, consider the ferromagnetic Ising model on a square
torus at $\beta J_{\mathrm{int}}=1.4$ and zero longitudinal field.
The classical single-site test fails, while strong spatial mixing
implies the quantum DS condition for sufficiently large blocks. The general stability theorem
therefore gives rapid mixing after adding a sufficiently small transverse
field (\Cref{qp:examples}).

\end{enumerate}

\paragraph{Implementation of the dynamics.}
   Together with the mixing estimates above, 
\Cref{qi:theorem} 
gives
quantum algorithms preparing  the Gibbs state    with
 $N\operatorname{polylog}(N/\varepsilon)$ elementary gates and
classical operations.

The bound holds at fixed inverse
temperature, interaction geometry and strength, local dimension,
block size, evolution time $T$, update rates, and positive
contraction rate $\kappa$. For this implementation result, we also
impose a fixed bound on the norms of the on-site Hamiltonians.
The local terms and construction parameters are supplied to the
required precision. The algorithm combines approximations of the
QBP operators with Lindblad simulation, using the methods
of~\cite{BCCKS,FLT2023,LW}.

\subsubsection{Proof strategy and main ingredients}
The first step is to establish Gibbs preservation and locality for the block
channels. The QBP identity \eqref{intro:congruence} gives the algebraic relation
needed for KMS detailed balance. Its norm bound depends only on interactions crossing
the block boundary. Lieb--Robinson bounds then control the error made by restricting
the dressing operator and the coherent correction to a neighborhood of that boundary
(\Cref{qg:construction,qg:collars}). Our estimates are uniform in the on-site
Hamiltonians because their evolution is treated in the interaction picture.

The main step in the chain proof is to convert decay of Gibbs correlations into a
bound on the influence of a block update. Duhamel identities separate the part of the
evolution containing no Gibbs replacement from the terms containing a replacement.
The first contribution decays exponentially in the internal time $T$. For the second,
clustering controls the influence between the block interior and its boundary, even
when the input is entangled with the exterior. Before summing these estimates, we
take their minimum with the bound from trace-norm contractivity. This controls the
short-distance terms independently of the clustering prefactor and gives the polynomial
factors
in \eqref{intro:parameters} (\Cref{q1:influence}). The required
clustering estimate for Gibbs states of finite intervals is proved using QBP and
inclusion--exclusion in \Cref{qc:clustering}. We then choose $\ell$ and $T$ so that
the contraction in the interior exceeds the total boundary influence.

Our perturbation argument compares the block channels on finite
regions and controls the locality errors for both the reference and
perturbed Hamiltonians. Telescoping both sequences of restricted
channels bounds the change in total influence uniformly in $N$. A sufficiently
small local perturbation therefore preserves the quantum DS
condition with a positive constant (\Cref{qp:general-main}). For a classical reference, an influence estimate for the native
QBP conditional expectations gives the required influence bounds;
a uniform bound on the interactions combined with a
compactness argument yield bounded block evolution times
(\Cref{qp:classical-qds}).

\subsection{Relation to previous work}
For commuting Hamiltonians, Kastoryano--Brand\~ao~\cite{KBcommuting} related
clustering to a uniform spectral gap. Subsequent work established rapid mixing at
every finite temperature in commuting spin chains~\cite{BCGcommuting,BCGMLSI,KACR}.
For commuting code Hamiltonians, Stengele et al.\ derived modified logarithmic
Sobolev inequalities from DS-type equilibrium conditions, including all-temperature
rapid mixing for two-dimensional Abelian quantum double models~\cite{SCGLcss,SCGL2026}.
Our chain theorem establishes logarithmic trace-norm mixing without a commutativity
assumption.

Several preparation algorithms use equilibrium structure directly.
Brand\~ao--Kastoryano~\cite{BKpreparation} gave a preparation scheme assuming decay of
correlations and an approximate Markov property for Gibbs states on induced subgraphs.
Other approaches include the dimension-reduction algorithm of
Bilgin--Boixo~\cite{BilginBoixo} and tensor-network approximations, which take the
form of matrix product operators in one dimension~\cite{MSVC}. Approximate Markov
properties have also been studied for one-dimensional Gibbs states~\cite{KatoBrandao}
and at high temperature through bounds on conditional mutual
information~\cite{KKB,KKBerratum,KatoKuwaharaBP}.

For noncommuting Hamiltonians, Rouz\'e--Stilck Fran\c{c}a--Alhambra proved
high-temperature rapid mixing and nearly linear preparation for KMS Gibbs
samplers~\cite{RFA1,RFA2}. Bakshi--Tan obtained high-temperature rapid mixing
uniformly in arbitrary external fields~\cite{BT}, and \v{S}m\'id et al.\ proved
rapid mixing for sufficiently weak interactions around noninteracting reference
systems~\cite{SMB2025}. Our stability result allows any finite-range
reference whose smoothed heat-bath channels satisfy a strict quantum DS
condition, including noncommuting references. For classical systems
satisfying uniform DS strong spatial mixing, we prove the quantum DS
condition for the QBP dynamics and apply the same theorem. The allowed
perturbation is controlled by the reference DS constant, the block
evolution times, and locality bounds.

Bergamaschi--Chen~\cite{BC} proved that noncommuting quantum spin chains admit Gibbs
samplers with a system-size independent spectral gap at every finite temperature,
and obtained nearly linear preparation by a separate adiabatic construction. A
spectral-gap estimate alone gives a worst-case trace-norm mixing bound involving
$\log\lambda_{\min}(\sigma_{\beta,H})^{-1}$, where $\lambda_{\min}$ denotes the
smallest eigenvalue. For bounded local Hamiltonians this term can grow linearly in
$N$. Our block dynamics instead satisfy a direct logarithmic trace-norm mixing bound.

Harrow--Mehraban--Soleimanifar~\cite{HMS} studied the relation between
partition-function analyticity and correlation decay, extending another aspect of
classical DS theory to quantum systems. Their analyticity condition allows local
measurement insertions; classically, these include fixing boundary spins. Our
condition concerns the influence of dynamical updates on arbitrary initial states
and reduces to the classical block transport condition. The two approaches thus
address equilibrium sensitivity and dynamical contraction, respectively; an implication
from their analyticity condition to our quantum block condition is not assumed here.

The most direct comparison is with Bakshi--Liu--Moitra--Tang~\cite{BLMT}, whose
Wasserstein criterion applies to general quantum channels and is verified there for
updates centered on single sites. We use the same contraction principle. The new
ingredient is a construction of block channels whose influences can be controlled
by decay of correlations, with quantitative bounds on both the block length $\ell$
and the internal time $T$.

For classical heat-bath dynamics, the update is an exact conditional expectation, so
its definition does not require an internal time $T$. The analogous projection is
problematic quantumly: a state-preserving conditional expectation onto the full
complementary matrix algebra would force the state to be a product across the cut
(by the bimodule property of conditional expectations~\cite{Tomiyama1957}).
We therefore use the finite-time channels
$e^{T\mathcal K_B}$ as pseudo-conditional expectations. Increasing $T$ reduces the
remaining interior influence, while the error in approximating the channel by local
dynamics grows linearly in $T$. Increasing $\ell$ changes the balance between the
number of interior sites and the boundary influence. In one dimension, bounded
boundary size and decay of correlations allow both parameters to be chosen so that
contraction holds at every finite temperature.

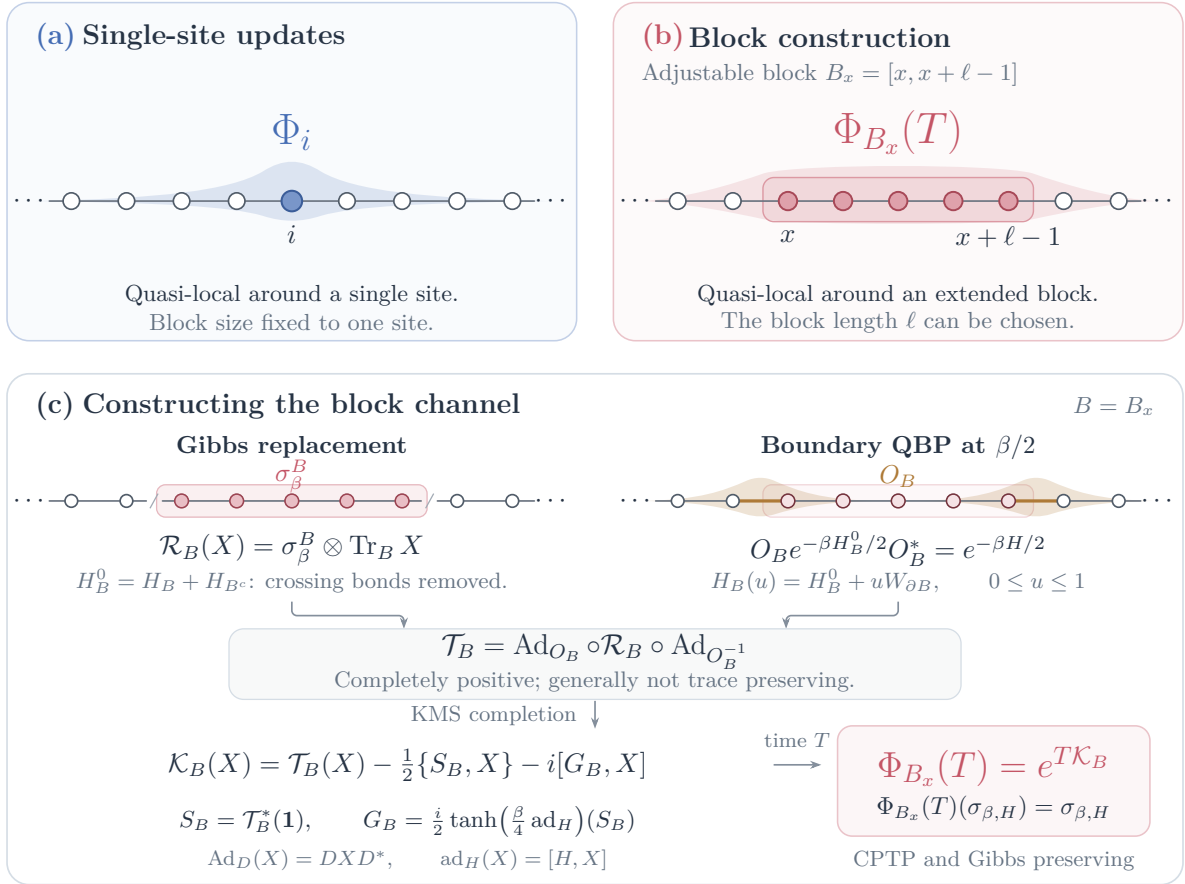
\begin{figure}[t]
  \centering

\begingroup
\definecolor{bqpInk}{HTML}{263344}
\definecolor{bqpMuted}{HTML}{657383}
\definecolor{bqpLine}{HTML}{D5DDE5}
\definecolor{bqpBlue}{HTML}{456DB4}
\definecolor{bqpRed}{HTML}{C45566}
\definecolor{bqpGold}{HTML}{AE7A37}
\tikzset{
  bqp/title/.style={font=\fontsize{11}{13}\selectfont\bfseries,
                   text=bqpInk,anchor=west,inner sep=0pt},
  bqp/subtitle/.style={font=\fontsize{9}{11}\selectfont,
                      text=bqpMuted,anchor=west,inner sep=0pt},
  bqp/text/.style={font=\fontsize{9}{11}\selectfont,
                  text=bqpMuted,align=center,inner sep=0pt},
  bqp/math/.style={font=\fontsize{10.4}{12.5}\selectfont,
                  text=bqpInk,align=center,inner sep=0pt},
  bqp/chain/.style={draw=bqpInk!68,line width=.65pt,line cap=round},
  bqp/arrow/.style={-{Stealth[length=1.55mm,width=1.1mm]},
                   draw=bqpMuted!80,line width=.65pt,
                   rounded corners=2pt},
  bqp/site/.style={draw=bqpInk!80,fill=white,line width=.65pt},
  bqp/step/.style={font=\fontsize{9.5}{11.5}\selectfont\bfseries,
                  text=bqpInk,inner sep=0pt,align=center}
}
\resizebox{\textwidth}{!}{
\begin{tikzpicture}[x=1cm,y=1cm]
  \path[use as bounding box] (0,0) rectangle (16,12);

  \path[rounded corners=3mm,fill=bqpBlue!3,draw=bqpBlue!32,
        line width=.6pt] (0,7.4) rectangle (7.75,12);
  \node[bqp/title,text=bqpBlue] at (.38,11.53) {(a)};
  \node[bqp/title] at (1.02,11.53) {Single-site updates};
  \node[bqp/math,font=\fontsize{17}{19}\selectfont,text=bqpBlue]
       at (3.875,10.24) {$\Phi_i$};

  \begin{scope}[shift={(3.875,9.30)}]

    \path[fill=bqpBlue,fill opacity=.16]
      (-3.50,0)
      .. controls (-1.80,.035) and (-.85,.14) .. (-.40,.40)
      .. controls (-.12,.58) and (.12,.58) .. (.40,.40)
      .. controls (.85,.14) and (1.80,.035) .. (3.50,0)
      .. controls (1.80,-.025) and (.85,-.07) .. (.40,-.20)
      .. controls (.12,-.29) and (-.12,-.29) .. (-.40,-.20)
      .. controls (-.85,-.07) and (-1.80,-.025) .. cycle;
    \draw[bqp/chain] (-3.30,0) -- (3.30,0);
    \foreach \k in {-4,...,4}
      \path[bqp/site] ({.75*\k},0) circle[radius=1.12mm];
    \path[draw=bqpBlue!85!black,fill=bqpBlue!72,line width=.8pt]
      (0,0) circle[radius=1.40mm];
    \node[bqp/math] at (-3.55,0) {$\cdots$};
    \node[bqp/math] at (3.55,0) {$\cdots$};
    \node[bqp/math] at (0,-.45) {$i$};
  \end{scope}
  \node[bqp/text,text=bqpInk] at (3.875,8.02)
    {Quasi-local around a single site.};
  \node[bqp/text] at (3.875,7.65)
    {Block size fixed to one site.};

  \path[rounded corners=3mm,fill=bqpRed!3,draw=bqpRed!35,
        line width=.6pt] (8.25,7.4) rectangle (16,12);
  \node[bqp/title,text=bqpRed] at (8.63,11.53) {(b)};
  \node[bqp/title] at (9.27,11.53) {Block construction};
  \node[bqp/subtitle] at (8.63,11.03)
    {Adjustable block $B_x=[x,x+\ell-1]$};
  \node[bqp/math,font=\fontsize{17}{19}\selectfont,text=bqpRed]
       at (12.125,10.24) {$\Phi_{B_x}(T)$};

  \begin{scope}[shift={(12.125,9.30)}]

    \path[fill=bqpRed,fill opacity=.14]
      (-3.50,0)
      .. controls (-2.90,.035) and (-2.28,.22) .. (-1.82,.38)
      .. controls (-1.30,.51) and (1.30,.51) .. (1.82,.38)
      .. controls (2.28,.22) and (2.90,.035) .. (3.50,0)
      .. controls (2.90,-.025) and (2.28,-.15) .. (1.82,-.25)
      .. controls (1.30,-.33) and (-1.30,-.33) .. (-1.82,-.25)
      .. controls (-2.28,-.15) and (-2.90,-.025) .. cycle;
    \path[rounded corners=1.5mm,fill=bqpRed,fill opacity=.075,
          draw=bqpRed!60,line width=.55pt]
      (-1.84,-.28) rectangle (1.84,.32);
    \draw[bqp/chain] (-3.30,0) -- (3.30,0);
    \foreach \k in {-4,...,4}
      \path[bqp/site] ({.75*\k},0) circle[radius=1.12mm];
    \foreach \k in {-2,...,2}
      \path[draw=bqpRed!85!black,fill=bqpRed!55,line width=.75pt]
        ({.75*\k},0) circle[radius=1.25mm];
    \node[bqp/math] at (-3.55,0) {$\cdots$};
    \node[bqp/math] at (3.55,0) {$\cdots$};
    \node[bqp/math] at (-1.5,-.48) {$x$};
    \node[bqp/math] at (1.5,-.48) {$x+\ell-1$};
  \end{scope}
  \node[bqp/text,text=bqpInk] at (12.125,8.02)
    {Quasi-local around an extended block.};
  \node[bqp/text] at (12.125,7.65)
    {The block length $\ell$ can be chosen.};

  \path[rounded corners=3mm,fill=white,draw=bqpLine,line width=.65pt]
    (0,0) rectangle (16,6.95);
  \node[bqp/title] at (.38,6.51) {(c)};
  \node[bqp/title] at (1.02,6.51) {Constructing the block channel};
  \node[bqp/subtitle,anchor=east] at (15.62,6.51) {$B=B_x$};

  \node[bqp/step] at (3.875,5.96) {Gibbs replacement};
  \begin{scope}[shift={(3.875,5.22)}]
    \path[rounded corners=1.2mm,fill=bqpRed!9,
          draw=bqpRed!42,line width=.5pt]
      (-1.84,-.22) rectangle (1.84,.22);
    \draw[bqp/chain] (-3.30,0) -- (-1.98,0);
    \draw[bqp/chain] (-1.76,0) -- (1.76,0);
    \draw[bqp/chain] (1.98,0) -- (3.30,0);

    \foreach \c in {-1.875,1.875}{
      \draw[draw=bqpMuted!60,line width=.5pt]
        ({\c-.055},-.10) -- ({\c+.055},.10);
    }
    \foreach \k in {-4,...,4}
      \path[bqp/site] ({.75*\k},0) circle[radius=.86mm];
    \foreach \k in {-2,...,2}
      \path[draw=bqpRed!80!black,fill=bqpRed!38,line width=.6pt]
        ({.75*\k},0) circle[radius=.93mm];
    \node[bqp/math,font=\fontsize{10}{11}\selectfont,text=bqpRed]
      at (0,.34) {$\sigma_\beta^B$};
    \node[bqp/math] at (-3.55,0) {$\cdots$};
    \node[bqp/math] at (3.55,0) {$\cdots$};
  \end{scope}
  \node[bqp/math] at (3.875,4.57)
    {$\mathcal R_B(X)=\sigma_\beta^B\otimes\operatorname{Tr}_B X$};
  \node[bqp/text,font=\fontsize{8.6}{10.4}\selectfont] at (3.875,4.13)
    {$H_B^0=H_B+H_{B^c}$: crossing bonds removed.};

  \node[bqp/step] at (12.125,5.96)
    {Boundary QBP at $\beta/2$};
  \begin{scope}[shift={(12.125,5.22)}]
    \path[rounded corners=1.2mm,fill=bqpRed!4,
          draw=bqpRed!24,line width=.5pt]
      (-1.84,-.22) rectangle (1.84,.22);
    \foreach \c in {-1.875,1.875}{
      \begin{scope}[shift={(\c,0)}]
        \path[fill=bqpGold,fill opacity=.21]
          (-1.6,0)
          .. controls (-.82,.015) and (-.50,.14) .. (-.20,.27)
          .. controls (-.07,.33) and (.07,.33) .. (.20,.27)
          .. controls (.50,.14) and (.82,.015) .. (1.6,0)
          .. controls (.82,-.015) and (.50,-.09) .. (.20,-.18)
          .. controls (.07,-.22) and (-.07,-.22) .. (-.20,-.18)
          .. controls (-.50,-.09) and (-.82,-.015) .. cycle;
      \end{scope}
    }
    \draw[bqp/chain] (-3.30,0) -- (3.30,0);
    \draw[draw=bqpGold,line width=1.25pt] (-2.25,0) -- (-1.50,0);
    \draw[draw=bqpGold,line width=1.25pt] (1.50,0) -- (2.25,0);
    \foreach \k in {-4,...,4}
      \path[bqp/site] ({.75*\k},0) circle[radius=.86mm];
    \foreach \k in {-2,...,2}
      \path[draw=bqpRed!60!black,fill=bqpRed!17,line width=.6pt]
        ({.75*\k},0) circle[radius=.93mm];
    \node[bqp/math,font=\fontsize{10}{11}\selectfont,text=bqpGold]
      at (0,.34) {$O_B$};
    \node[bqp/math] at (-3.55,0) {$\cdots$};
    \node[bqp/math] at (3.55,0) {$\cdots$};
  \end{scope}
  \node[bqp/math] at (12.125,4.57)
    {$O_Be^{-\beta H_B^0/2}O_B^*=e^{-\beta H/2}$};
  \node[bqp/text,font=\fontsize{8.6}{10.4}\selectfont] at (12.125,4.13)
    {$H_B(u)=H_B^0+uW_{\partial B},\qquad 0\le u\le1$};

  \draw[bqp/arrow] (3.875,3.85) -- (3.875,3.66) -- (5.40,3.66) -- (5.40,3.48);
  \draw[bqp/arrow] (12.125,3.85) -- (12.125,3.66) -- (10.60,3.66) -- (10.60,3.48);
  \path[rounded corners=2mm,fill=bqpInk!3,draw=bqpLine,line width=.5pt]
    (3.02,2.52) rectangle (12.98,3.48);
  \node[bqp/math,font=\fontsize{11.1}{13}\selectfont] at (8,3.17)
    {$\mathcal T_B=\operatorname{Ad}_{O_B}\circ\mathcal R_B
                           \circ\operatorname{Ad}_{O_B^{-1}}$};
  \node[bqp/text,font=\fontsize{8.7}{10.4}\selectfont] at (8,2.76)
    {Completely positive; generally not trace preserving.};

  \draw[bqp/arrow] (8,2.50) -- (8,2.12);
  \node[bqp/text,anchor=east,font=\fontsize{8.5}{10}\selectfont]
    at (7.76,2.30) {KMS completion};
  \node[bqp/math,font=\fontsize{10.4}{12}\selectfont]
    at (5.45,1.63)
    {$\mathcal K_B(X)=\mathcal T_B(X)-\tfrac12\{S_B,X\}-i[G_B,X]$};
  \node[bqp/math,font=\fontsize{9.0}{11}\selectfont]
    at (5.45,.86)
    {$S_B=\mathcal T_B^*(\mathbf 1),\qquad
      G_B=\tfrac i2\tanh\bigl(\tfrac\beta4\operatorname{ad}_H\bigr)(S_B)$};
  \node[bqp/text,font=\fontsize{8}{9.5}\selectfont] at (5.45,.34)
    {$\operatorname{Ad}_D(X)=DXD^*,\qquad \operatorname{ad}_H(X)=[H,X]$};

  \draw[bqp/arrow] (10.43,1.63) -- (11.08,1.63);
  \node[bqp/text,font=\fontsize{8}{9.5}\selectfont] at (10.75,1.97) {time $T$};
  \path[rounded corners=2mm,fill=bqpRed!5,draw=bqpRed!35,line width=.6pt]
    (11.30,.66) rectangle (15.55,2.16);
  \node[bqp/math,text=bqpRed,font=\fontsize{14}{16}\selectfont]
    at (13.425,1.61) {$\Phi_{B_x}(T)=e^{T\mathcal K_B}$};
  \node[bqp/math,font=\fontsize{9.5}{11.5}\selectfont]
    at (13.425,1.04) {$\Phi_{B_x}(T)(\sigma_{\beta,H})=\sigma_{\beta,H}$};
  \node[bqp/text,font=\fontsize{8.5}{10.5}\selectfont]
    at (13.425,.33) {CPTP and Gibbs preserving};
\end{tikzpicture}
}
\endgroup
\caption{\textbf{Single-site and block updates.}
(a) The channels analyzed in \cite{RFA2,BLMT} are centered on individual sites.
(b) Our channels act around intervals $B_x=[x,x+\ell-1]$ of adjustable length.
The shaded regions indicate decaying spatial tails; neither channel is strictly
supported on the highlighted sites.
(c) The block construction starts from the Gibbs replacement $\mathcal R_B$.
Writing $W_{\partial B}=H-H_B^0$, QBP along
$H_B(u)=H_B^0+uW_{\partial B}$ gives the operator $O_B$ in
\eqref{intro:congruence}. Dressing $\mathcal R_B$ gives the completely positive map
$\mathcal T_B$. KMS completion then yields the generator $\mathcal K_B$ and the
Gibbs-preserving channel $\Phi_B(T)$. The arrows describe this algebraic construction.}
\end{figure}

For mutually commuting Hamiltonian terms, our field-free single-site qubit
specialization agrees with the BLMT sampler up to a fixed rescaling of time;
see \Cref{rem:commuting-BLMT} for more details.
The relation between
the contraction criteria can also be made explicit. For single-site channels
$\Phi_i$, consider the random-site update $P=N^{-1}\sum_{i=1}^N\Phi_i$ and the
influence matrix
\[
 \mathsf M_{ij}
 =\frac{1}N
   \sup_{\substack{X=X^*\ne0\\ \Tr_jX=0}}
   \frac{\|\Phi_i(X)\|_{W_1}}{\|X\|_{W_1}}
 =\frac{1}N D_{\Phi_i}(j).
\]
The equality follows from $\|X\|_{W_1}=\|X\|_1/2$ when $\Tr_jX=0$.
Rows index the update $i$, and columns index the site $j$ of the initial difference.
A column sum therefore measures the total influence of that difference, averaged
over updates. The BLMT condition bounds the largest column sum by
$\|\mathsf M\|_{1\to1}\le1-\gamma/N$, with $\gamma>0$ independent of $N$.
For our weighted block updates, the same calculation gives
\[
 a=\sum_Ba_B,\qquad
 \mathsf M_{Bj}=\frac{a_B}{a}D_B(j),\qquad
 \|\mathsf M\|_{1\to1}=1-\frac{\kappa}{a}<1.
\]
Thus the distinction is in the updates and the estimates used to verify contraction.
The general matrix criterion accommodates blocks in either formulation. We use the
term finite-block DS condition to emphasize the choice of a spatial scale on which
the influence bound holds. This distinction is separate from row-sum versus
column-sum conventions, which are also used to distinguish Dobrushin and DS conditions
in the literature, including~\cite[Section~2.1, footnote~5]{BLMT}.

\paragraph{Classical block dynamics and functional inequalities.}
The classical theory relates spatial mixing to functional inequalities uniformly
over finite volumes and boundary conditions. Stroock--Zegarlinski established the
equivalence between DS mixing and logarithmic Sobolev inequalities for the finite-range
spin systems considered in~\cite{SZmix,SZdiscrete}. Lu--Yau developed block and
martingale methods for Glauber and Kawasaki dynamics~\cite{LuYau}, while
Zegarlinski~\cite{Zeg92} obtained logarithmic Sobolev inequalities from Dobrushin
influence estimates and local inequalities. In one dimension, finite-state systems
with bounded finite-range interactions satisfy a uniform logarithmic Sobolev
inequality, and hence logarithmic mixing, at every finite
temperature~\cite{Zeg1D,SZdiscrete}. Our chain theorem gives a corresponding quantum
result in Wasserstein distance and trace norm.

The role of blocks is particularly clear when a single-site condition fails inside
the uniqueness regime. For the ferromagnetic nearest-neighbor Ising model on
$\ZZ^D$, Shlosman~\cite{Shlosman} showed that finite-volume DS tests, optimized over
the volume, detect the entire high-temperature uniqueness regime. A single-site test
does not have this property in general. Larger blocks are not necessary for every
model: for the homogeneous one-dimensional zero-field Ising heat bath, the exact
single-site influence sum is $\tanh(\beta J_{\mathrm{int}})<1$ at every finite temperature.
Allowing blocks provides a criterion that can also exploit correlations over several
sites when the single-site estimate is insufficient.

Wasserstein contraction is also related to coarse Ricci curvature~\cite{Ollivier}.
M\"unch connected this curvature to isoperimetry and logarithmic Sobolev
inequalities~\cite{Munch}. Caputo--M\"unch--Salez proved that, under nonnegative
sectional curvature, a finite-time Wasserstein contraction implies relative-entropy
contraction with the same coefficient~\cite{CMS}. This provides a further classical
connection between transport estimates and entropy decay. Our results establish
Wasserstein contraction and trace-norm mixing for quantum block dynamics; deriving
analogous quantum functional inequalities remains an open direction.

\paragraph{Organization of the paper.}
After the preliminaries, Section~\ref{qg:construction} introduces
the smoothed heat-bath generators, establishes their real-time locality
estimates, and states the finite-time implementation theorem.
Section~\ref{sec:criteria} states the contraction criterion and
proves that, for uniformly bounded interactions,
uniform classical DS strong spatial mixing implies the full quantum DS condition. Section~\ref{q1:section} states and proves the
one-dimensional mixing theorem, using decay of correlations and
locality of the block dynamics. Section~\ref{qp:section} establishes
stability under local perturbations and gives applications to
classical reference systems.
Appendix~\ref{app:clustering} contains the clustering proofs.
Appendix~\ref{app:implementation} proves \Cref{qi:theorem} about efficient implementation of our generators.

\section{Preliminaries}
\label{sec:prelim}

\subsection{Interactions and external fields}
\label{sec:interactions}
Let $\Lambda$ be a finite set of $N$ sites   . At each site $i$ we associate a
Hilbert space $\HH_i\simeq\CC^{d_i}$, with $2\le d_i\le d$, and write
$\HH_U=\bigotimes_{i\in U}\HH_i$ for $U\subseteq\Lambda$.
The algebra of operators on $\HH_U$ is denoted by $\BB(\HH_U)$.
For $A\in\BB(\HH_U)$, we denote by $A^*$ its adjoint.
A state is a positive operator of trace one. For a state $\rho$ on
$\HH_\Lambda$, its marginal on $U$ is $\rho_U=\Tr_{U^c}\rho$, where
$U^c=\Lambda\setminus U$. Operators on $U$ are extended by the identity
on $U^c$ whenever they occur in a full-system expression; tensor factors
are reordered according to their site labels.
For a nonnegative radius $a$, set
$U^{\le a}=\{i\in\Lambda:\dist(i,U)\le a\}$.
We distinguish the marginal $(\sigma_{\beta,H})_U$ from the Gibbs state
$\sigma_\beta^U$ of the truncated Hamiltonian $H_U$ defined below.

We
  describe the interactions by a finite labelled family
$(h_\alpha)_{\alpha\in\mathcal A}$ and separate the on-site terms:
\begin{equation}
 \begin{gathered}
 H=\sum_{i\in\Lambda}V_i+\sum_{\alpha\in\mathcal A}h_\alpha,
 \qquad V_i=V_i^*,\quad\supp V_i\subseteq\{i\},\\
 h_\alpha=h_\alpha^*,\qquad\supp h_\alpha\subseteq X_\alpha,
 \qquad 2\le|X_\alpha|\le k,\quad\diam X_\alpha\le r.
 \end{gathered}
 \label{eq:labelled-interactions}
\end{equation}
   Here $r\ge1$ and $k,\nu\in\NN$ are
fixed, and each site belongs to
at most $\nu$ of the sets $X_\alpha$, counting labels. Distinct labels
may have the same support. The operators  $V_i$    and $h_\alpha$ act on
$\HH_i$ and $\HH_{X_\alpha}$, respectively, and are tensored with the
identity on the remaining sites. The  on-site    operators may have
arbitrary finite norms. Distances and diameters are measured in the
underlying graph or in the lattice metric specified in the estimate;
for a lattice, $\Lambda\subset\ZZ^D$. The  interaction graph joins
   distinct  sites whenever
they    belong to a common
$X_\alpha$ ; its
maximum degree   $g$  satisfies
$g\le\nu(k-1)$.

 We say that the underlying geometry has \emph{polynomial volume growth}
if there are constants $C_{\rm vol}\ge1$ and $D\ge1$, independent of
$N$, such that
 
  \[
 |U^{\le u}|\le C_{\rm vol}|U|(1+u)^D,
 \qquad \varnothing\ne U\subseteq\Lambda,\quad u\ge0.
\] 
 It suffices to check this bound for $U=\{i\}$, since
$U^{\le u}=\bigcup_{i\in U}\{i\}^{\le u}$.

The interaction strength between sites and the Gibbs state are
\begin{equation}
 J_{\mathrm{int}}
 :=\sup_{i\in\Lambda}\sum_{\alpha:\,i\in X_\alpha}\|h_\alpha\|,
 \qquad
 \sigma_{\beta,H}=\sigma_\beta
 :=\frac{e^{-\beta H}}{\Tr e^{-\beta H}}.
 \label{eq:interaction-strength-gibbs}
\end{equation}
Thus $J_{\mathrm{int}}$ excludes the on-site terms $V_i$ .
  Similarly, for a region
$\Omega\subseteq\Lambda$, we define
\begin{equation}
 H_\Omega:=\sum_{i\in\Omega}V_i
          +\sum_{\substack{\alpha\in\mathcal A\\X_\alpha\subseteq\Omega}}
                  h_\alpha,
 \qquad
 \sigma_\beta^\Omega:=\frac{e^{-\beta H_\Omega}}{\Tr_{\HH_\Omega}e^{-\beta H_\Omega}}.
 \label{eq:restricted-hamiltonian}
\end{equation}
The restriction keeps exactly the labelled terms whose specified
supports are entirely contained in $\Omega$, with no added boundary
terms. In the definition of $\sigma_\beta^\Omega$, $H_\Omega$ acts on
$\HH_\Omega$. When used as an operator on $\HH_\Lambda$, it is
extended by the identity on $\HH_{\Lambda\setminus\Omega}$.

For a  block $B$   ,
the  decoupled Hamiltonian and crossing interaction
  are
\begin{equation}
 H_B^0=H_B+H_{B^c},\qquad
 W_{\partial B}=H-H_B^0
 =\sum_{\substack{\alpha\in\mathcal A\\
      X_\alpha\cap B\ne\varnothing,\ X_\alpha\cap B^c\ne\varnothing}}
       h_\alpha.
 \label{eq:crossing-interaction}
\end{equation}
   In
particular, $W_{\partial B}$ is supported on the union of these
crossing supports. We write
\[
 w_B=\|W_{\partial B}\|\le C_B,
 \qquad
 C_B:=\sum_{\substack{\alpha\in\mathcal A\\
      X_\alpha\cap B\ne\varnothing,\ X_\alpha\cap B^c\ne\varnothing}}
       \|h_\alpha\|.
\]

  For
nearest-neighbor
chains, the bond term $W_{i,i+1}$ is supported
on $\{i,i+1\}$ and satisfies $\|W_{i,i+1}\|\le J_{\mathrm{int}}$.
Every interval $B$ has at most two crossing bonds, so
$w_B\le C_B\le2J_{\mathrm{int}}$.
Unless stated otherwise, constants
  may depend on the fixed geometric
parameters (including $D$ on a lattice), $d,r,k,\nu,J_{\mathrm{int}}$
and $\beta$, but not on $N$ or the  norms of the on-site terms.

We set
$\Ad_D(X)=DXD^*$  for any operator $D$, without requiring it to
be unitary   , and
$D^{-*}=(D^{-1})^*$ when $D$ is invertible.
The symbols $\one$ and $\id$ denote an identity operator and an identity
map, respectively. Commutation with the Hamiltonian defines the linear map
\begin{equation}
\ad_H(A):=[H,A]=HA-AH.
\label{eq:hamiltonian-commutator}
\end{equation}
For complex $z$, we define conjugation by the Hamiltonian and its
real-time specialization by
\begin{equation}
\alpha_z^H(A)=e^{zH}Ae^{-zH},\qquad
 \tau_t^H(A)=\alpha_{it}^H(A)=e^{itH}Ae^{-itH}.
\label{eq:hamiltonian-conjugation}
\end{equation}
The Bohr frequencies of $H$ are the differences $\omega=E-E'$ between
its eigenvalues. Writing $P_E$ for the spectral projection associated
with $E$, define the component of $A$ at frequency $\omega$ by
\[
A_\omega=\sum_{E-E'=\omega}P_EAP_{E'}.
\]
Thus $\ad_H(A_\omega)=\omega A_\omega$. For any function $q$ defined on
the Bohr frequencies of $H$, the spectral functional calculus of
$\ad_H$ is given by
\begin{equation}
q(\ad_H)(A):=\sum_\omega q(\omega)A_\omega
 =\sum_{E,E'}q(E-E')P_EAP_{E'}.
\label{eq:bohr-functional-calculus}
\end{equation}
Here $E,E'$ range over the eigenvalues of $H$, and $\omega$ over their
distinct differences. Taking $q_z(\omega)=e^{z\omega}$ recovers the
conjugation maps introduced in \eqref{eq:hamiltonian-conjugation}:
\[
\alpha_z^H(A)=q_z(\ad_H)(A)=e^{z\ad_H}(A)
 =\sum_\omega e^{z\omega}A_\omega.
\]

\subsection{Channels, norms, and transport}
\label{sec:channels-transport}
We denote by $\norm\cdot$ and $\nOne\cdot$ the operator and trace norms. For a linear map
$\mathcal F:\BB(\HH_{\rm in})\to\BB(\HH_{\rm out})$ between finite-dimensional matrix
algebras, we denote by $\mathcal F^*$ its
Hilbert--Schmidt adjoint. Its diamond norm is
\[
\nD{\mathcal F}=\sup_{m\ge1}\sup_{X\ne0}
\frac{\nOne{(\mathcal F\otimes\id_m)(X)}}{\nOne X}.
\]
Here $m$ ranges over the positive integers, and $X$ over all nonzero operators in
$\BB(\HH_{\rm in}\otimes\CC^m)$, the matrix algebra of the input and an $m$-dimensional
reference
system.
Its completely bounded (cb) norm is defined using the
operator norm:
\[
 \|\mathcal F\|_{\mathrm{cb}}
 =\sup_{m\ge1}\sup_{X\ne0}
   \frac{\norm{(\mathcal F\otimes\id_m)(X)}}{\norm X}
 =\nD{\mathcal F^*}.
\]
The last equality follows from duality of the operator and trace norms.

We recall that a channel is a completely positive
trace-preserving (CPTP) map. Given a lattice $\Lambda$, by a channel supported on a
region $R\subset\Lambda$ we mean a channel on $\BB(\HH_R)$ tensored with the
identity on $R^c$.
A Lindbladian is a linear map $\mathcal L$ for which $\Phi_t=e^{t\mathcal L}$ is CPTP
for every
$t\ge0$. These maps form a continuous semigroup, with $\Phi_0=\id$ and
$\Phi_{t+s}=\Phi_t\circ\Phi_s$ for $s,t\ge0$, so that the evolution
$\rho_t=\Phi_t(\rho_0)$ satisfies
$\dot\rho_t=\mathcal L(\rho_t)$.

   We use  $\Phi$ and
  $\mathcal L$    for channels and generators acting
on
states. Their adjoints act on observables. In particular,  $\Phi^*$
 is
unital and satisfies
\[
  \Tr[O\Phi(X)]=\Tr[\Phi^*(O)X],\qquad \Phi^*(\one)=\one.
\]
The corresponding Heisenberg generator is $\mathcal L^*$, so
$(e^{t\mathcal L})^*=e^{t\mathcal L^*}$.
 For a faithful state $\sigma$, the KMS inner product on observables is
\[
 \langle A,C\rangle_\sigma
 =\Tr(A^*\sigma^{1/2}C\sigma^{1/2}).
\]
We call a generator KMS reversible when $\mathcal L^*$ is self-adjoint
for this inner product. Since $\mathcal L^*(\one)=0$, reversibility
implies $\Tr(\sigma\mathcal L^*(A))=0$ for every $A$, hence
$\mathcal L(\sigma)=0$.
 The following    two forms of Duhamel's identity will be used to estimate
the influence of a block update.

\begin{lemma}
  \label{lem:trajectories}
  Let $\cT(X)=\sum_a V_aXV_a^*$ for a finite family of operators $V_a$, with
  $\sum_a V_a^*V_a=\cT^*(\one)\ge s\one$, and let $G=G^*$. Define
  \[
    \cK(X)=\cT(X)-\tfrac12\sum_a\{V_a^*V_a,X\}-i[G,X],\qquad
    Z_t=e^{-t(\frac{1}{2}\sum_a V_a^*V_a+iG)}.
  \]
  Put $\cN_t(X)=Z_tXZ_t^*$ and $\Phi_t=e^{t\cK}$. Then $\norm{Z_t}^2\le e^{-st}$,
  $\Phi_t$ is a channel, and $\cJ_T=\Phi_T-\cN_T$ is CP, with the Duhamel identities
  \begin{equation}
    \cJ_T=\int_0^T\Phi_{T-t}\cT\cN_t\,dt
    =\int_0^T\cN_{T-t}\cT\Phi_t\,dt.
    \label{eq:duhamelidentities}
  \end{equation}
\end{lemma}
\begin{proof}
  The derivative of $\norm{Z_t\psi}^2$ is at most $-s\norm{Z_t\psi}^2$, so
  $\norm{Z_t}^2\le e^{-st}$ and $\nD{\cN_t}=\norm{Z_t}^2\le e^{-st}$.
  Since $\cK$ is in Lindblad form, $\Phi_t=e^{t\cK}$ is CPTP.
  The generator of $\cN_t$ is $\cK-\cT$, so the two forms of Duhamel's identity
  give \eqref{eq:duhamelidentities} directly. Each integrand is a composition of
  the completely positive maps $\Phi_t$, $\cT$, and $\cN_t$.
  Their integrals are therefore completely positive, proving that $\cJ_T$ is CP.
\end{proof}

\noindent The transport norm or Wasserstein $1$-norm is defined by decomposing a traceless
Hermitian operator into terms that
vanish when one site is traced out. For $X=X^*$ with $\Tr X=0$, define \cite{DPMTL}
\begin{equation}
  \nW X=\frac{1}{2}\inf\left\{\sum_{j\in\Lambda}\nOne{X_j}:
  X=\sum_jX_j,\quad X_j=X_j^*,\quad\Tr_jX_j=0\right\}.
  \label{eq:Wdef}
\end{equation}

The dual seminorm of the global norm \eqref{eq:Wdef} on Hermitian observables is
\cite[Propositions~8--9]{DPMTL}
\begin{equation}
  \Lip(O)=2\max_j\inf_{Y=Y^*}\norm{O-I_j\otimes Y}.
  \label{eq:Lip}
\end{equation}
With $\mathcal E_j(X)=I_j/d_j\otimes\Tr_jX$, where $d_j=\dim\HH_j$, Proposition~15 of
Appendix~B in
\cite{DPMTL} gives
\[
\max_j\norm{(\id-\mathcal E_j)O}\le\Lip(O)
\le2\max_j\norm{(\id-\mathcal E_j)O}.
\]
For a channel $\Psi$ on $\BB(\HH_\Lambda)$ and a site $j\in\Lambda$, define its
\emph{influence coefficient at $j$} by
\begin{align}\label{eq:influence}
  D_\Psi(j):=
  \sup_{\substack{X=X^*\ne0\\\Tr_jX=0}}
  \frac{\nW{\Psi(X)}}{\nOne X/2},
\end{align}
where the supremum runs over operators $X\in\BB(\HH_\Lambda)$.
We now collect some elementary properties of the Wasserstein distance.
\begin{lemma}[Elementary transport estimates]
  \label{lem:transport}
  Let $\varnothing\ne R\subseteq\Lambda$, with $N=|\Lambda|$ and $m=|R|$, and let
  $\nW\cdot$ denote the
  global transport norm \eqref{eq:Wdef}. For traceless Hermitian $X\in\BB(\HH_\Lambda)$,
  \begin{equation}
    \tfrac12\nOne X\le\nW X\le\tfrac N2\nOne X.
    \label{eq:Wcompare}
  \end{equation}
  If $\Tr_jX=0$ for some $j\in\Lambda$, then $\nW X=\nOne X/2$.
  If $Y=Y^*\in\BB(\HH_\Lambda)$, $\Tr_RY=0$, and $j_1,\ldots,j_m$ enumerate $R$, then
  \begin{equation}
    \nW Y\le\sum_{b=0}^{m-1}\nOne{\Tr_{j_1,\ldots,j_b}Y}
    \le |R|\nOne Y.
    \label{eq:telescoping}
  \end{equation}
  Here $\Tr_\varnothing=\id$, so the term for $b=0$ is $\nOne Y$.
  {More generally, for any traceless Hermitian $Y$ and any
  $S\subseteq\Lambda$,
  \begin{equation}
    \|Y\|_{W_1}\le |S|\,\|Y\|_1+
                     \|\Tr_SY\|_{W_1}.
    \label{qp:compact-reconstruction}
  \end{equation}
  The last norm is taken on $\Lambda\setminus S$ and is understood
  as zero when $S=\Lambda$.}
  A channel $\Psi$ supported away from $j$ satisfies $D_\Psi(j)=1$.
  If two channels $\Phi,\Psi$ are supported on $R$, then
  \begin{equation}
    \nW{(\Phi-\Psi)(X)}\le |R|\nD{\Phi-\Psi}\nOne X
    \quad(X=X^*).
    \label{eq:localW}
  \end{equation}
\end{lemma}
   \begin{proof}
     The norm comparison
\eqref{eq:Wcompare} and the   one-site identity
follow from
\cite[Proposition~2, equations~(10)--(11)]{DPMTL}
.
  We prove the remaining bounds directly on the full system ,
     keeping the exterior $R^c$ throughout.

  {For the telescoping estimates, let $Y$ be any traceless
  Hermitian operator. Set $R_b=\{j_1,\ldots,j_b\}$, with
  $R_0=\varnothing$, and define
  \[
    Y^{(0)}=Y,\qquad
    Y^{(b)}=\frac{\one_{R_b}}{\dim\HH_{R_b}}\otimes\Tr_{R_b}Y
    \quad(1\le b\le m).
  \]
  The Hermitian operators $Z_b=Y^{(b-1)}-Y^{(b)}$ satisfy
  \[
    Y=\sum_{b=1}^m Z_b+Y^{(m)},\qquad \Tr_{j_b}Z_b=0.
  \]
  Tensoring a transport decomposition of $\Tr_RY$ with the normalized
  identity on $R$ preserves its cost, so
  \[
    \nW{Y^{(m)}}\le\nW{\Tr_RY}.
  \]
  By \eqref{eq:Wdef} and trace-norm contractivity of partial trace,
  \begin{align*}
    \nW Y
    &\le\frac12\sum_{b=1}^m\nOne{Z_b}+\nW{\Tr_RY}\\
    &\le\frac12\sum_{b=1}^m
       \bigl(\nOne{\Tr_{R_{b-1}}Y}+\nOne{\Tr_{R_b}Y}\bigr)
       +\nW{\Tr_RY}\\
    &\le\sum_{b=0}^{m-1}\nOne{\Tr_{R_b}Y}+\nW{\Tr_RY}\\
    &\le |R|\nOne Y+\nW{\Tr_RY}.
  \end{align*}
  This proves \eqref{qp:compact-reconstruction} with $S=R$;
  the case $S=\varnothing$ is immediate. When $\Tr_RY=0$, the same
  calculation gives \eqref{eq:telescoping}.}

 If  $\Psi$  is  supported away
from $j$   , write
   $\Psi=\id_j\otimes\Psi_{j^c}$,  where
$j^c=\Lambda\setminus\{j\}$.
  For $X=X^*$  with   $\Tr_jX=0$,  we have
   $\Tr_j\Psi(X)=\Psi_{j^c}(\Tr_jX)=0$. The
one-site identity and
  trace-norm contractivity give

  \[
    \nW{\Psi(X)}=\tfrac12\nOne{\Psi(X)}\le\tfrac12\nOne X.
  \]
     Thus  $D_\Psi(j)\le1$.
Equality is attained by $X=A\otimes\rho$,
  where $A$ is a nonzero traceless Hermitian operator on $\HH_j$
  and $\rho$ is    a  state on
$\HH_{j^c}$:  since  $\Psi_{j^c}(\rho)$ is
  a state,   $\nOne{\Psi(X)}=\nOne A=\nOne X$.

  Finally,
  if $\Phi$ and $\Psi$ are channels supported on $R$,
  trace preservation gives
  \[
    \Tr_R\Phi(X)=\Tr_RX=\Tr_R\Psi(X).
  \]

  Therefore $Z=(\Phi-\Psi)(X)$
satisfies $\Tr_RZ=0$,  and    the
  telescoping bound yields
  \[
    \nW Z\le |R|\nOne Z
    \le |R|\nD{\Phi-\Psi}\nOne X,
  \]
  proving
\eqref{eq:localW}.
\end{proof}

\subsection{   Locality
estimates and Lieb--Robinson  bounds}
\label{sec:locality-kernels}
We recall   the
  locality estimate used in Section~\ref{qg:construction}
to control the effect of truncating the Hamiltonian in the QBP
construction. It compares the evolution of an observable supported
on $S$ under $H$ with its evolution under $H_\Omega$, defined in
\eqref{eq:restricted-hamiltonian}, when $\Omega$ contains a
neighbourhood of $S$.
We assume polynomial volume growth with constants $C_{\rm vol},D$
as defined in Section~\ref{sec:interactions}.

  \begin{theorem}[Real-time locality, {\cite[Theorem~3.4(ii)]{NSY}}]
\label{qg:lr}
Let $H$ satisfy \eqref{eq:labelled-interactions}, with interaction
strength $J_{\mathrm{int}}$ as in
 \eqref{eq:interaction-strength-gibbs}. There are constants
$C_{\rm LR},\mu_{\rm LR}>0$, depending only on
 $r,D,C_{\rm vol}$,
such that, with $v_{\rm LR}=C_{\rm LR}J_{\mathrm{int}}$,

  \begin{equation}
 \norm{\tau_t^H(O)-\tau_t^{H_\Omega}(O)}
 \le C_{\rm LR}|S|\norm O\,
             (e^{v_{\rm LR}|t|}-1)e^{-\mu_{\rm LR}u}
 \label{qg:lr-eq}
\end{equation}
for every operator $O$ supported on a nonempty set
$S\subseteq\Lambda$, every $t\in\RR$, $u\ge0$, and
$S^{\le u}\subseteq\Omega\subseteq\Lambda$
 .
The constants are independent of the norms of the on-site terms
$V_i$ and of the local Hilbert-space dimensions.
\end{theorem}

\section{Gibbs resets    and
 quantum belief propagation}
\label{qg:construction}
We consider a local Hamiltonian in the form of \eqref{eq:labelled-interactions}. In what
follows,
the local Gibbs state associated with $H$ on a region $B\subseteq \Lambda$ and its reset channel
are respectively denoted by:
\begin{equation}
 \sigma_\beta^B=\frac{e^{-\beta H_B}}{\Tr e^{-\beta H_B}},
 \qquad \text{ and }\qquad \mathcal R_B(X)=\sigma_\beta^B\otimes\Tr_B X.
 \label{qg:reset}
\end{equation}
   When the
interactions crossing the boundary of $B$ are removed,
$\mathcal R_B$ replaces the state on $B$ by the corresponding Gibbs
state. We use quantum belief propagation to modify this map so that
the resulting dynamics preserve the Gibbs state of the full Hamiltonian .

\subsection{KMS-symmetric maps through quantum belief propagation}

Quantum belief propagation (QBP)
 describes how  a Gibbs operator
 changes  under a perturbation of the Hamiltonian
  \cite{HastingsQBP,KimQBP,KatoBrandao}.
The    relation between the two Gibbs operators is given by
\eqref{qg:congruence-eq}, where $O_B$  solves a differential equation
involving real-time evolution. For reasons that will become clear in the following, we apply
this construction at inverse
temperature $\beta/2$.

   For $\beta>0$ and a self-adjoint operator $K$ on a
finite-dimensional
Hilbert space, define
\begin{equation}
 p_\beta(t)=\frac{4}{\beta\pi}\log\coth\frac{\pi|t|}{\beta},
 \qquad
 \mathcal F_K(A)=\int_{\mathbb R}p_\beta(t)\tau_t^K(A)\,dt,
 \label{qg:F}
\end{equation}
where the first formula is for $t\ne0$ and $A$ is any operator on that
Hilbert space. The integral defining $\mathcal F_K$
converges because
$p_\beta$ has an integrable logarithmic singularity at zero and decays
exponentially at infinity.    \begin{lemma}
\label{lem:qbp-kernels}
For $\beta>0$ and $K=K^*$, the kernel and map in
\eqref{qg:F}
satisfy
\begin{equation}
 p_\beta\ge0,\qquad \int_{\mathbb R}p_\beta(t)\,dt=1,\qquad
 \widehat p_\beta(\omega)
 =\frac{\tanh(\beta\omega/4)}{\beta\omega/4},
 \label{qg:F-kernel}
\end{equation}
with Fourier convention $\widehat p(\omega)=\int_{\mathbb R}p(t)e^{it\omega}\,dt$
and quotient equal to one at zero. For every $L\ge0$,
\begin{equation}
 \int_{|t|>L}p_\beta(t)\,dt\le e^{-2\pi L/\beta},\qquad
 \int_{\mathbb R}|t|p_\beta(t)\,dt
 =\frac{4\beta}{\pi^3}\sum_{n\ge0}(2n+1)^{-3}\le\beta.
 \label{eq:qbp-kernel-tail}
\end{equation}
Moreover, $\|\mathcal F_K(A)\|\le\|A\|$ for every operator $A$.
\end{lemma}
\begin{proof}
For $t\ne0$, the power series for the logarithm gives
\[
 p_\beta(t)=\frac{8}{\beta\pi}\sum_{n=0}^\infty
       \frac{e^{-2\pi(2n+1)|t|/\beta}}{2n+1}.
\]
All summands are positive. Termwise integration gives
\[
 \int_{|t|>L}p_\beta(t)\,dt
 =\frac{8}{\pi^2}\sum_{n\ge0}
       \frac{e^{-2\pi(2n+1)L/\beta}}{(2n+1)^2}
 \le e^{-2\pi L/\beta}.
\]
At $L=0$, the equality and $\sum_{n\ge0}(2n+1)^{-2}=\pi^2/8$
give unit mass. Integrating with a factor $|t|$ gives the first
moment in
\eqref{eq:qbp-kernel-tail}. Fourier transformation yields
\[
 \widehat p_\beta(\omega)
 =\frac{8}{\pi^2}\sum_{n\ge0}
   \frac{1}{(2n+1)^2+(\beta\omega/(2\pi))^2}
 =\frac{\tanh(\beta\omega/4)}{\beta\omega/4},
\]
where we used the partial-fraction identity
$\sum_{n\ge0}((2n+1)^2+z^2)^{-1}
=\pi\tanh(\pi z/2)/(4z)$, extended continuously at zero.
Finally, unitary invariance of the operator norm gives
\[
 \|\mathcal F_K(A)\|
 \le\int_{\mathbb R}p_\beta(t)\|\tau_t^K(A)\|\,dt
 =\|A\|.
\]
\end{proof}

\noindent Next, let $O_B(u)$ be the solution of
\begin{equation}
 O_B'(u)=-\frac{\beta}{4}\mathcal F_{H_B(u)}(W_{\partial B})O_B(u),
 \qquad O_B(0)=\one,
 \qquad O_B=O_B(1).
 \label{qg:ODE}
\end{equation}

\begin{proposition}
\label{qg:congruence}
The    operator $O_B$ defined in
\eqref{qg:ODE}
 satisfies
\begin{align}
 O_B e^{-\beta H_B^0/2}O_B^*&=e^{-\beta H/2},\label{qg:congruence-eq}\\
 \|O_B^{\pm1}\|&\le e^{\beta\|W_{\partial B}\|/4},\label{qg:O-norm}\\
 \|O_B-\one\|+\|O_B^{-1}-\one\|&\le2\bigl(e^{\beta\|W_{\partial B}\|/4}-1\bigr).
 \label{qg:O-identity}
\end{align}
   More generally, for every $0\le u\le v\le1$,
\begin{equation}
 \max\left\{\|O_B(v)O_B(u)^{-1}\|,
             \|O_B(u)O_B(v)^{-1}\|\right\}
 \le e^{(v-u)\beta\|W_{\partial B}\|/4}.
 \label{qg:O-interval}
\end{equation}

\end{proposition}
\begin{proof}
Set $Q(u)=e^{-\beta H_B(u)/2}$. In an eigenbasis of $H_B(u)$,
its Fr\'echet derivative has entries
\[
 Q'(u)_{ab}=
 \frac{e^{-\beta E_a/2}-e^{-\beta E_b/2}}{E_a-E_b}
 (W_{\partial B})_{ab},
\]

 where the divided difference is interpreted by continuity when
the two eigenvalues coincide . By
\eqref{qg:F-kernel}, this is precisely
\[
 Q'(u)=-\frac{\beta}{4}
 \{\mathcal F_{H_B(u)}(W_{\partial B}),Q(u)\}.
\]
The    operator
$O_B(u)Q(0)O_B(u)^*$    satisfies
  the same linear matrix
differential equation with the same initial value. Uniqueness proves
\eqref{qg:congruence-eq}.    By
\Cref{lem:qbp-kernels},
\[
 \left\|\frac{\beta}{4}\mathcal F_{H_B(u)}(W_{\partial B})\right\|
 \le\frac{\beta}{4}\|W_{\partial B}\|.
\]
The inverse satisfies
\[
 \frac{d}{du}O_B(u)^{-1}
 =\frac{\beta}{4}O_B(u)^{-1}\mathcal F_{H_B(u)}(W_{\partial B}).
\]
We apply  Gr\"onwall's inequality to
\eqref{qg:ODE}
and this equation,
starting at interpolation parameter $u$, to obtain
\eqref{qg:O-interval}.
Taking $u=0$, $v=1$ proves
\eqref{qg:O-norm}. Integrating from zero
to one also gives, for each choice of sign,
\[
 \|O_B^{\pm1}-\one\|
 \le\int_0^1\frac{\beta\|W_{\partial B}\|}{4}
              e^{v\beta\|W_{\partial B}\|/4}\,dv
 =e^{\beta\|W_{\partial B}\|/4}-1,
\]
which  proves \eqref{qg:O-identity}.
\end{proof}

\noindent For $K=K^*$, define
\[
 \mathcal G_K(A)=\frac{i}{2}
 \tanh\left(\frac{\beta}{4}\ad_K\right)(A).
\]

 \begin{lemma}
\label{lem:qbp-coherent-transform}
For $\beta>0$, $K=K^*$ and every operator $A$,
\begin{align}
 \mathcal G_K(A)
 &=\frac{i\beta}{8}\int_{\mathbb R}p_\beta(t)
                         \tau_t^K([K,A])\,dt,
 \label{eq:qbp-G-source}\\
 &=\frac{1}{\beta}\int_0^\infty
 \frac{\tau_t^K(A)-\tau_{-t}^K(A)}{\sinh(2\pi t/\beta)}\,dt.
 \label{qg:G-integral}
\end{align}
Both integrals converge in operator norm, and
$\|\mathcal G_K(A)\|\le\beta\|[K,A]\|/8$.
\end{lemma}
\begin{proof}
Applying
\eqref{qg:F-kernel}
to each Bohr component of $A$ gives
\[
 \mathcal G_K(A)=\frac{i\beta}{8}\mathcal F_K([K,A]),
\]
which proves
\eqref{eq:qbp-G-source}. For
\eqref{qg:G-integral}, expand
$1/\sinh(2\pi t/\beta)=2\sum_{n\ge0}e^{-2\pi(2n+1)t/\beta}$
for $t>0$ and integrate the sine to obtain
\[
 \int_0^\infty\frac{\sin(\omega t)}{\sinh(2\pi t/\beta)}\,dt
 =\frac{\beta}{4}\tanh(\beta\omega/4).
\]
Applying this identity to each Bohr component proves
\eqref{qg:G-integral}.
The first integral converges because $p_\beta$ is integrable. In the
second, $\tau_t^K(A)-\tau_{-t}^K(A)=O(t)$ at zero, and the denominator
grows exponentially at infinity. Finally,
\Cref{lem:qbp-kernels}
gives
\[
 \|\mathcal G_K(A)\|
 =\frac{\beta}{8}\|\mathcal F_K([K,A])\|
 \le\frac{\beta}{8}\|[K,A]\|.
\]
\end{proof}

\noindent Next, we choose a Hermitian Hilbert--Schmidt orthonormal basis
$(F_{B,a})_{a=1}^{(\dim\HH_B)^2}$ of $\BB(\HH_B)$ , normalized by
$\Tr(F_{B,a}F_{B,b})=\delta_{ab}$.  Completeness of this basis gives,
for every operator $X$ on $B$ and any exterior reference,
\begin{align}
 \sum_{a=1}^{(\dim\HH_B)^2}F_{B,a}XF_{B,a}
 =\one_B\otimes\Tr_BX~ \Rightarrow~  \sum_{a=1}^{(\dim\HH_B)^2}
  (\sigma_\beta^B)^{1/2}F_{B,a}XF_{B,a}(\sigma_\beta^B)^{1/2}
 =\mathcal R_B(X).
 \label{eq:qbp-basis-completeness}
\end{align}
The first identity follows for matrix units, and the sum is unchanged
by an orthonormal change of operator basis.

   Using  the KMS
completion \cite{AmorimCarlen,DLL},    we define the jump operators and block generator by

\begin{align}
 L_{B,a}&=O_B(\sigma_\beta^B)^{1/2}F_{B,a}O_B^{-1},\notag\\
 \mathcal T_B(X)&=\sum_aL_{B,a}XL_{B,a}^*
 =\bigl(\operatorname{Ad}_{O_B}\circ\mathcal R_B
 \circ\operatorname{Ad}_{O_B^{-1}}\bigr)(X),\notag\\
 S_B&=\mathcal T_B^*(\one)=\sum_aL_{B,a}^*L_{B,a},\notag\\
 G_B&=\mathcal G_H(S_B),\notag\\
 \mathcal K_B(X)&=\mathcal T_B(X)-\tfrac12\{S_B,X\}-i[G_B,X],
 \qquad \Phi_B(T)=e^{T\mathcal K_B}.
 \label{qg:generator}
\end{align}

The map $\operatorname{Ad}_D$ is completely positive for every $D$,
but need not preserve the trace. The same is true of $\mathcal T_B$.
The anticommutator and coherent correction in \eqref{qg:generator}
give its KMS-symmetric Lindblad completion~\cite{AmorimCarlen,DLL}.
As proved in \Cref{qg:KMS}, each jump satisfies
$L_{B,a}e^{-\beta H/2}=e^{-\beta H/2}L_{B,a}^*$.
Equivalently,
\[
  A_{B,a}:=e^{\beta H/4}L_{B,a}e^{-\beta H/4}=A_{B,a}^*,
  \qquad
  L_{B,a}=e^{-\beta H/4}A_{B,a}e^{\beta H/4}.
\]
This is the canonical representation of~\cite{DLL}. However, here the
Hermitian operators $A_{B,a}$ depend on the QBP construction,
rather than being fixed local proposals. Importantly, our locality and mixing estimates use the particular structure of QBP-dressed block Gibbs resets.

\begin{proposition}[Gibbs preservation and KMS detailed balance]
\label{qg:KMS}
The map $\mathcal K_B$ in \eqref{qg:generator} is a Lindblad generator.
Its associated quantum dynamics preserves $\sigma_\beta=e^{-\beta H}/\Tr(e^{-\beta H})$ and is
KMS
reversible: its Heisenberg adjoint is self-adjoint for
\[
 \langle A,C\rangle_{\sigma_\beta}
 =\Tr\bigl(A^*\sigma_\beta^{1/2}C\sigma_\beta^{1/2}\bigr).
\]
Every $\Phi_B(T)$ is consequently a Gibbs-preserving, KMS-reversible
channel.
\end{proposition}
   \begin{proof}
Put $Q=e^{-\beta H/2}$ and $Q_0=e^{-\beta H_B^0/2}$.
For $L_{B,a}^0=(\sigma_\beta^B)^{1/2}F_{B,a}$, Hermiticity of the
basis elements gives $L_{B,a}^0Q_0=Q_0(L_{B,a}^0)^*$.
The    identity
$Q=O_BQ_0O_B^*$ therefore implies\[
 L_{B,a}Q=O_BL_{B,a}^0Q_0O_B^*
       =O_BQ_0(L_{B,a}^0)^*O_B^*=QL_{B,a}^*.
\]
Let $\Gamma_Q(X)=QXQ$ and $D_B=-S_B/2+iG_B$.
In an eigenbasis of $Q$, with eigenvalues $q_a>0$, the definition of
$G_B$ gives
\[
 (G_B)_{ab}=\frac{i}{2}\frac{q_b-q_a}{q_a+q_b}(S_B)_{ab},\qquad
 (D_B)_{ab}=-\frac{q_b}{q_a+q_b}(S_B)_{ab},\qquad
 QD_B=D_B^*Q.
\]
Since $\mathcal K_B(X)=\mathcal T_B(X)+D_B^*X+XD_B$, the preceding
intertwining relations yield
\[
 \Gamma_Q^{-1}\mathcal T_B\Gamma_Q=\mathcal T_B^*,\qquad
 \Gamma_Q^{-1}\mathcal K_B\Gamma_Q=\mathcal K_B^*.
\]
This is KMS self-adjointness of $\mathcal K_B^*$.
The displayed expression \eqref{qg:generator} is in Lindblad form,
so $\mathcal K_B^*(\one)=0$ and its semigroup is CPTP.
Consequently $\mathcal K_B(Q^2)=\Gamma_Q\mathcal K_B^*(\one)=0$.
Normalizing $Q^2$ proves Gibbs preservation. Exponentiation preserves
KMS self-adjointness, proving the assertion for $\Phi_B(T)$.
     \end{proof}

  \begin{lemma}
 \label{qg:loss}
The above construction satisfies
\begin{equation}
 e^{-\beta w_B}\one\le S_B\le e^{\beta w_B}\one,
 \qquad \|\mathcal T_B\|_\diamond=\|S_B\|\le e^{\beta w_B}.
 \label{qg:loss-eq}
\end{equation}
These estimates    also hold when
\eqref{qg:ODE}
is solved with a restricted Hamiltonian, keeping  the same boundary interaction and
  reset channel.
\end{lemma}
\begin{proof}
The adjoint reset is positive and unital, and
\[
 S_B=O_B^{-*}\mathcal R_B^*(O_B^*O_B)O_B^{-1}.
\]
Since $e^{-\beta w_B/2}\one\le O_B^*O_B\le e^{\beta w_B/2}\one$, positivity and a second
conjugation give the order bounds. The diamond norm of a completely
positive map equals the operator norm of its adjoint on the identity.
\end{proof}

\begin{remark}[Dependence on the boundary interaction]
The    bounds on $O_B$ and
$O_B^{-1}$  depend on the block only through its crossing
interaction $\|W_{\partial B}\|$. For an interval in a nearest-neighbor chain,
$\|W_{\partial B}\|\le2J_{\mathrm{int}}$   .
Thus the bounds on    $O_B^{\pm1}$ and  $S_B$ are independent of
the interval length. Increasing the
block length can therefore improve contraction in the interior
without increasing these boundary bounds. This observation is at the core of our
proof of rapid mixing for one-dimensional systems
at every finite temperature (\Cref{q1:main}).
\end{remark}

\begin{remark}
\label{qg:BLMT}
\label{rem:classicalrates}
When $H=H_B+H_{B^c}$, we have $O_B=\one$, $S_B=\one$ and
$G_B=0$. Since $\mathcal R_B^2=\mathcal R_B$,
\[
 \mathcal K_B=\mathcal R_B-\id,\qquad
 \Phi_B(T)=e^{-T}\id+(1-e^{-T})\mathcal R_B.
\]
Thus the state on $B$ is replaced by $\sigma_\beta^B$ at rate one,
removing correlations with $B^c$ while leaving its marginal unchanged.
When interactions cross the boundary, the conjugations in
$\mathcal T_B$ modify this replacement, and the completed generator
preserves the full Gibbs state.

Next, assume that all interaction terms
are diagonal in a common product
basis. For $U\subseteq\Lambda$, write
$\Omega_U=\prod_{i\in U}[d_i]$, where $[d_i]=\{1,\ldots,d_i\}$.
A configuration is $(u,\xi)\in\Omega_B\times\Omega_{B^c}$, with
associated projection $P_{u,\xi}=|u,\xi\rangle\langle u,\xi|$.
Write the diagonal entries of the Hamiltonian as
\[
 H(u,\xi)=H_B(u)+W_\xi(u)+H_{B^c}(\xi),\qquad
 W_\xi(u)=\langle u,\xi|W_{\partial B}|u,\xi\rangle.
\]
The full Gibbs distribution is
$\pi(u,\xi)=e^{-\beta H(u,\xi)}/\Tr e^{-\beta H}$.
The block Gibbs law and the conditional Gibbs law are, respectively,
\[
 \sigma_\beta^B(u)=\langle u|\sigma_\beta^B|u\rangle
 =\frac{e^{-\beta H_B(u)}}{\sum_{v\in\Omega_B}e^{-\beta H_B(v)}},
 \qquad
 \pi_B^\xi(u)
 =\frac{\sigma_\beta^B(u)e^{-\beta W_\xi(u)}}{\sum_{v\in\Omega_B}\sigma_\beta^B(v)e^{-\beta
W_\xi(v)}}.
\]
In particular, $\pi_B^\xi$ includes the interactions with the fixed
exterior configuration $\xi$, whereas $\sigma_\beta^B$ does not.

 Commutativity gives
$\mathcal F_{H_B(t)}(W_{\partial B})=W_{\partial B}$, so
\eqref{qg:ODE}
yields
\[
 O_B=e^{-\beta W_{\partial B}/4}
     =e^{-\beta H/4}e^{\beta H_B^0/4}.
\]
The completeness identity
\eqref{eq:qbp-basis-completeness}
then
shows that $\mathcal T_B$ preserves the diagonal algebra and that
$S_B$ is diagonal. Hence $[H,S_B]=0$ and $G_B=0$.
For distinct $u,v\in\Omega_B$, the transition rate is
\begin{align}
 r_\xi(u,v)
 &:=\Tr[P_{v,\xi}\mathcal K_B(P_{u,\xi})]
   =\sum_a|\langle v,\xi|L_{B,a}|u,\xi\rangle|^2\notag\\
 &=\sigma_\beta^B(v)e^{\beta[W_\xi(u)-W_\xi(v)]/2}
   =\sqrt{\sigma_\beta^B(u)\sigma_\beta^B(v)}
     \sqrt{\frac{\pi_B^\xi(v)}{\pi_B^\xi(u)}}.
 \label{eq:classicalrates}
\end{align}
Only the jump term contributes to this off-diagonal matrix entry.
For the diagonal entry, its term with $v=u$ cancels against the
corresponding anticommutator term, leaving
$-\sum_{v\ne u}r_\xi(u,v)$. Transitions changing $\xi$ have zero rate. Thus the classical
restriction, acting on observables $f:\Omega_\Lambda\to\mathbb R$, is
\begin{equation}
 (K_Bf)(u,\xi)=\sum_{v\ne u}r_\xi(u,v)
                   [f(v,\xi)-f(u,\xi)].
 \label{qg:classical-generator}
\end{equation}
These formulas hold for every Hermitian orthonormal basis chosen
in
\eqref{eq:qbp-basis-completeness}.
For a diagonal state $\rho_p=\sum_{u,\xi}p(u,\xi)P_{u,\xi}$, set
$R_B(T)=e^{TK_B}$, with rows indexed by the initial configuration.
Then
\begin{equation}
 \Phi_B(T)(\rho_p)
 =\sum_{v,\xi}(pR_B(T))(v,\xi)P_{v,\xi}.
 \label{qg:classical-channel}
\end{equation}
Detailed balance follows directly from
\[
 \pi_B^\xi(u)r_\xi(u,v)
 =\sqrt{\sigma_\beta^B(u)\sigma_\beta^B(v)
         \pi_B^\xi(u)\pi_B^\xi(v)}
 =\pi_B^\xi(v)r_\xi(v,u).
\]
Since the exterior configuration is unchanged, these rates also
satisfy detailed balance for the full Gibbs distribution $\pi$.

For comparison, Petz recovery after deleting $B$ is the map
\cite{Petz1986}
\begin{equation}
 \mathcal P_B(X)=\sigma_{\beta,H}^{1/2}
 \left[I_B\otimes\sigma_{\beta,B^c}^{-1/2}
       (\Tr_BX)\sigma_{\beta,B^c}^{-1/2}\right]
 \sigma_{\beta,H}^{1/2},\qquad
 \sigma_{\beta,B^c}=\Tr_B\sigma_{\beta,H}.
 \label{eq:Petz}
\end{equation}
This formula is defined for arbitrary Hamiltonians. In the present
classical setting, its action on a diagonal state is
\[
 \mathcal P_B(\rho_p)
 =\sum_{u,\xi}\pi_B^\xi(u)
       \left(\sum_{v\in\Omega_B}p(v,\xi)\right)P_{u,\xi}.
\]
It therefore induces the conditional heat-bath kernel
\begin{equation}
 \mathsf H_B((u,\xi),(v,\xi'))
 =\one_{\{\xi=\xi'\}}\pi_B^\xi(v),\qquad
 \pi\mathsf H_B=\pi,\qquad \mathsf H_B^2=\mathsf H_B.
 \label{eq:fulltargetheatbath}
\end{equation}
On diagonal observables, this is the Gibbs conditional expectation
$\mathsf H_Bf(u,\xi)=\sum_v\pi_B^\xi(v)f(v,\xi)$.
The classical restriction of $\mathcal P_B-\id$ has rates
$r_\xi^{\rm hb}(u,v)=\pi_B^\xi(v)$ for $u\ne v$, and
\begin{equation}
 r_\xi(u,v)
 =\sqrt{\frac{\sigma_\beta^B(u)\sigma_\beta^B(v)}{\pi_B^\xi(u)\pi_B^\xi(v)}}\,
    r_\xi^{\rm hb}(u,v).
 \label{qg:heatbath-rate-comparison}
\end{equation}
The heat-bath rate depends only on the destination $v$. The QBP
rate contains an additional factor symmetric in $u,v$ and generally
depends on the starting configuration. When $W_{\partial B}=0$,
$\pi_B^\xi=\sigma_\beta^B$, and the two generators coincide.

For each fixed $\xi$, all rates between distinct block
configurations are positive. The unique stationary law is therefore
$\pi_B^\xi$, and the finite-state Markov chain converges to it. Thus,
for every diagonal state $\rho_p$,
\[
 \lim_{T\to\infty}R_B(T)=\mathsf H_B,\qquad
 \lim_{T\to\infty}\Phi_B(T)(\rho_p)=\mathcal P_B(\rho_p).
\]
At finite $T$, $\Phi_B(T)$ is generally neither the heat-bath channel
nor a scalar convex combination of that channel and the identity.
Section~\ref{sec:classicalDSbridge}
quantifies the convergence and
shows that sufficiently large finite $T$ preserves a strict
heat-bath DS condition (
\Cref{thm:classical}
).

For noncommuting interactions, the same structure controls interior influences. After localizing
the boundary operators, they act trivially on an interior site $j$;
when $\Tr_jX=0$, the partial trace in $\mathcal R_B$ therefore makes
the jump contribution in Duhamel's formula vanish. For a
nearest-neighbor chain, \Cref{q1:influence} uses this cancellation
to bound the interior influence of the localized dynamics by
$\exp[-T e^{-2\beta J_{\mathrm{int}}}]$, independently of block length.
Passing to the full dynamics adds the locality error controlled in
\Cref{qg:chain-transport}.

This motivates interpreting $\Phi_B(T)$ as a pseudo-conditional
expectation: it preserves the Gibbs state and approximately removes
interior influences while its action outside the block is controlled
by locality. The
finite-time approximation to classical heat bath
described above explains the terminology; keeping  $T$
   finite also allows the locality estimates of
\Cref{qg:chain-locality}
to be used for noncommuting interactions .

\end{remark}

\begin{remark}[Commuting reduction and the BLMT normalization]
\label{rem:commuting-BLMT}
If $[H_B^0,W_{\partial B}]=0$, then
$\mathcal F_{H_B(u)}(W_{\partial B})=W_{\partial B}$ for every
$u\in[0,1]$. Solving \eqref{qg:ODE} gives
\begin{equation}
  O_B=e^{-\beta W_{\partial B}/4}
     =e^{-\beta H/4}e^{\beta H_B^0/4}.
  \label{qg:commuting}
\end{equation}
This condition holds when all Hamiltonian terms, including the
on-site terms, mutually commute.

For a qubit site $B=\{i\}$ with $V_i=0$, choose the normalized
Pauli basis $F_{i,a}=P_a/\sqrt2$, where
$P_a\in\{\one,\sigma_x,\sigma_y,\sigma_z\}$.
Then $\sigma_\beta^{\{i\}}=\one/2$ and
$H_{\{i\}}^0=H_{\{i\}^c}$ acts trivially on $i$.
Under the preceding commuting assumption, \eqref{qg:generator}
therefore gives
\begin{equation}
  L_{i,a}:=L_{\{i\},a}
  =\frac12 e^{-\beta H/4}P_a e^{\beta H/4}.
  \label{qg:BLMT-jumps}
\end{equation}
Let $\mathcal L_i^{\mathrm{BLMT}}$ denote the single-site generator
of \cite[Definition~4.1]{BLMT}, written in the Schr\"odinger
picture, and put $\mathcal K_i^{\mathrm{QBP}}:=\mathcal K_{\{i\}}$.
The identity jump $\one/2$ has zero dissipator, and its contribution
$\one/4$ to the loss has zero coherent correction.
Every other dissipator is one quarter of the corresponding BLMT
term. Since $\mathcal G_H$ is linear, the coherent correction has
the same factor. Consequently,
\begin{equation}
  \mathcal K_i^{\mathrm{QBP}}
  =\frac14\mathcal L_i^{\mathrm{BLMT}},
  \qquad
  e^{T\mathcal K_i^{\mathrm{QBP}}}
  =e^{(T/4)\mathcal L_i^{\mathrm{BLMT}}}.
  \label{qg:BLMT-factor}
\end{equation}
With nonzero jointly commuting on-site fields,
\eqref{qg:commuting} still holds, but the local Gibbs weight
$(\sigma_\beta^{\{i\}})^{1/2}$ remains in the jumps in
\eqref{qg:generator}.
\end{remark}

\subsection{Real-time locality and restriction estimates}

   We use the  real-time
  estimate of
\Cref{qg:lr}
to compare the QBP
operators with operators constructed from a Hamiltonian on a smaller
region. We first bound the coherent transform
on a finite region, then estimate the error made by restricting its
Hamiltonian.

 \begin{lemma}[A coherent multiplier bound]
\label{qg:multiplier}
For a self-adjoint matrix $K$ on a $D$-dimensional Hilbert space and
any real $c$,
\begin{equation}
 \left\|\frac{i}{2}\tanh(c\operatorname{ad}_K)\right\|_{\mathrm{cb}}
 \le4+\log_2D.
 \label{qg:multiplier-eq}
\end{equation}
In particular, the bound on a region $R$ of local dimension at most
$d$ is $4+|R|\log_2d$.
\end{lemma}
\begin{proof}
The Schur multiplier retaining the lower triangular part of an
$n\times n$ matrix has completely bounded norm at most
$1+\lceil\log_2n\rceil$. For $n=2^m$, split into two equal diagonal
blocks and one lower off-diagonal block. The direct sum of the
diagonal-block maps has the previous norm bound, while the
 off-diagonal compression has norm one. Induction gives $m+1$;
padding to a power of two proves the stated bound. This argument
works unchanged with matrix-valued entries.

For any lists $(a_i),(b_j)$ of length $D$, jointly order their $2D$
entries. The rectangular mask $\mathbf1_{\{a_i\ge b_j\}}$ is a
compression of a triangular mask on that ordered list, with ties
ordered consistently. It therefore has completely bounded norm at
most $2+\lceil\log_2D\rceil$. Diagonalize $K$ and put $a_i=cE_i$
and $b_j=cE_j+s$. The identity
\[
 \tanh z=\int_{\mathbb R}\operatorname{sign}(z-s)
                       \tfrac12\operatorname{sech}^2s\,ds
\]
expresses the desired Schur multiplier as an average of twice those
threshold masks minus the identity Schur multiplier. Multiplication
by $i/2$ gives a bound
$5/2+\lceil\log_2D\rceil\le4+\log_2D$. The finitely many threshold
ties have measure zero in this integral and do not affect the result.
\end{proof}

\begin{lemma}
\label{qg:transform-locality}

 For a nearest-neighbour chain, there are universal constants
$c,C>0$ such that the following holds with
$\lambda=c/(1+\beta J_{\mathrm{int}})$.
If $A$ is supported on a
nonempty set $S$ and $S^{\le q}\subseteq\Omega\subseteq\Lambda$,
with $q\ge0$, then
\begin{align}
 \|\mathcal F_H(A)-\mathcal F_{H_\Omega}(A)\|
 &\le C|S|\|A\|e^{-\lambda q},\label{qg:allT-F}\\
 \|\mathcal G_H(A)-\mathcal G_{H_\Omega}(A)\|
 &\le C\log(2+\beta J_{\mathrm{int}})|S|\|A\|e^{-\lambda q}.
 \label{qg:allT-G}
\end{align}
   The same bounds hold for $H_{\Omega'}$ and $H_\Omega$ when
$S\subseteq\Omega\subseteq\Omega'\subseteq\Lambda$ and
$S^{\le q}\cap\Omega'\subseteq\Omega$ .
\end{lemma}
\begin{proof}

  By
\Cref{lem:qbp-kernels},
\[
 \int_{|t|>s}p_\beta(t)\,dt\le e^{-2\pi s/\beta}.
\]

 For $q\ge1$   , split the integrals at
$t_* = \mu_{\rm LR}q/(2v_{\rm LR})$ .
Before $t_*$,    integrating
\eqref{qg:lr-eq}
against the probability
density $p_\beta$  costs at most    $Ce^{-\mu_{\rm LR}q/2}$.
After  $t_*$
 , use the kernel tail and unitary norm preservation.
Since $\beta v_{\rm LR}=C_{\rm LR}\beta J_{\mathrm{int}}$, both  contributions are at
most    $Ce^{-cq/(1+\beta J_{\mathrm{int}})}$, up to the
common factor $|S|\norm A$ .
For $q=0$ ,  the norm-one property of each $\mathcal F$ suffices.

For $\mathcal G$,    use
\eqref{qg:lr-eq}
and
$[\beta\sinh(2\pi t/\beta)]^{-1}\le(2\pi t)^{-1}$.
The  short-time contribution
  , with the common factor omitted, is at most
\[
 Ce^{-\mu_{\rm LR}q}
 \int_0^{\mu_{\rm LR}q/2}\frac{e^z-1}{z}\,dz
 \le Ce^{-\mu_{\rm LR}q/2}.
\]
The long-time contribution is at most a constant times
$\log\coth(\pi t_*/\beta)$. If $q/(\beta J_{\mathrm{int}})$ is bounded below,    its
exponential tail  bounds it by
 $Ce^{-cq/(1+\beta J_{\mathrm{int}})}$ . Otherwise it is at most
$C\log(2+\beta J_{\mathrm{int}})$, and $q/(1+\beta J_{\mathrm{int}})$ is bounded
above. Decreasing $c$ combines
these bounds into
$C\log(2+\beta J_{\mathrm{int}})e^{-cq/(1+\beta J_{\mathrm{int}})}$ . At $q=0$   ,
split instead at
$t=1/v_{\rm LR}$ ; the same small-time
 bound  and long-time integral
give $C\log(2+\beta J_{\mathrm{int}})$.    Restoring $|S|\norm A$  and using the
reference-system version of \Cref{qg:lr} proves the stated bounds.
\end{proof}

 \subsection{Implementation on graphs of polynomial volume growth}
\label{qi:section}

 We next consider the implementation of the smoothed heat-bath dynamics for
a prescribed evolution time. For a labelled family $\mathscr B$
of $M$ blocks, update rates $a_B>0$, and prescribed block durations
$0<T_B\le T$ with a fixed upper bound $T>0$, set
\begin{equation}
 \Phi_B=e^{T_B\mathcal K_B},\qquad
 \mathcal L=\sum_{B\in\mathscr B}a_B(\Phi_B-\id).
 \label{qi:dynamics}
\end{equation}
Here $\mathcal K_B$ is defined in 
\eqref{qg:generator}. The time $t$
below is the evolution time of $\mathcal L$, while $T_B$ is the duration
of the update on $B$. The earlier common-time construction is the
special case $T_B=T$ for every block.

  \begin{theorem}[Implementation of the smoothed heat-bath dynamics]
 \label{qi:theorem}
Let $H$ be a finite-range Hamiltonian with uniformly bounded local
terms on a graph of $N$ sites with polynomial volume growth and
bounded local dimension. Let $\mathscr B$  
 consist  of $M=O(N)$
 nonempty  connected blocks, with $|B|\le b$ for every
$B\in\mathscr B$, where $b$ is independent of $N$   , and let
the
positive rates $a_B$ be uniformly bounded. Fix $\beta,T>0$ and prescribe
$0<T_B\le T$ for each block.
Given the interaction terms, rates, block durations, and
construction parameters as
in Appendix~\ref{app:implementation:model}, for every  $t\ge0$ and
 $0<\varepsilon<1$ there is a circuit  
  whose channel
$\widehat{\mathcal E}_t$ satisfies
$\|\widehat{\mathcal E}_t-e^{t\mathcal L}\|_\diamond\le\varepsilon$
and uses at most
\begin{equation}
 C' N(1+t)\left[1+\log\frac{2N(1+t)}{\varepsilon}\right]^a
 \label{qi:complexity}
\end{equation} 
elementary gates and classical bit operations.
The constants $C',a$ depend only on  
 $\beta,T$ , the local dimension,
the bounds on interaction norms and range,  the bound on the update
rates,  and the graph and block geometry; they are independent of
$N,t,\varepsilon$.
For the chain construction, the circuit-cost prefactor grows
 doubly exponentially with $\beta J_{\mathrm{int}}$.
\end{theorem}

\noindent
The complexity bound is uniform over all prescribed durations
$0<T_B\le T$. The localization error for block $B$ is proportional to
$T_B$ and hence bounded using $T$; its simulation cost is bounded in
the same way. Thus allowing different block durations changes neither
\eqref{qi:complexity} nor its dependence on system size.
The proof is given in Appendix~\ref{app:implementation}. In a nutshell, we express
$e^{t\mathcal L}$ as an average over a Poisson-distributed number
of block updates. Lieb--Robinson bounds localize each update to a
neighbourhood of logarithmic radius; encodings of its jump operators
and coherent correction then allow its implementation by the
Lindblad simulation algorithm of~\cite{LW}. 

\begin{remark}[Temperature dependence in one dimension]
\label{qi:temperature-remark}
We believe that the doubly exponential temperature dependence of the
one-dimensional circuit-cost prefactor in \Cref{qi:theorem} can be
reduced to a single exponential in a polynomial of
$\beta J_{\mathrm{int}}$. The present bound comes from dense
representations of the block Gibbs states and separate encodings of
exponentially many jump operators. A possible approach is to prepare
block Gibbs purifications recursively, joining adjacent intervals
through QBP, and then use the purification circuit to encode all
dressed reset jumps collectively. The normalization at each joining
step is controlled by the interaction across a single bond. This
suggests a prefactor $\exp[O((1+\beta J_{\mathrm{int}})^2)]$, while
retaining nearly linear dependence on $N$ and polylogarithmic
dependence on $N/\varepsilon$ at fixed physical parameters, under the
input assumptions of \Cref{qi:theorem}. We leave a detailed complexity
analysis for future work.
\end{remark}

\section{Quantum Dobrushin--Shlosman contraction}
\label{sec:criteria}\label{sec:quantumDS}
   We now give a condition on
a family of Gibbs-preserving channels
that implies Wasserstein contraction and rapid mixing. The condition
allows arbitrary blocks, update rates and evolution times. For the
QBP channels, we will verify it by estimating separately the
contraction inside each block and the influence near its boundary.
{The section concludes by proving that uniform classical DS strong
spatial mixing implies this full quantum condition for classical
Hamiltonians with uniformly bounded interactions.}

\subsection{Classical Dobrushin--Shlosman theory and the choice of block dynamics}
\label{sec:classicalDSbridge}
\label{sec:classicalproof}

Before we introduce the condition in full generality, we motivate it by considering the
classical setting.
We recall that the block heat-bath dynamics repeatedly replaces the spins in a chosen block by a
sample
from their Gibbs distribution conditioned on the remaining spins in the finite system. In a
nutshell, the
Dobrushin--Shlosman (DS) condition allows these
updates to be coupled so that, for sufficiently large blocks, the expected number of sites
at which two configurations differ decreases, effectively yielding a contraction in the
Wasserstein metric.
This mechanism is described in \cite[Theorem~7.1 and Corollary~7.4, Eqs.~(7.5)--(7.6)]{CLW}; the
discrete-time
transport interpretation is that of \cite[Propositions~19--20]{Ollivier}. Here we briefly recall
the
connection between the classical Dobrushin--Shlosman condition and more standard notions of
spatial mixing and their dynamical consequences, before formulating a quantum analogue of DS.

\paragraph{Strong spatial mixing}
Fix a finite lattice $\Lambda\subset\ZZ^D$, $N=|\Lambda|$, with spin space
$[d]$, and a bounded interaction of range $r$. Let $H$ be the corresponding
product-basis diagonal Hamiltonian and let
\[
  \pi(x)=Z^{-1}e^{-\beta H(x)},\qquad x\in[d]^\Lambda,
\]
be its full Gibbs distribution. This $\pi$ is the target measure throughout
this section. For every nonempty $U\subseteq\Lambda$ and configuration
$\xi\in[d]^{\Lambda\setminus U}$, we denote its conditional distribution on $U$ by $
  \pi_U^\xi(u)$, where $\xi$ specifies only the spins in $\Lambda\setminus U$.

\begin{definition}[Dobrushin--Shlosman strong spatial mixing]
  \label{def:classicalDS}
  Write $Q_L=\{0,\ldots,L-1\}^D$ and let $\mathscr V_L$ be the family of
  nonempty finite unions of the disjoint cubes $Lx+Q_L$, $x\in\ZZ^D$.
  We say that $\pi$ satisfies DS strong spatial mixing if there are an integer
  $L_0\ge1$ with $\Lambda\in\mathscr V_{L_0}$ and constants
  $C_{\rm DS},a_{\rm DS}>0$ such that, for every
  $U\in\mathscr V_{L_0}$ with $U\subseteq\Lambda$, every nonempty
  $A\subseteq U$, every $j\in\Lambda\setminus U$, and all
  $\xi,\xi'\in[d]^{\Lambda\setminus U}$ agreeing off $j$,
  \begin{equation}
    \bigl\|(\pi_U^\xi)_A-(\pi_U^{\xi'})_A\bigr\|_{\rm TV}
    \le C_{\rm DS}e^{-a_{\rm DS}\dist(A,j)}.
    \label{eq:DSspatial}
  \end{equation}
  Here total variation is one half of the $\ell^1$ distance. The constants
  are independent of $U,A,j,\xi,\xi'$. For a family of finite-volume
  Hamiltonians, we require the same $L_0,C_{\rm DS},a_{\rm DS}$ and
  interaction range $r$ to obtain bounds uniform in $\Lambda$.
  Equivalently, up to changing the constants, the condition can be written
  in terms of Wasserstein distance: for the same $U,A,j,\xi,\xi'$,
  \begin{equation}
    W_{H,A}\bigl((\pi_U^\xi)_A,(\pi_U^{\xi'})_A\bigr)
    \le C_W e^{-\delta\dist(j,A)},
    \label{eq:DSwasserstein}
  \end{equation}
  where $W_{H,A}$ is the Wasserstein distance for the Hamming metric on
  $[d]^A$, and $C_W,\delta>0$ are uniform over these choices.
  Starting from \eqref{eq:DSspatial}, one may take $\delta=a_{\rm DS}/2$ and
  $C_W=2C_{\rm DS}\sum_{v\in\ZZ^D}e^{-a_{\rm DS}\|v\|_\infty/2}<\infty$.
\end{definition}

\begin{proof}[Proof of equivalence of \eqref{eq:DSspatial} and \eqref{eq:DSwasserstein}]
To derive \eqref{eq:DSwasserstein}, write $p=(\pi_U^\xi)_A$ and
$p'=(\pi_U^{\xi'})_A$. For a $1$-Lipschitz function $f$ on $[d]^A$ in the
Hamming metric, write $p(f):=\sum_{u\in[d]^A}p(u)f(u)$, and similarly for $p'(f)$.
Order $A=\{i_1,\ldots,i_s\}$ by nonincreasing distance from $j$, fix
$u^0\in[d]^A$, and set $f_0=f(u^0)$ and
$f_k(u)=f(u_{i_1},\ldots,u_{i_k},u^0_{i_{k+1}},\ldots,u^0_{i_s})$.
Then $f-f_0=\sum_{k=1}^s(f_k-f_{k-1})$, where each summand has absolute
value at most one and depends only on $A_k=\{i_1,\ldots,i_k\}$.
Since $\dist(A_k,j)=\dist(i_k,j)$, \eqref{eq:DSspatial} gives
\[
  |p(f)-p'(f)|
  \le2\sum_{k=1}^s\|p_{A_k}-p'_{A_k}\|_{\rm TV}
  \le2C_{\rm DS}\sum_{i\in A}e^{-a_{\rm DS}\dist(i,j)}.
\]
Taking the supremum over $f$ gives the same bound for $W_{H,A}(p,p')$
by Kantorovich duality. Split each exponential into two equal factors,
bound one by $e^{-\delta\dist(j,A)}$, and sum the other over $\ZZ^D$
to obtain \eqref{eq:DSwasserstein}.
Conversely, $\|p-p'\|_{\rm TV}\le W_{H,A}(p,p')$ gives
\eqref{eq:DSspatial} with constants $C_W,\delta$.
Taking $A=U$ gives \cite[Eq.~(7.5)]{CLW} for Hamming spin distance,
with decay rate $\delta/D$ when using their $\ell^1$ lattice distance.
\end{proof}
\noindent Conditions \eqref{eq:DSspatial} and \eqref{eq:DSwasserstein} above are imposed on the
conditional distributions of the chosen
$H$, including when $H$ has open boundary conditions. When $H$ is obtained
from an infinite-lattice interaction with fixed exterior spins, the usual
DS spatial estimate implies \eqref{eq:DSspatial} by consistency of the
conditional distributions. The case $L_0=1$ is treated in
\cite[Eq.~(1.2)]{DS1985}; the extension to general $L_0$ is treated in
\cite[Definition~2.4 and Theorem~2.7]{Martinelli} and \cite[Section~1]{CaputoParisi}.

Classical results relate condition \eqref{eq:DSspatial} to spectral-gap and
logarithmic Sobolev estimates; see \cite[Theorem~2.7 and Section~4.2]{Martinelli},
\cite{SZmix,SZdiscrete}, and \cite{Cesi}.
Although these works seek bounds uniform over arbitrarily large volumes and
boundary conditions in infinite-lattice systems, the same proofs, restricted
to conditional distributions within $\Lambda$, give the inequalities directly
for $\pi$ under \eqref{eq:DSspatial}.
Here we use the latter directly to prove Wasserstein contraction of block
heat-bath dynamics towards the full Gibbs measure $\pi$.

\paragraph{Wasserstein contraction of block heat-bath dynamics}
Keep the same $H$ and $\pi$.
For a block $B\subseteq\Lambda$ and configurations $x,y\in[d]^\Lambda$,
write $x=(u,\xi)$ and $y=(v,\xi')$, where $u,v\in[d]^B$ are their respective
restrictions to $B$, and $\xi,\xi'\in[d]^{\Lambda\setminus B}$ are their respective
restrictions to $\Lambda\setminus B$. The full-system heat-bath kernel on
$[d]^\Lambda$ is
\begin{equation}
  \mathsf H_B\bigl((u,\xi),(v,\xi')\bigr)
  =\one_{\{\xi=\xi'\}}\pi_B^\xi(v),
  \qquad\text{ so that }\qquad \pi\mathsf H_B=\pi.
  \label{eq:classical-heatbath-recalled}
\end{equation}
 On diagonal states it is
exactly the Petz channel $\mathcal P_B$ of \eqref{eq:Petz}.

Next, let $\mathscr B$ be a finite labelled family of nonempty blocks, with
$M=|\mathscr B|$, $1\le |B|\le b$ for every $B\in\mathscr B$, and exactly $m$ blocks containing
each site.
 Let $\partial B$ be the sites in
$\Lambda\setminus B$ sharing an interaction term with $B$. The conditional law
$\pi_B^\xi$ depends on $\xi$ only on $\partial B$. We use the Hamming metric
$d_H$ on $[d]^\Lambda$ and its restriction to $[d]^B$, with Wasserstein
distances $W_H$ and $W_{H,B}$, respectively. Write $x\sim_j y$ when $x$ and
$y$ differ at $j$ and agree at every other site. Let $R_B$ be a stochastic
matrix on $[d]^\Lambda$ that leaves the configuration outside $B$ unchanged
and preserves the conditional Gibbs law $\pi_B^\xi$ of $H$ for every fixed
$\xi\in[d]^{\Lambda\setminus B}$; explicitly,
\[
  R_B\bigl((u,\xi),(v,\xi')\bigr)=0\quad\text{if }\xi\ne\xi',
  \qquad
  \sum_{u\in[d]^B}\pi_B^\xi(u)R_B\bigl((u,\xi),(v,\xi)\bigr)
  =\pi_B^\xi(v).
\]
The heat-bath kernel $\mathsf H_B$ in \eqref{eq:classical-heatbath-recalled} is one
such matrix. For a family $(R_B)_{B\in\mathscr B}$ and each site $j\in\Lambda$, define
\[
  D_{R_B}^{\rm diag}(j)
  :=\max_{x\sim_j y}W_H\bigl(R_B(x,\cdot),R_B(y,\cdot)\bigr),
  \qquad
  \Gamma_{\mathscr B}((R_B)_{B\in\mathscr{B}})
  :=\frac1m\max_{j\in\Lambda}\sum_{B\in\mathscr B}
       \bigl(D_{R_B}^{\rm diag}(j)-1\bigr).
\]
This is the diagonal restriction of \eqref{eq:influence}: a diagonal input
with $\Tr_jX=0$ decomposes, separately for each configuration off $j$, into
signed pairs differing at $j$, with total coefficient mass $\|X\|_1/2$.
The equality of classical and quantum $W_1$ on diagonal states is
\cite[Proposition~6]{DPMTL}.

Next, for $j\in\Lambda\setminus B$, introduce the conditional response
{%
\begin{equation}
  c_{Bj}(\pi):=\max_{x\sim_j y}
  W_{H,B}\bigl(\pi_B^{x_{B^c}},\pi_B^{y_{B^c}}\bigr),
  \qquad B^c=\Lambda\setminus B.
  \label{eq:classical-conditional-response}
\end{equation}
}
\noindent Heat bath erases the original disagreement if $j\in B$. If $j\notin B$,
that disagreement remains outside $B$ and contributes exactly one to every
transport cost. Consequently,
\[
  D_{\mathsf H_B}^{\rm diag}(j)
  =\begin{cases}0,&j\in B,\\1+c_{Bj}(\pi),&j\notin B.\end{cases}
\]
We now recall the large-block conclusion of \cite[Corollary~7.4]{CLW}.
Its Wasserstein hypothesis follows from \eqref{eq:DSwasserstein} with
$A=U$, and its block-counting argument applies on the coarse lattice
$L_0\ZZ^D$. For an integer block scale $\ell\ge1$, take the nonempty blocks
\begin{equation}
  B_x=\Lambda\cap(L_0x+Q_{\ell L_0}),\qquad x\in\ZZ^D,
  \label{eq:classical-blocks}
\end{equation}
labelled by $x$, with repeated blocks counted with their multiplicities.
Each block belongs to $\mathscr V_{L_0}$ and has size at most $(\ell L_0)^D$,
and each site belongs to exactly $m=\ell^D$ blocks.
For $\ell$ sufficiently large, depending only on
$D,r,L_0,C_{\rm DS},a_{\rm DS}$, the family
$\mathcal P=(\mathcal P_B)_{B\in\mathscr B}$ therefore satisfies
\begin{equation}
    \Gamma_{\mathscr B}(\mathcal P)
    :=\Gamma_{\mathscr B}((\mathsf H_B)_{B\in\mathscr B})\le-\kappa\quad
\quad\Longleftrightarrow\quad
    \sum_{\substack{B\in\mathscr B\\j\notin B}}c_{Bj}(\pi)
    \le(1-\kappa)m
    \quad\text{for every }j\in\Lambda,
  \label{eq:classicaltransportcriterion}
\end{equation}
with $\kappa>0$; by increasing $\ell$, one may take $\kappa=1/2$.
Here the factor $1/m$ in $\Gamma_{\mathscr B}$ accounts for our
normalization of the generator, whereas \cite{CLW} uses the sum of the
block generators without this factor.

More explicitly, let $h=\lceil r/L_0\rceil+1$.
For a fixed site $j$, a block excluding $j$ can have
$c_{Bj}(\pi)\ne0$ only if it contains a site within distance $r$ of $j$.
Its allowed starting coordinates then lie in a coarse-lattice cube
of side $\ell+2h$. Removing the $\ell^D$ labels whose blocks contain
$j$ leaves at most $(\ell+2h)^D-\ell^D$ labels.
Since \eqref{eq:DSwasserstein} gives $c_{Bj}(\pi)\le C_W$,
\begin{equation}
 \Gamma_{\mathscr B}(\mathcal P)
 \le-1+C_W\big[(1+2h/\ell)^D-1\big].
 \label{eq:classical-block-count}
\end{equation}
The second term tends to zero as $\ell\to\infty$;
choosing it at most $1/2$ gives \eqref{eq:classicaltransportcriterion} with $\kappa=1/2$.

In turn, \cite[Theorems~7.1 and~2.2(b)]{CLW} show that, under
\eqref{eq:classicaltransportcriterion}, the block heat-bath generator
\begin{equation}
  L_{\rm hb}:=\frac1m\sum_B(\mathsf H_B-\id)
             =\frac Mm(P_{\rm hb}-\id),
  \qquad P_{\rm hb}:=\frac1M\sum_B\mathsf H_B,
  \label{eq:HBclocks}
\end{equation}
satisfies, for all probability measures $\mu,\nu$ on $[d]^\Lambda$,
\begin{equation}
  W_H(\mu e^{tL_{\rm hb}},\nu e^{tL_{\rm hb}})
  \le e^{-\kappa t}W_H(\mu,\nu),\qquad t\ge0.
  \label{eq:HBcontraction}
\end{equation}
The corresponding discrete-time estimate follows by coupling the same
block label and applying \cite[Propositions~19--20]{Ollivier}:
\[
  W_H(\mu P_{\rm hb}^k,\nu P_{\rm hb}^k)
  \le\left(1-\frac{\kappa m}{M}\right)^kW_H(\mu,\nu),\qquad k\ge0.
\]
For $\eps>0$, define the worst-case total-variation mixing times by
\[
  \begin{aligned}
    t_{\rm mix}^{\rm hb}(\eps)
    &:=\inf\left\{t\ge0:\ \sup_{\mu}
      \|\mu e^{tL_{\rm hb}}-\pi\|_{\rm TV}\le\eps\right\},\\
    k_{\rm mix}^{\rm hb}(\eps)
    &:=\min\left\{k\in\mathbb{N}:\ \sup_{\mu}
      \|\mu P_{\rm hb}^k-\pi\|_{\rm TV}\le\eps\right\}
  \end{aligned}
\]
where the suprema run over all probability measures $\mu$ on $[d]^\Lambda$.
For the mixing-time scalings, restrict to a box $\Lambda=Q_L$,
where $L=N^{1/D}$ is a positive integer multiple of $L_0$, so $N=L^D$.
The construction above slides a cube of side $\ell L_0$ in steps of $L_0$
and retains only the sites lying in $\Lambda$.
With the block scale fixed sufficiently large,
$m=\mathcal O(1)$, $M=\mathcal O(N)$, and $M/m=\mathcal O(N)$,
while $\kappa$ is bounded below independently of $N$.
Since $\|\nu-\pi\|_{\rm TV}\le W_H(\nu,\pi)\le N$, the preceding
contraction estimates give
$t_{\rm mix}^{\rm hb}(\eps)=\mathcal O(\log(N/\eps))$ and
$k_{\rm mix}^{\rm hb}(\eps)=\mathcal O(N\log(N/\eps))$ for $0<\eps<1$.

\paragraph{Replacing heat bath by classical pseudo-conditional updates}
Keep the fixed full target distribution $\pi$ and its diagonal Hamiltonian
$H$. Let $K_B$ be the classical restriction
of the block generator $\cK_B$ in \eqref{qg:generator}, so that
$K_B(x,y)=\langle y|\cK_B(|x\rangle\langle x|)|y\rangle$.
See \Cref{rem:classicalrates} for this classical restriction.
By \eqref{eq:classicalrates}, its off-diagonal rates within a fixed
exterior configuration $\xi$ are
\[
  r_\xi(u,v)
  =\pi_B^\xi(v)
     \sqrt{\frac{\sigma_\beta^B(u)}{\pi_B^\xi(u)}
            \frac{\sigma_\beta^B(v)}{\pi_B^\xi(v)}},\qquad u\ne v.
\]
Transitions changing $\xi$ have zero rate. Write $R_B(T)=e^{TK_B}$ for the corresponding
finite-time kernels.

\begin{theorem}[Classical DS contraction for finite-time QBP updates]

\label{thm:classical}
Let $\mathscr B$ be an $m$-fold covering, let $b=\max_B|B|$, and suppose
\eqref{eq:classicaltransportcriterion} holds. For each exterior configuration, write
\[
 H(u,\xi)=H_B(u)+W_\xi(u)+H_{B^c}(\xi),
 \qquad \omega_B=\sup_\xi\big[\max_uW_\xi(u)-\min_uW_\xi(u)\big].
\]
Set
\[
 s_* =\min_{B,\xi}\left[e^{\beta\min_uW_\xi(u)/2}
       \sum_v\sigma_\beta^B(v)e^{-\beta W_\xi(v)/2}\right]
 \ge e^{-\beta\max_B\omega_B/2}\ge e^{-\beta\max_B C_B},
\]
and $\chi= m^{-1}\max_j\#\{B:j\in B\cup\partial B\}$. Here $\partial B$ is the set of
exterior sites sharing an interaction with $B$; in a range-$r$ lattice,
$\chi\le(2r+1)^D$. Let $R_B(T)$ be the diagonal restriction of $e^{T\mathcal K_B}$ for
the QBP generator. Then
\begin{equation}
 \Gamma_{\mathscr B}((R_B(T))_B)
 \le-\kappa+2b\chi e^{-s_*T}.
 \label{eq:classicaltransportinheritance}
\end{equation}
In particular,
\begin{equation}
 T\ge s_*^{-1}\log(4b\chi/\kappa)
 \label{qd:classical-time}
\end{equation}
{implies $\Gamma_{\mathscr B}((R_B(T))_B)\le-\kappa/2$.} In a nearest-neighbor chain,
$C_B\le2J_{\mathrm{int}}$ for
every interval, independently of its length.
\end{theorem}
\begin{proof}

  By \Cref{rem:classicalrates}, the diagonal restriction has rates
\eqref{eq:classicalrates}, leaves the exterior configuration fixed,
and has stationary block law $\pi_B^\xi$.

 Put
\[
 Z_\xi=\sum_v\sigma_\beta^B(v)e^{-\beta W_\xi(v)/2},\qquad
 \mu_\xi(v)=Z_\xi^{-1}\sigma_\beta^B(v)e^{-\beta W_\xi(v)/2},\qquad
 \lambda_\xi(u)=Z_\xi e^{\beta W_\xi(u)/2}.
\]
 For an observable $f:\Omega_B\to\RR$, write
$\mu_\xi(f)=\sum_v\mu_\xi(v)f(v)$.     For this fixed exterior configuration, the generator
 is
   \[
 K_Bf(u)=\sum_{v\ne u}r_\xi(u,v)[f(v)-f(u)]
       =\lambda_\xi(u)[\mu_\xi(f)-f(u)].
\]
 Since $\lambda_\xi(u)\ge s_*$, split it into a common reset to $\mu_\xi$ at rate
$s_*$
and the nonnegative residual rates
$(\lambda_\xi(u)-s_*)\mu_\xi(v)$.
Couple two copies by using the same common reset times and draws. After their first
common reset they coincide, and their residual transitions can be coupled identically
thereafter. Thus their probability of disagreement at time $T$ is at most $e^{-s_*T}$.
Coupling one copy initially distributed as $\pi_B^\xi$, which is stationary, proves
\[
 \sup_xW_H(R_B(T)(x,\cdot),\mathsf H_B(x,\cdot))\le |B|e^{-s_*T}.
\]
The rate bound follows from
$Z_\xi\ge e^{-\beta\max W_\xi/2}$ and $\omega_B\le2C_B$.
   For every pair $x\sim_jy$, the triangle inequality gives
\begin{align*}
 W_H(R_B(T)(x,\cdot),R_B(T)(y,\cdot))
 &\le W_H(\mathsf H_B(x,\cdot),\mathsf H_B(y,\cdot))
       +2|B|e^{-s_*T},\\
 D_{R_B(T)}^{\rm diag}(j)
 &\le D_{\mathsf H_B}^{\rm diag}(j)+2|B|e^{-s_*T}.
\end{align*}  The change is
zero if $j\notin B\cup\partial B$, because both kernels then act independently of $j$
and keep it unchanged. Summing proves \eqref{eq:classicaltransportinheritance}. For the
range bound, every contributing block contains a site in the radius-$r$ neighborhood of
$j$, and every such site belongs to $m$ blocks.
\end{proof}

For $T$ satisfying \eqref{qd:classical-time}, the classical dynamics
  \begin{equation}
    P_T=\frac1M\sum_B R_B(T),\qquad
    L_T=\frac1m\sum_B(R_B(T)-\id)=\frac Mm(P_T-\id)
    \label{eq:classicalouter}
  \end{equation}
  satisfy, for all probability measures $\mu,\nu$ on $[d]^\Lambda$,
  integers $k\ge0$, and $t\ge0$,
  \begin{equation}
    W_H(\mu (P_T)^k,\nu (P_T)^k)
      \le\left(1-\frac{\kappa m}{2M}\right)^kW_H(\mu,\nu),\quad
    W_H(\mu e^{tL_T},\nu e^{tL_T})
      \le e^{-\kappa t/2}W_H(\mu,\nu).
    \label{eq:classicalmodifiedcontraction}
  \end{equation}
Coupling the same block label and applying
\cite[Propositions~19--20]{Ollivier} gives the discrete bound;
Poissonization gives the continuous bound.

\subsection{Quantum Dobrushin--Shlosman contraction}

The preceding section motivates a quantum analogue of the
classical contraction condition in \Cref{thm:classical}: replace
$R_B(T)$ by the pseudo-conditional expectations $\Phi_B(T)=e^{T\mathcal K_B}$.
For a classical $H$, these channels restrict to $R_B(T)$ on diagonal
states, where quantum $W_1$ agrees with $W_H$
\cite[Proposition~6]{DPMTL}. For an $m$-fold covering with rates $1/m$,
restricting the following condition to diagonal inputs therefore
recovers the classical finite-time DS condition. The quantum condition
tests all Hermitian inputs satisfying $\Tr_jX=0$ in \eqref{eq:influence}.
Here, we also allow each block update to be applied at its own rate $a_B>0$.
Choosing the relative rates can improve convergence by updating slowly
relaxing regions more often, as in classical Gibbs sampling
\cite{CLR2018,MM2017}.

\begin{definition}[Weighted quantum Dobrushin--Shlosman condition]

\label{def:quantumDS}
Let $(\Phi_B)_{B\in\mathscr B}$ be a finite labelled family of channels with common
fixed state $\sigma$, and let $a_B>0$. With $D_{\Phi_B}(j)$ from \eqref{eq:influence},
define
\[
 \kappa=\min_{j\in\Lambda}\sum_{B\in\mathscr B}a_B[1-D_{\Phi_B}(j)].
\]
The family satisfies the weighted quantum Dobrushin--Shlosman condition when $\kappa>0$.
For an $m$-fold covering with $a_B=1/m$, this is
\begin{equation}
 \Gamma_{\mathscr B}(\Phi):=\frac{1}m\max_j\sum_B[D_{\Phi_B}(j)-1]\le-\kappa.
 \label{eq:qds}
\end{equation}
   The
sums count every label , including when    distinct labelled blocks
have the same intersection with the finite lattice .
\end{definition}
\noindent  For a channel $P$ and a generator $\mathcal L$ with stationary state
$\sigma$, we use trace-norm mixing times
\begin{equation}
 \begin{split}
 t_{\rm mix}(\mathcal L,\eps)
 &:=\inf\{t\ge0:\sup_\rho\|e^{t\mathcal L}(\rho)-\sigma\|_1\le\eps\},\\
 k_{\rm mix}(P,\eps)
 &:=\min\{k\in\NN_0:\sup_\rho\|P^k(\rho)-\sigma\|_1\le\eps\}.
 \end{split}
 \label{eq:quantum-mixing-definition}
\end{equation}
Both suprema are over density matrices on the full system.

  \begin{theorem}[Wasserstein contraction and mixing times]

\label{thm:qds}
Under the preceding condition, put $a=\sum_Ba_B$ and
\[
 P=\frac{1}a\sum_Ba_B\Phi_B,\qquad
 \mathcal L=\sum_Ba_B(\Phi_B-\id)=a(P-\id).
\]
For every traceless Hermitian $X$, every $t\ge0$, and every integer $k\ge0$,
\begin{equation}
 \|P^k(X)\|_{W_1}\le(1-\kappa/a)^k\|X\|_{W_1},\qquad
 \|e^{t\mathcal L}(X)\|_{W_1}\le e^{-\kappa t}\|X\|_{W_1}.
 \label{eq:qdscontraction}
\end{equation}
In particular, $\sigma$ is the unique stationary state and, for $0<\varepsilon<1$,
\begin{equation}
 t_{\mathrm{mix}}(\mathcal L,\varepsilon)\le\kappa^{-1}\log(2N/\varepsilon),\qquad
 k_{\mathrm{mix}}(P,\varepsilon)\le
 \left\lceil\frac{a}{\kappa}\log(2N/\varepsilon)\right\rceil.
 \label{eq:qdsTV}
\end{equation}
\end{theorem}
\begin{proof}
For an admissible decomposition $X=\sum_jX_j$ in \eqref{eq:Wdef},
\[
 \|P(X)\|_{W_1}
 \le\frac{1}a\sum_j\sum_Ba_B D_{\Phi_B}(j)\frac{\|X_j\|_1}{2}
 \le(1-\kappa/a)\frac{1}{2}\sum_j\|X_j\|_1.
\]
Taking the infimum and iterating proves the discrete estimate. The nonnegative
coefficients in
\[
 e^{t\mathcal L}=e^{-at}\sum_{k\ge0}\frac{(at)^k}{k!}P^k
\]
give the continuous estimate. Each channel fixes $\sigma$, and \Cref{lem:transport}
gives
$\frac{1}{2}\|\rho-\sigma\|_1\le\|\rho-\sigma\|_{W_1}\le N$.
Thus $\|e^{t\mathcal L}\rho-\sigma\|_1\le2Ne^{-\kappa t}$, and the analogous discrete
bound follows using $(1-z)^k\le e^{-kz}$.
\end{proof}

\begin{remark}
\label{rem:quantumDS-comparison}
The Wasserstein contractions derived above can be viewed as
quantum coarse Ricci curvature bounds in the sense of Gao--Rouz\'e
\cite{GR}, extending Ollivier's classical notion.
De Palma--Rouz\'e \cite[Section~4, Proposition~9]{DPR} establish such
contraction at high temperature for commuting Gibbs states using
single-site recovery maps. Related Dobrushin-type conditions underlie
the optimal mixing bounds of \cite{RFA2,BLMT}; their high-temperature
constructions use jump operators obtained from single-site operators.
The general contraction principle of \cite{BLMT} also applies to block
channels; here we construct such updates and control their influences.
The update size is a parameter to optimize: the channels $\Phi_B$ relax entire
blocks, allowing contraction in the interior
to compensate for the influence near the boundary.
In spirit, this extends the previous quantum results as the classical
Dobrushin--Shlosman theory extends the single-site Dobrushin approach
\cite{DS,CLW}.
{The following sections use this freedom to prove rapid mixing over a
wider temperature range;} in particular, \Cref{q1:main} gives rapid mixing at
every fixed finite $\beta$ for bounded finite-range quantum chains,
beyond the general high-temperature bounds of \cite{RFA2,BLMT}.
The classical theory already illustrates the reason for passing to
blocks: a single-site Dobrushin bound can fail while a block condition
still holds \cite{DS,CLW}.
Classically, failure of that bound does not preclude rapid mixing of
the single-site dynamics \cite{Martinelli}.
We leave relating the mixing times of our block dynamics to those of
the samplers considered in \cite{RFA2,BLMT} to future work.
\end{remark}

{%
\subsection{Classical DS implies quantum DS}
\label{sec:classical-to-quantum}
\label{qp:classical-example}
We now prove that, for uniformly bounded interactions, uniform classical DS strong spatial mixing implies the full
quantum DS condition for the smoothed heat-bath block channels, with sufficiently
large blocks and evolution times bounded independently of the volume. Let $H^{\mathrm c}=\sum_Xh_X^{\mathrm c}$ be any finite-range,
product-basis diagonal Hamiltonian. Fix a range bound $r$
at least as large as the interaction range, and choose
$R_B=B^{\le r}$, containing the supports of all prescribed interactions
crossing $B$. Put $E_B=R_B\setminus B$. For $u\in\Omega_B$
and an exterior configuration $\zeta\in\Omega_{B^c}$, recall the
crossing interaction and its conditional energy:
\begin{equation}
 W_{\partial B}^{\mathrm c}
 :=\sum_{\substack{X\cap B\ne\varnothing\\X\cap B^c\ne\varnothing}}
 h_X^{\mathrm c},\qquad
 W_\zeta(u):=\langle u,\zeta|W_{\partial B}^{\mathrm c}|u,\zeta\rangle.
 \label{qd:block-crossing-energy}
\end{equation}
Thus $W_\zeta$ is the interaction energy across the boundary of
$B$ with the exterior spins fixed to $\zeta$. In particular,
\[
 H^{\mathrm c}(u,\zeta)=H_B(u)+W_\zeta(u)+H_{B^c}(\zeta),\qquad
 \pi_B^\zeta(u)=\frac{e^{-\beta[H_B(u)+W_\zeta(u)]}}
 {\sum_v e^{-\beta[H_B(v)+W_\zeta(v)]}}.
\]
The potential $W_\zeta$ depends only on $\zeta|_{E_B}$.
In the local block formulas below, exterior configurations are indexed
by $\xi\in\Omega_{E_B}$; $W_\xi$ and $\pi_B^\xi$ are defined using any extension of $\xi$ to $B^c$.
In operator formulas, we use the same symbol for the conditional Gibbs
state on $\HH_B$, defined by
\begin{equation}
 \pi_B^\xi:=\sum_{u\in\Omega_B}\pi_B^\xi(u)|u\rangle\langle u|.
 \label{qp:conditional-gibbs-state}
\end{equation}
Thus $\pi_B^\xi(u)$ denotes its diagonal entry, and classical
Wasserstein distances between these states refer to their probability laws.
Recall that, for $j\notin B$, $c_{Bj}(\pi)$ defined in
\eqref{eq:classical-conditional-response} is the largest Hamming
Wasserstein distance between $\pi_B^\xi$ and $\pi_B^{\xi'}$ over
exterior configurations agreeing off $j$. The formula below has the same Gibbs-weighted pinching structure as
the conditional expectations in \cite[Theorem~4.7]{SCGLcss}; here we
prove it directly for the QBP generator of a general local classical
Hamiltonian.
For each $u\in\Omega_B$, define the diagonal operator on $\HH_{E_B}$
\[
 Q_{B,u}:=\sum_{\xi\in\Omega_{E_B}}
       \pi_B^\xi(u)|\xi\rangle\langle\xi|.
\]
These operators commute, are positive, and satisfy
$\sum_u Q_{B,u}=I_{E_B}$. Let $(P_C)_C$ be their nonzero joint
spectral projections, with joint eigenvalues $p_{B,C}(u)$:
\[
 Q_{B,u}=\sum_C p_{B,C}(u)P_C,\qquad u\in\Omega_B.
\]
Thus $p_{B,C}(u)\ge0$ and $\sum_u p_{B,C}(u)=1$. Define the density
matrix on $\HH_B$
\[
 \pi_B^C:=\sum_{u\in\Omega_B}p_{B,C}(u)|u\rangle\langle u|.
\]
For any exterior configuration $\xi$ with $P_C|\xi\rangle=|\xi\rangle$,
one has $p_{B,C}(u)=\pi_B^\xi(u)$ for every $u$. By
\eqref{qp:conditional-gibbs-state}, $\pi_B^C=\pi_B^\xi$ as density
matrices; this state is independent of the choice of $\xi$ within
that spectral subspace.
Define the pinching map and conditional Gibbs density by
\[
 \mathscr P_B(Y)=\sum_C P_CYP_C,\qquad
 \widehat\rho_B=\sum_C\pi_B^C\otimes P_C,
 \qquad \Tr_B\widehat\rho_B=I_{E_B}.
\]
By convention, all maps below act as the identity on $R_B^c$.
\begin{lemma}
\label{qp:intrinsic-projection}
The limit $\Gamma_B:=\lim_{t\to\infty}e^{t\mathcal K_B(H^{\mathrm c})}$
exists in diamond norm on $R_B$. Its Heisenberg adjoint
$\mathbb E_B=\Gamma_B^*$ is the conditional expectation
\begin{equation}
 \mathbb E_B(A)=I_B\otimes\mathscr P_B\!\left(
 \Tr_B[\widehat\rho_B^{1/2}A\widehat\rho_B^{1/2}]\right).
 \label{qp:intrinsic-heisenberg}
\end{equation}
\end{lemma}
\begin{proof}
Recall from \eqref{qg:generator} that the jump map is the completely
positive map $\mathcal T_B$, while $S_B$ is the positive operator
that determines the loss term:
\[
 \mathcal T_B(X)=O_B\left[\sigma_\beta^B\otimes
 \Tr_B(O_B^{-1}XO_B^{-*})\right]O_B^*,
 \qquad S_B=\mathcal T_B^*(\one).
\]
For a classical Hamiltonian, $O_B=e^{-\beta W_{\partial B}^{\mathrm c}/4}$
and $S_B$ are diagonal, and the coherent correction vanishes, as
shown in \Cref{rem:classicalrates}. Thus
\[
 \mathcal K_B(H^{\mathrm c})(X)
 =\mathcal T_B(X)-\tfrac12(S_BX+XS_B).
\]
 The diagonal
entries of the loss operator are
\[
 \langle u,\xi|S_B|u,\xi\rangle
 =\sum_{z\in\Omega_B}r_\xi(u,z).
\]
where the rate $r_\xi(u,v)$ was defined in \eqref{eq:classicalrates}, with $r_\xi(u,u)=\sigma_\beta^B(u)$.
When $u\ne v$, the jump map $\mathcal T_B$ annihilates
$|u,\xi\rangle\langle v,\eta|$, since its partial trace over $B$
vanishes and the conjugating operators are diagonal.
The anticommutator term multiplies this matrix unit by the negative scalar
$-\tfrac12\sum_{z\in\Omega_B}[r_\xi(u,z)+r_\eta(v,z)]$.
Thus exponentiating on this one-dimensional subspace gives
\[
 \begin{aligned}
 e^{t\mathcal K_B(H^{\mathrm c})}(|u,\xi\rangle\langle v,\eta|)=\exp\!\left(-\frac t2\sum_{z\in\Omega_B}
 [r_\xi(u,z)+r_\eta(v,z)]\right)|u,\xi\rangle\langle v,\eta|
 \longrightarrow0.
 \end{aligned}
\]
The limit is zero because all the rates in the sum in the exponent are positive.
Next, when the block indices agree, $\mathcal T_B$ can change their common
value. For each fixed exterior pair $\xi,\eta$, the subspace spanned by
$|u,\xi\rangle\langle u,\eta|$, $u\in\Omega_B$, is invariant.
On this subspace the generator has matrix entries
\begin{equation}
 \begin{aligned}
 \langle v,\xi|\mathcal K_B(H^{\mathrm c})
 (|u,\xi\rangle\langle u,\eta|)|v,\eta\rangle=\sqrt{r_\xi(u,v)r_\eta(u,v)}
 -\frac{\delta_{vu}}2\sum_{z\in\Omega_B}
 [r_\xi(u,z)+r_\eta(u,z)].
 \end{aligned}
 \label{qp:classical-coherence-sector}
\end{equation}
Every power of the generator preserves this same subspace.
Consequently, the power series for $e^{t\mathcal K_B(H^{\mathrm c})}$
restricts to the ordinary matrix exponential of
\eqref{qp:classical-coherence-sector}. We now determine its limit.
The entries in \eqref{qp:classical-coherence-sector} are nonnegative
off the diagonal, and the sum of column $u$ is
\[
 \begin{aligned}
 \sum_{v\in\Omega_B}
 \langle v,\xi|\mathcal K_B(H^{\mathrm c})
 (|u,\xi\rangle\langle u,\eta|)|v,\eta\rangle=-\frac12\sum_{v\in\Omega_B}
    (\sqrt{r_\xi(u,v)}-\sqrt{r_\eta(u,v)})^2.
 \end{aligned}
\]
This sum vanishes exactly when $r_\xi(u,v)=r_\eta(u,v)$ for every $v$.
By \eqref{eq:classicalrates} and the strict positivity of
$\sigma_\beta^B$, this is equivalent to $\beta(W_\xi-W_\eta)$
being constant, hence to $\pi_B^\xi=\pi_B^\eta$.
In particular, either all column sums vanish or all are strictly negative. Thus:

If $\pi_B^\xi=\pi_B^\eta$, the rates coincide, so
\eqref{qp:classical-coherence-sector} is exactly the classical block
Markov generator. Its off-diagonal transition rates are positive,
and its stationary law is $\pi_B^\xi$, as shown in
\Cref{rem:classicalrates}. Its exponential therefore converges to the
matrix whose every column is $\pi_B^\xi$.

If $\pi_B^\xi\ne\pi_B^\eta$, all column sums are strictly negative, so
\[
 \gamma_{\xi\eta}:=\frac12\min_{z\in\Omega_B}
 \sum_{v\in\Omega_B}
 (\sqrt{r_\xi(z,v)}-\sqrt{r_\eta(z,v)})^2>0.
\]
Since $\tfrac12\sum_z[r_\xi(u,z)+r_\eta(u,z)]\le\|S_B\|$,
the matrix in \eqref{qp:classical-coherence-sector}, shifted by
$\|S_B\|\id$, has nonnegative entries. Hence, for $t\ge0$,
\[
 \begin{aligned}
 \langle v,\xi|e^{t\mathcal K_B(H^{\mathrm c})}
 (|u,\xi\rangle\langle u,\eta|)|v,\eta\rangle=e^{-t\|S_B\|}\sum_{n=0}^{\infty}\frac{t^n}{n!}
 \langle v,\xi|(\mathcal K_B(H^{\mathrm c})+\|S_B\|\id)^n
 (|u,\xi\rangle\langle u,\eta|)|v,\eta\rangle\ge0.
 \end{aligned}
\]
Differentiating the exponential and using the computed column sums gives
\[
 \begin{aligned}
 \frac{d}{dt}\sum_{v\in\Omega_B}
 \langle v,\xi|e^{t\mathcal K_B(H^{\mathrm c})}
 (|u,\xi\rangle\langle u,\eta|)|v,\eta\rangle&=-\frac12\sum_{z\in\Omega_B}
 \left[\sum_{v\in\Omega_B}
 (\sqrt{r_\xi(z,v)}-\sqrt{r_\eta(z,v)})^2\right]\\
 &\qquad\qquad \times
 \langle z,\xi|e^{t\mathcal K_B(H^{\mathrm c})}
 (|u,\xi\rangle\langle u,\eta|)|z,\eta\rangle\\
 &\le-\gamma_{\xi\eta}\sum_{z\in\Omega_B}
 \langle z,\xi|e^{t\mathcal K_B(H^{\mathrm c})}
 (|u,\xi\rangle\langle u,\eta|)|z,\eta\rangle.
 \end{aligned}
\]
The column sum equals $1$ at $t=0$. Gr\"onwall's inequality therefore yields
\[
 \begin{aligned}
 0&\le\langle v,\xi|e^{t\mathcal K_B(H^{\mathrm c})}
 (|u,\xi\rangle\langle u,\eta|)|v,\eta\rangle\\
 &\le\sum_{z\in\Omega_B}
 \langle z,\xi|e^{t\mathcal K_B(H^{\mathrm c})}
 (|u,\xi\rangle\langle u,\eta|)|z,\eta\rangle\\
 &\le e^{-\gamma_{\xi\eta}t}\longrightarrow0.
 \end{aligned}
\]
By the definition of the joint spectral projections and
\eqref{qp:conditional-gibbs-state},
$P_C|\xi\rangle=|\xi\rangle$ if and only if
$\pi_B^\xi=\pi_B^C$. The limits established above therefore give
\[
 \lim_{t\to\infty}e^{t\mathcal K_B(H^{\mathrm c})}
 (|u,\xi\rangle\langle v,\eta|)
 =\delta_{uv}\sum_C\pi_B^C\otimes
 P_C|\xi\rangle\langle\eta|P_C.
\]
For an arbitrary operator $X$ on $\HH_{R_B}$, recall that
\[
 \Tr_BX
 =\sum_{\xi,\eta\in\Omega_{E_B}}
 \left(\sum_{u\in\Omega_B}
 \langle u,\xi|X|u,\eta\rangle\right)
 |\xi\rangle\langle\eta|.
\]
Expanding $X$ in matrix units and applying the preceding limit yields
\begin{equation}
 \begin{aligned}
 \Gamma_B(X)
 &=\sum_C\sum_{\substack{u\in\Omega_B\\
                         \xi,\eta\in\Omega_{E_B}}}
 \langle u,\xi|X|u,\eta\rangle\,
 \pi_B^C\otimes P_C|\xi\rangle\langle\eta|P_C\\
 &=\sum_C\pi_B^C\otimes P_C(\Tr_BX)P_C.
 \end{aligned}
 \label{qp:intrinsic-projection-formula}
\end{equation}
Taking the adjoint gives
\eqref{qp:intrinsic-heisenberg}.
\end{proof}
\noindent For $j\in E_B$, put
\[
 S_{Bj}=\{j\}\cup
 \bigcup_{\substack{X\text{ crossing }B:\,j\in X}}
          (X\setminus B),
\]
where the union uses the prescribed possible interaction supports,
including those with zero coefficient. Since each interaction set
containing $j$ has diameter at most $r$,
\[
 S_{Bj}\subseteq\{j\}^{\le r},\qquad |S_{Bj}|\le(2r+1)^D.
\]
The same bound holds for the number of exterior sites in crossing
terms containing any given $j\in B$.

\begin{lemma}[Full quantum influence of the conditional projection]
\label{qp:lift-lemma}
For arbitrary Hermitian inputs, including inputs entangled with the exterior,
\begin{equation}
 D_{\Gamma_B}(j)\le
 \begin{cases}
 0,&j\in B,\\
 2|S_{Bj}|+2c_{Bj}(\pi),&j\in E_B,\\
 1,&j\notin R_B.
 \end{cases}
 \label{qp:lift-influence}
\end{equation}
\end{lemma}

\begin{proof}
Let $X=X^*$ satisfy $\Tr_jX=0$. If $j\in B$, then $\Tr_BX=0$, so
\eqref{qp:intrinsic-projection-formula} gives $\Gamma_B(X)=0$.
If $j\notin R_B$, the channel is supported away from $j$, and
\Cref{lem:transport} gives $D_{\Gamma_B}(j)=1$.

It remains to consider $j\in E_B$. Since $\Tr X=0$ and $\Gamma_B$
is a channel, $\Gamma_B(X)$ is traceless and Hermitian. Applying
\eqref{qp:compact-reconstruction} with $Y=\Gamma_B(X)$ and $S=S_{Bj}$ gives
\[
 \begin{aligned}
 \|\Gamma_B(X)\|_{W_1}
 &{\le}
 |S_{Bj}|\,\|\Gamma_B(X)\|_1+
 \|\Tr_{S_{Bj}}\Gamma_B(X)\|_{W_1}\\
 &\le |S_{Bj}|\,\|X\|_1+
 \|\Tr_{S_{Bj}}\Gamma_B(X)\|_{W_1},
 \end{aligned}
\]
where the second line uses trace-norm contractivity. It therefore remains to prove
$\|\Tr_{S_{Bj}}\Gamma_B(X)\|_{W_1}\le c_{Bj}(\pi)\|X\|_1$. Write the exterior configuration as $(a,\zeta,\eta)$, with
\[
 a\in[d_j],\qquad
 \zeta\in\Omega_{S_{Bj}\setminus\{j\}},\qquad
 \eta\in\Omega_{E_B\setminus S_{Bj}},
\]
and choose $a_0\in[d_j]$. Every crossing term containing $j$ has
its exterior support contained in $S_{Bj}$. Such terms cancel in the
following difference, while all other terms are independent of $a$:
\[
 W_{a\zeta\eta}(u)-W_{a\zeta\eta'}(u)
 =W_{a_0\zeta\eta}(u)-W_{a_0\zeta\eta'}(u).
\]
Consequently, for each $\zeta$, the commuting family
\[
 \left(\sum_{\eta\in\Omega_{E_B\setminus S_{Bj}}}
 \pi_B^{a\zeta\eta}(u)|\eta\rangle\langle\eta|\right)_{u\in\Omega_B}
\]
has the same joint spectral projections for every $a$; denote them
by $(P_C^\zeta)_C$. Choose one configuration $\eta_{\zeta,C}$ in
each nonzero projection, so that
$P_C^\zeta|\eta_{\zeta,C}\rangle=|\eta_{\zeta,C}\rangle$. Thus, for every $a,u$,
\[
 \sum_\eta\pi_B^{a\zeta\eta}(u)|\eta\rangle\langle\eta|
 =\sum_C\pi_B^{a\zeta\eta_{\zeta,C}}(u)P_C^\zeta.
\]
Define the Hermitian operators
\[
 A_{a,\zeta,C}:=P_C^\zeta
 \langle a,\zeta|\Tr_BX|a,\zeta\rangle P_C^\zeta,
\]
retaining all factors on $R_B^c$. Independence of $P_C^\zeta$ from
$a$ and $\Tr_jX=0$ give
\[
 \sum_aA_{a,\zeta,C}
 =P_C^\zeta\langle\zeta|\Tr_{B\cup\{j\}}X|\zeta\rangle P_C^\zeta=0.
\]
Partial trace, measurement of $(a,\zeta)$, and spectral pinching,
with their outcome labels retained, form a trace-preserving completely
positive map. Hence
\[
 \sum_{a,\zeta,C}\|A_{a,\zeta,C}\|_1\le\|X\|_1.
\]
Using the conditional Gibbs states \eqref{qp:conditional-gibbs-state}
in \eqref{qp:intrinsic-projection-formula} and tracing $S_{Bj}$ yields
\[
 \begin{aligned}
 \Tr_{S_{Bj}}\Gamma_B(X)
 &=\sum_{a,\zeta,C}\pi_B^{a\zeta\eta_{\zeta,C}}\otimes A_{a,\zeta,C}\\
 &=\sum_{a,\zeta,C}
 (\pi_B^{a\zeta\eta_{\zeta,C}}-\pi_B^{a_0\zeta\eta_{\zeta,C}})
 \otimes A_{a,\zeta,C}.
 \end{aligned}
\]
In the following tensor-product estimate, $\|\cdot\|_{W_1}$ is the
quantum transport norm
\eqref{eq:Wdef} on the retained sites $\Lambda\setminus S_{Bj}$,
whereas $W_{H,B}$ is the classical Hamming Wasserstein distance on $B$.
The states $\pi_B^\xi$ in these tensor products are defined by
\eqref{qp:conditional-gibbs-state}.
For any Hermitian operator $A$ on $\Lambda\setminus(B\cup S_{Bj})$,
the operator $(\pi_B^\xi-\pi_B^{\xi'})\otimes A$ is traceless and
Hermitian, so this quantum norm applies even when $A$ is neither
diagonal nor positive.

For a one-site Hamming step $x\sim_i y$ in $B$,
\[
 \Tr_i\bigl[(|x\rangle\langle x|-|y\rangle\langle y|)\otimes A\bigr]=0,
 \qquad
 \frac12\bigl\|(|x\rangle\langle x|-|y\rangle\langle y|)\otimes A\bigr\|_1
 =\|A\|_1.
\]
Thus \Cref{lem:transport} gives quantum $W_1$-norm $\|A\|_1$ for
each such step. Connecting arbitrary $x,y\in\Omega_B$ by
$d_H(x,y)$ one-site steps and applying the triangle inequality yields
\[
 \bigl\|(|x\rangle\langle x|-|y\rangle\langle y|)\otimes A\bigr\|_{W_1}
 \le d_H(x,y)\|A\|_1.
\]
Now let $q(x,y)$ be an optimal classical coupling of $\pi_B^\xi$
and $\pi_B^{\xi'}$, so that $q(x,y)\ge0$ and
\[
 \sum_yq(x,y)=\pi_B^\xi(x),\qquad
 \sum_xq(x,y)=\pi_B^{\xi'}(y).
\]
By \eqref{qp:conditional-gibbs-state}, these marginals give the operator identity
\[
 (\pi_B^\xi-\pi_B^{\xi'})\otimes A
 =\sum_{x,y\in\Omega_B}q(x,y)
 (|x\rangle\langle x|-|y\rangle\langle y|)\otimes A.
\]
Consequently,
\[
 \begin{aligned}
 \|(\pi_B^\xi-\pi_B^{\xi'})\otimes A\|_{W_1}
 &\le\sum_{x,y}q(x,y)
 \bigl\|(|x\rangle\langle x|-|y\rangle\langle y|)\otimes A\bigr\|_{W_1}\\
 &\le\sum_{x,y}q(x,y)d_H(x,y)\|A\|_1\\
 &=W_{H,B}(\pi_B^\xi,\pi_B^{\xi'})\|A\|_1.
 \end{aligned}
\]
The configurations $a\zeta\eta_{\zeta,C}$ and
$a_0\zeta\eta_{\zeta,C}$ agree off $j$, so the definition of $c_{Bj}(\pi)$ implies
\[
 \|\Tr_{S_{Bj}}\Gamma_B(X)\|_{W_1}
 \le c_{Bj}(\pi)\sum_{a,\zeta,C}\|A_{a,\zeta,C}\|_1
 \le c_{Bj}(\pi)\|X\|_1.
\]
Substituting this bound into the estimate at the start of the proof gives
\[
 \|\Gamma_B(X)\|_{W_1}\le(|S_{Bj}|+c_{Bj}(\pi))\|X\|_1.
\]
Dividing by $\|X\|_1/2$ in \eqref{eq:influence} proves
$D_{\Gamma_B}(j)\le2|S_{Bj}|+2c_{Bj}(\pi)$.
\end{proof}

\noindent The next lemma passes from the limiting channels of
\Cref{qp:intrinsic-projection,qp:lift-lemma} to finite-time updates,
with times bounded independently of the volume when all local
interaction strengths, including the on-site terms, are bounded.

\begin{lemma}[Uniform finite-time influence bounds]
\label{qp:compact-lift}
\label{app:perturbation-compactness}
\label{app:perturbation-lift}
Let $H^{\mathrm c}=\sum_Xh_X^{\mathrm c}$ be a classical Hamiltonian
on $\Lambda\subset\mathbb Z^D$ with $\dim\HH_i\le d$ and interaction
range at most $r$. Suppose that
\begin{equation}
 \sup_{i\in\Lambda}\sum_{X\ni i}\|h_X^{\mathrm c}\|\le J,
 \qquad |R_B|\le r_0,\qquad
 c_{Bj}(\pi)\le C_W\quad(j\in E_B),
 \label{qp:bounded-classical-strength}
\end{equation}
for every block $B$, where the sum includes the on-site terms and
$c_{Bj}(\pi)$ is defined in \eqref{eq:classical-conditional-response}.
Set
\[
 q_B(j)=
 \begin{cases}
 0,&j\in B,\\
 2(2r+1)^D+2C_W,&j\in R_B\setminus B,\\
 1,&j\notin R_B.
 \end{cases}
\]
For every $\eta>0$ there is a finite $T>0$, depending only on
$\beta,d,D,r,r_0,J,C_W,\eta$, such that each block admits a time
$0<T_B\le T$ with
\[
 D_{e^{T_B\mathcal K_B(H^{\mathrm c})}}(j)
 \le q_B(j)+\eta\mathbf1_{\{j\in R_B\}}.
\]
\end{lemma}

\begin{proof}
Fix $B\subseteq R_B$, the site dimensions and the prescribed
interaction supports in $R_B$. Let $\mathfrak H_B$ consist of the
diagonal interactions $H'=\sum_{X\subseteq R_B}h_X'$ on these
supports satisfying \eqref{qp:bounded-classical-strength}, with
$\pi'$ the Gibbs distribution of $H'$. This set is compact in the
interaction terms: they are bounded, and the conditional laws and
$c_{Bj}(\pi')$ depend continuously on them. Moreover,
\[
 \|H'_B\|\le r_0J,\qquad
 \sigma_{\beta,H'}^B\ge d^{-r_0}e^{-2\beta r_0J}\one_B>0.
\]
Thus \Cref{qp:intrinsic-projection,qp:lift-lemma} apply to every
$H'\in\mathfrak H_B$. Denote its limiting block channel by
$\Gamma_B(H')$. For each $H'$, choose $t_{H'}>0$ such that
\[
 \|e^{t_{H'}\mathcal K_B(H')}-\Gamma_B(H')\|_\diamond
 <\frac{\eta}{4r_0}.
\]
By continuity of the generator formula \eqref{qg:generator} and
the matrix exponential, there exists $\delta_{H'}>0$ such that
\[
 \begin{gathered}
 H''=\sum_Xh_X''\in\mathfrak H_B,\quad 
 \sum_{X\subseteq R_B}\|h_X''-h_X'\|<\delta_{H'}
 \Longrightarrow
 \|e^{t_{H'}\mathcal K_B(H'')}-e^{t_{H'}\mathcal K_B(H')}\|_\diamond
 <\frac{\eta}{4r_0}.
 \end{gathered}
\]
The sum runs over the fixed prescribed interaction supports.
For every such $H''$, the triangle inequality gives
\[
 \begin{aligned}
 \|e^{t_{H'}\mathcal K_B(H'')}-\Gamma_B(H')\|_\diamond
\le\|e^{t_{H'}\mathcal K_B(H'')}-e^{t_{H'}\mathcal K_B(H')}\|_\diamond+\|e^{t_{H'}\mathcal K_B(H')}-\Gamma_B(H')\|_\diamond
 <\frac{\eta}{2r_0}.
 \end{aligned}
\]
For $j\in R_B$ and any nonzero Hermitian $X$ on $\HH_\Lambda$
with $\Tr_jX=0$, \eqref{eq:localW} gives
\[
 \|(e^{t_{H'}\mathcal K_B(H'')}-\Gamma_B(H'))(X)\|_{W_1}
 \le |R_B|\,
 \|e^{t_{H'}\mathcal K_B(H'')}-\Gamma_B(H')\|_\diamond\|X\|_1.
\]
Apply the triangle inequality, divide by $\|X\|_1/2$, and take the
supremum over these $X$ as in \eqref{eq:influence}. Since
$|R_B|\le r_0$, this yields
\[
 \begin{aligned}
 D_{e^{t_{H'}\mathcal K_B(H'')}}(j)
 &\le D_{\Gamma_B(H')}(j)
   +2r_0\|e^{t_{H'}\mathcal K_B(H'')}-\Gamma_B(H')\|_\diamond\\
 &<q_B(j)+\eta\qquad(j\in R_B).
 \end{aligned}
\]
The open balls $\sum_X\|h_X''-h_X'\|<\delta_{H'}$, centred at
$H'\in\mathfrak H_B$, cover $\mathfrak H_B$. Choose a finite subcover
with centres $H'_1,\ldots,H'_m$. Since $H^{\mathrm c}_{R_B}\in\mathfrak H_B$,
it belongs to the ball centred at some $H'_k$. Apply the preceding
influence bound with
\[
 H'=H'_k,\qquad H''=H^{\mathrm c}_{R_B},\qquad
 T_B:=t_{H'_k}\le\max_{1\le a\le m}t_{H'_a}.
\]
By \Cref{rem:classicalrates},
$\mathcal K_B(H^{\mathrm c}_{R_B})=\mathcal K_B(H^{\mathrm c})$,
so this choice gives the required bound for the original Hamiltonian.
There are only finitely many choices of supports, site dimensions and
$B\subseteq R_B$, up to relabelling of at most $r_0$ sites.
Taking the maximum of the selected times over these choices gives
$T$ with the stated dependence. Outside $R_B$, the influence is one
by \Cref{lem:transport}.
\end{proof}

\noindent Combining the preceding lemmas, we now show that classical
strong spatial mixing implies the quantum DS condition for sufficiently
large blocks and suitably chosen evolution times.

\begin{proposition}[Classical strong spatial mixing implies quantum DS]
\label{qp:classical-qds}
Consider a family of finite-range classical Hamiltonians
$H^{\mathrm c}=\sum_Xh_X^{\mathrm c}$ on $\Lambda\subset\mathbb Z^D$
satisfying the uniform strong spatial mixing condition of
\Cref{def:classicalDS}. Assume that
\[
 \sup_{i\in\Lambda}\sum_{X\ni i}\|h_X^{\mathrm c}\|\le J
\]
with the same finite $J$ for all Hamiltonians and volumes in the
family. Fix bounds $d$ on the
one-site dimensions and $r$ on the interaction range.
There are a block scale $\ell$, a finite $T>0$, and times
$0<T_B\le T$ for the labelled blocks \eqref{eq:classical-blocks},
with rates $a_B=\ell^{-D}$, such that
\begin{equation}
 \Phi_B^{\mathrm c}=e^{T_B\mathcal K_B(H^{\mathrm c})},\qquad
 \min_j\sum_Ba_B[1-D_{\Phi_B^{\mathrm c}}(j)]\ge\frac12.
 \label{qp:classical-reference-certificate}
\end{equation}
The choices of $\ell$ and
the upper bound $T$ depend only on $L_0,C_{\rm DS},a_{\rm DS}$
from \Cref{def:classicalDS} and on $\beta,J,d,D,r$.
Hence the family $(\Phi_B^{\mathrm c})_B$, with rates $a_B=\ell^{-D}$,
satisfies the quantum DS condition \eqref{eq:qds} with constant $1/2$.
\end{proposition}

\begin{proof}
Use the neighbourhoods $R_B=B^{\le r}$ chosen above, and put
$h=\lceil r/L_0\rceil+1$. Equation~\eqref{eq:DSwasserstein} gives
$c_{Bj}(\pi)\le C_W$. Since $|R_B|\le(\ell L_0+2r)^D$,
\Cref{qp:compact-lift} gives, for any $\eta>0$, times $T_B\le T$
such that
\[
 D_{\Phi_B^{\mathrm c}}(j)\le
 \begin{cases}
 \eta,&j\in B,\\
 2(2r+1)^D+2C_W+\eta,&j\in R_B\setminus B,\\
 1,&j\notin R_B.
 \end{cases}
\]
Each site belongs to $\ell^D$ labelled blocks in
\eqref{eq:classical-blocks}, and the normalized number of sets
$R_B\setminus B$ containing it is at most $(1+2h/\ell)^D-1$.
With $\chi=\sup_j\sum_{B:j\in R_B}a_B\le(1+2h/\ell)^D$, summing
the preceding influence bounds gives
\begin{equation}
 \min_j\sum_Ba_B[1-D_{\Phi_B^{\mathrm c}}(j)]
 \ge 1-[2(2r+1)^D+2C_W-1][(1+2h/\ell)^D-1]-\eta\chi.
 \label{qp:ssm-choice}
\end{equation}
Choose $\ell$ so the boundary term is at most $1/4$, then take
$\eta=[4(1+2h/\ell)^D]^{-1}$. The last term is at most $1/4$,
which proves \eqref{qp:classical-reference-certificate}.
\end{proof}

}

\section{One-dimensional mixing at every finite temperature}
\label{q1:section}

Throughout this section, $\Lambda=[N]=\{1,\ldots,N\}$ is a
chain with distance $\dist(i,j)=|i-j|$ and local Hilbert spaces
$\HH_i\simeq\CC^{d_i}$, where $2\le d_i\le d$.
We consider the nearest-neighbour Hamiltonian
\begin{equation}
 H=\sum_{i=1}^{N}V_i+\sum_{i=1}^{N-1}W_i.
 \label{q1:H}
\end{equation}
Here $V_i=V_i^*$ acts on site $i$, and $W_i=W_i^*$ acts on
sites $i,i+1$. The on-site terms $V_i$ may have arbitrary finite norms.
In particular,
$\|W_i\|\le J_{\mathrm{int}}$ for every $1\le i<N$.
Constants denoted by $C_d,A_d$ may change from line to line and depend
only on $d$; $c,C$ denote universal positive constants.

      The main result of  this section,
\Cref{q1:main},
 establishes rapid
mixing of the smoothed heat-bath dynamics  at every finite inverse temperature.
Our Gibbs sampler is defined as follows: for $\ell\ge2$, consider the labelled family of blocks
\[
 \mathscr W_\ell=(B_a)_{a=2-\ell}^{N},\qquad
 B_a=[a,a+\ell-1]\cap[N].
\]
We retain distinct labels when intervals coincide. Every site belongs
to exactly $\ell$ labelled blocks.
For an evolution time $T>0$, define the block channels and their dynamics by
\begin{equation}
 \Phi_B=e^{T\mathcal K_B},\qquad
 \mathcal L=\frac{1}{\ell}\sum_{B\in\mathscr W_\ell}(\Phi_B-\id).
 \label{q1:outer}
\end{equation}

  \begin{theorem}[All-temperature chain mixing]
 \label{q1:main}

  For every $\beta>0$, there are a block length $\ell\ge2$ and an
evolution time $T>0$, depending only on $\beta,J_{\mathrm{int}},d$,
such that $(\Phi_B)_{B\in\mathscr W_\ell}$ satisfies the quantum DS
condition
\eqref{eq:qds}
with $\kappa=1/2$.
These parameters can be chosen with
\begin{equation}
 \ell\le C_d(1+\beta J_{\mathrm{int}})^8e^{8\beta J_{\mathrm{int}}},\qquad
 T\le C_de^{2\beta J_{\mathrm{int}}}\log(2+\beta J_{\mathrm{int}}).
 \label{q1:parameters}
\end{equation}
 Then $\sigma_{\beta,H}$ is the unique stationary state of $\mathcal L$,
and for every state $\rho$ and $t\ge0$,
\begin{equation}
 \|e^{t\mathcal L}(\rho)-\sigma_{\beta,H}\|_1
 \le 2N e^{-t/2}.
 \label{q1:mixing}
\end{equation}

 \end{theorem}
\noindent  To establish Theorem \ref{q1:main}, we combine   decay of correlations recalled in
Section \ref{qc:section}  with locality bounds for the
block dynamics, proved in Section~\ref{qg:collars},  to verify the
quantum DS condition. To obtain explicit
temperature dependence, we use  trace-norm contractivity    to bound
each marginal error before summing over distances. We choose the evolution time $T$
large enough to suppress the residual interior influence, independently
of the block length $\ell$; we then increase $\ell$ until the total
interior contraction dominates the boundary influence .

\subsection{
 Decay of correlations with arbitrary on-site terms }
\label{qc:section}

For a state $\rho$ and observables $A,B$, we define their
covariance by
\[
 \Cov_\rho(A,B):=\Tr(\rho AB)-\Tr(\rho A)\Tr(\rho B).
\]
Exponential clustering of covariance for translation-invariant infinite quantum
chains goes back to Araki \cite{Araki1969}; later extensions of
complex-time locality appear in \cite{PGPH}.
Our mixing proof below needs a summable covariance profile valid in every
interval restriction. The following theorem combines QBP \cite{HastingsQBP} with the
inclusion--exclusion strategy of
Kimura--Kuwahara \cite{KKclustering} to obtain such a profile with
explicit temperature dependence. Full proofs are postponed to Appendix~\ref{app:clustering}.

\begin{theorem}
\label{qc:clustering}

There are universal constants
$C_*,c_*>0$ such that, for every $0<\beta<\infty$ and observables
$A,B$ supported on disjoint intervals at distance $R$,
\begin{equation}
 \abs{\Cov_{\sigma_\beta}(A,B)}
 \le C_*\norm{A}\norm{B}
       \exp\!\left[-c_*\frac{e^{-4\beta J_{\mathrm{int}}}}{1+\beta J_{\mathrm{int}}}\sqrt
R\right].
 \label{qc:covariance}
\end{equation}
The constants do not depend on $N$, the lengths of the observable
supports, or the local Hilbert-space dimensions.
The same estimate, with the same constants, holds for the Gibbs state
of the Hamiltonian restricted to any subinterval.
\end{theorem}

\begin{remark}[Exponential decay and dependence on on-site terms]
\label{qc:two-profiles}
The exponent $\sqrt R$ in \eqref{qc:covariance} comes from the
localization errors in our proof. In \eqref{qc:cut-bound}, these errors
are amplified exponentially in the number of bonds removed between
the observable supports. We control them by taking the separation
proportional to this number, so only $O(\sqrt R)$ such bonds fit in
the intervening region. This does not rule out exponential decay
with constants independent of the on-site norms.
For uniformly bounded local terms, exponential clustering is
known for translation-invariant infinite chains \cite{Araki1969}
and for finite chains without translation invariance
\cite[Corollary~I.2]{BC}. Applying the latter result after including
the on-site terms in adjacent bonds gives constants depending on
$\sup_i\|V_i\|$. Our mixing proof uses only \Cref{qc:clustering},
which avoids this dependence and gives explicit temperature bounds.
\end{remark}

\subsection{Local block dynamics and boundary estimates}
\label{qg:collars}

We first approximate each block generator by a generator supported on the block and a
neighbourhood of its boundary. We bound the QBP operator, the jump map, and the coherent
correction separately, then combine these estimates to control the block channels in
transport norm. The block replacement channel $\mathcal R_B$ is kept unchanged; only the
operators coupling the block to its complement are localized.

\subsubsection{Boundary neighbourhoods and the QBP operators}

For an interval $B=[a,b]\cap\ZZ\subseteq[N]$, let
\begin{equation}
 \partial B=\{a-1,b\}\cap\{1,\ldots,N-1\}.
 \label{q1:cuts}
\end{equation}
Each $q\in\partial B$ labels the bond $\{q,q+1\}$ at an end of $B$.
For an integer $u\ge0$, define the boundary
neighbourhood $C_u(B)$ of radius $u$ around these bonds by
\begin{equation}
 C_u(B)=[N]\cap\bigcup_{q\in\partial B}[q-u,q+u+1],
 \qquad I_u(B)=B\setminus C_u(B).
 \label{qg:collar-def}
\end{equation}
Thus $C_u(B)$ consists of sites at distance at most $u$ from a crossing bond, and
\[
 |C_u(B)|\le4(u+1),\qquad |B\cap C_u(B)|\le2(u+1),\qquad
 B\cup C_u(B)=B^{\le u+1}.
\]
The radius $u$ is used only to estimate the exact block channels. If $B=[N]$, then
$C_u(B)=\varnothing$ and $\mathcal K_B=\mathcal R_B-\id$; all localized maps below agree
with the exact maps.

Recall the QBP operator $O_B(\theta)$, $\theta\in[0,1]$, from \eqref{qg:ODE} and the crossing
interaction $W_{\partial B}$ from \eqref{eq:crossing-interaction}.
With the bond labels just defined,
\[
 W_{\partial B}=\sum_{q\in\partial B}W_q,\qquad
 \|W_{\partial B}\|\le\sum_{q\in\partial B}\|W_q\|
 \le2J_{\mathrm{int}}.
\] Let
$O_{B,u}(\theta)$ solve the same equation with $H_B(\theta)$ replaced by
$H_{C_u(B)}-(1-\theta)W_{\partial B}$, and set $O_{B,u}=O_{B,u}(1)$. Both $O_{B,u}$ and its
inverse are supported on $C_u(B)$.

\begin{lemma}[Bounds on the localized QBP operators]
\label{q1:boundary-operators}
Set $M=e^{\beta J_{\mathrm{int}}/2}$ and $\lambda=c/(1+\beta J_{\mathrm{int}})$, with
$c>0$ as in \Cref{qg:transform-locality}, decreased if necessary so that $\lambda\le1$.
Then
\[
 \|O_B^{\pm1}\|,\ \|O_{B,u}^{\pm1}\|\le M,
\]
and
\begin{equation}
 \|O_B-O_{B,u}\|+\|O_B^{-1}-O_{B,u}^{-1}\|
 \le C\beta J_{\mathrm{int}}e^{\beta J_{\mathrm{int}}/2}e^{-\lambda u}.
 \label{qg:O-tail}
\end{equation}
\end{lemma}
\begin{proof}
The norm bounds follow from \Cref{qg:congruence} and $\|W_{\partial
B}\|\le2J_{\mathrm{int}}$. Write
\begin{align*}
 A_B(\theta)&=-\frac\beta4\mathcal F_{H_B(\theta)}(W_{\partial B}),\\
 A_{B,u}(\theta)&=-\frac\beta4
 \mathcal F_{H_{C_u(B)}-(1-\theta)W_{\partial B}}(W_{\partial B}).
\end{align*}
Applying \Cref{lem:qbp-kernels,qg:transform-locality} to the at most two crossing bonds
gives
\[
 \|A_B(\theta)\|,\|A_{B,u}(\theta)\|\le\beta J_{\mathrm{int}}/2,
 \qquad
 \|A_B(\theta)-A_{B,u}(\theta)\|\le C\beta J_{\mathrm{int}}e^{-\lambda u}.
\]
For $\theta,\theta'\in[0,1]$, define
\[
 U_B(\theta',\theta)=O_B(\theta')O_B(\theta)^{-1},\qquad
 U_{B,u}(\theta',\theta)=O_{B,u}(\theta')O_{B,u}(\theta)^{-1}.
\]
Equation~\eqref{qg:O-interval} yields $\|U_B(\theta',\theta)\|,\|U_{B,u}(\theta',\theta)\| \le
e^{\beta
J_{\mathrm{int}}|\theta'-\theta|/2}$. Variation of constants gives
\begin{align*}
 O_B-O_{B,u}
 &=\int_0^1U_B(1,\theta)[A_B(\theta)-A_{B,u}(\theta)]U_{B,u}(\theta,0)\,d\theta,\\
 O_B^{-1}-O_{B,u}^{-1}
 &=-\int_0^1U_B(0,\theta)[A_B(\theta)-A_{B,u}(\theta)]U_{B,u}(\theta,1)\,d\theta.
\end{align*}
Each integral is bounded by
\[
 C\beta J_{\mathrm{int}}e^{-\lambda u}
 \int_0^1e^{\beta J_{\mathrm{int}}\theta/2}
             e^{\beta J_{\mathrm{int}}(1-\theta)/2}\,d\theta
 =C\beta J_{\mathrm{int}}e^{\beta J_{\mathrm{int}}/2}e^{-\lambda u}.
\]
Adding these bounds proves \eqref{qg:O-tail}.
\end{proof}

\subsubsection{Locality of the jump map}

Recall from \eqref{qg:generator} that $\mathcal T_B=\operatorname{Ad}_{O_B}\circ\mathcal
R_B \circ\operatorname{Ad}_{O_B^{-1}}$ and $S_B=\mathcal T_B^*(\one)$. Define their
localized versions by
\[
 \mathcal T_{B,u}=\operatorname{Ad}_{O_{B,u}}\circ\mathcal R_B
                  \circ\operatorname{Ad}_{O_{B,u}^{-1}},
 \qquad S_{B,u}=\mathcal T_{B,u}^*(\one).
\]

\begin{lemma}[Bounds on the jump map]
\label{q1:jump-locality}
The map $\mathcal T_{B,u}$ is supported on $B^{\le u+1}$,
and $S_{B,u}$ is supported on $C_u(B)$.

With $M,\lambda$ from \Cref{q1:boundary-operators},
\begin{equation}
 M^{-4}\one\le S_B,S_{B,u}\le M^4\one,
 \qquad \|\mathcal T_B\|_\diamond,\|\mathcal T_{B,u}\|_\diamond\le M^4.
 \label{qg:chain-loss}
\end{equation}
Moreover,
\begin{equation}
 \|\mathcal T_B-\mathcal T_{B,u}\|_\diamond+\|S_B-S_{B,u}\|
 \le C\beta J_{\mathrm{int}}e^{2\beta J_{\mathrm{int}}}e^{-\lambda u}.
 \label{qg:T-tail}
\end{equation}
\end{lemma}
\begin{proof}
The support of $\mathcal T_{B,u}$ follows from its definition:
$\mathcal R_B$ acts only on $B$, and $O_{B,u}^{\pm1}$ are supported
on $C_u(B)$, so their composition acts only on
$B\cup C_u(B)=B^{\le u+1}$. For $S_{B,u}$, the adjoint formula gives
\begin{align*}
 S_{B,u}
 &=O_{B,u}^{-*}\mathcal R_B^*(O_{B,u}^*O_{B,u})O_{B,u}^{-1}\\
 &=O_{B,u}^{-*}\bigl[\one_B\otimes
       \Tr_B(\sigma_\beta^B O_{B,u}^*O_{B,u})\bigr]O_{B,u}^{-1}.
\end{align*}
Since $O_{B,u}^*O_{B,u}$ acts trivially on $I_u(B)=B\setminus C_u(B)$,
tracing first over $I_u(B)$ replaces $\sigma_\beta^B$ by its marginal
on $B\cap C_u(B)$. Tracing over the latter region then gives an
operator supported on $C_u(B)\setminus B$. Thus
$\mathcal R_B^*(O_{B,u}^*O_{B,u})$ is supported on $C_u(B)$.
Multiplying by $O_{B,u}^{-*}$ and $O_{B,u}^{-1}$ preserves this
support, proving the claim for $S_{B,u}$.

The operator inequalities follow from \Cref{qg:loss} and \Cref{q1:boundary-operators}.
Since the jump maps are completely positive, their diamond norms equal $\|S_B\|$ and
$\|S_{B,u}\|$. For any operators $A,A'$,
\[
 \|\operatorname{Ad}_A-\operatorname{Ad}_{A'}\|_\diamond
 \le(\|A\|+\|A'\|)\|A-A'\|.
\]
Using $\|\mathcal R_B\|_\diamond=1$ and changing one conjugating factor at a time, we
obtain
\begin{align*}
 \|\mathcal T_B-\mathcal T_{B,u}\|_\diamond
 &\le M^2\|\operatorname{Ad}_{O_B}-\operatorname{Ad}_{O_{B,u}}\|_\diamond
    +M^2\|\operatorname{Ad}_{O_B^{-1}}-
                    \operatorname{Ad}_{O_{B,u}^{-1}}\|_\diamond\\
 &\le2M^3(\|O_B-O_{B,u}\|+\|O_B^{-1}-O_{B,u}^{-1}\|),\\
 \|S_B-S_{B,u}\|
 &\le\|\mathcal T_B-\mathcal T_{B,u}\|_\diamond.
\end{align*}
Substitution of \eqref{qg:O-tail} proves \eqref{qg:T-tail}.
\end{proof}

\subsubsection{Locality of the coherent correction}

Recall $G_B=\mathcal G_H(S_B)$ from \eqref{qg:generator} and set $G_{B,u}=\mathcal
G_{H_{C_u(B)}}(S_{B,u})$, which is supported on $C_u(B)$. We estimate this operator
using the multiplier bound of \Cref{qg:multiplier} and the real-time locality estimate
\eqref{qg:allT-G}.

\begin{lemma}
\label{q1:coherent-locality}
With $\lambda=c/(1+\beta J_{\mathrm{int}})$ as in
\Cref{q1:boundary-operators},
\begin{align*}
 \|G_B\|,\|G_{B,u}\|
 &\le C_de^{2\beta J_{\mathrm{int}}}(1+\beta J_{\mathrm{int}})^4
                  \log(2+\beta J_{\mathrm{int}}),\\
 \|G_B-G_{B,u}\|
 &\le C_de^{2\beta J_{\mathrm{int}}}(1+\beta J_{\mathrm{int}})^4
                  \log(2+\beta J_{\mathrm{int}})e^{-\lambda u/2}.
\end{align*}
\end{lemma}
\begin{proof}
Set $D_0=S_{B,0}-\one$ and $D_m=S_{B,m}-S_{B,m-1}$ for $m\ge1$. By \eqref{qg:chain-loss}
and \eqref{qg:T-tail},
\begin{align*}
 \|D_0\|&\le2e^{2\beta J_{\mathrm{int}}},\\
 \|D_m\|&\le\|S_{B,m}-S_B\|+\|S_B-S_{B,m-1}\|\\
 &\le C\beta J_{\mathrm{int}}e^{2\beta J_{\mathrm{int}}}
                         (1+e^\lambda)e^{-\lambda m}.
\end{align*}
Consequently, for all $m\ge0$,
\[
 \operatorname{supp}D_m\subseteq C_m(B),\qquad
 \|D_m\|\le Ce^{2\beta J_{\mathrm{int}}}(1+\beta J_{\mathrm{int}})e^{-\lambda m}.
\]
As $\mathcal G_K(\one)=0$ and $|C_0(B)|\le4$, \Cref{qg:multiplier} gives $\|G_{B,0}\|\le
C_de^{2\beta J_{\mathrm{int}}}$. For $u\ge1$, linearity gives
\begin{align}
 G_{B,u}-G_{B,u-1}
 ={}\mathcal G_{H_{C_u(B)}}(D_u)+\sum_{m=0}^{u-1}
 [\mathcal G_{H_{C_u(B)}}-\mathcal G_{H_{C_{u-1}(B)}}](D_m).
 \label{qg:G-boundary-identity}
\end{align}
The multiplier bound controls the first term by $C_d(u+1)\|D_u\|$. For the $m$th
summand, apply \eqref{qg:allT-G} with support size at most $4(m+1)$ and distance at
least $\max\{u-m-2,0\}$. Thus
\begin{align*}
 \|G_{B,u}-G_{B,u-1}\|
 &\le C_d(u+1)\|D_u\|
 +C\log(2+\beta J_{\mathrm{int}})
       \sum_{m=0}^{u-1}(m+1)\|D_m\|e^{-\lambda\max\{u-m-2,0\}}\\
 &\le C_de^{2\beta J_{\mathrm{int}}}(1+\beta J_{\mathrm{int}})
       \log(2+\beta J_{\mathrm{int}})e^{-\lambda u}
       \left[(u+1)+\sum_{m=0}^{u-1}(m+1)\right].
\end{align*}
Since the bracket is at most $(u+1)^2$,
\begin{equation}
 \|G_{B,u}-G_{B,u-1}\|
 \le C_de^{2\beta J_{\mathrm{int}}}(1+\beta J_{\mathrm{int}})
       \log(2+\beta J_{\mathrm{int}})(u+1)^2e^{-\lambda u}.
 \label{qg:G-boundary}
\end{equation}
For a proper block in a finite chain, $C_u(B)=[N]$ for all large $u$, and then
$G_{B,u}=G_B$. The estimates
\begin{align*}
 \sum_{m\ge0}(m+1)^2e^{-\lambda m}
 &\le C(1+\beta J_{\mathrm{int}})^3,\\
 \sum_{m\ge u}(m+1)^2e^{-\lambda m}
 &\le e^{-\lambda u/2}\sum_{m\ge0}(m+1)^2e^{-\lambda m/2}
 \le C(1+\beta J_{\mathrm{int}})^3e^{-\lambda u/2}
\end{align*}
therefore prove the norm and tail bounds by telescoping \eqref{qg:G-boundary}. The case
$B=[N]$ has $G_B=G_{B,u}=0$.
\end{proof}

\subsubsection{Locality of block dynamics}

Next, we combine \Cref{q1:jump-locality,q1:coherent-locality} to control the localized
Lindblad generators and their associated dynamics,
\begin{equation}
 \mathcal K_{B,u}(X)=\mathcal T_{B,u}(X)
       -\tfrac12\{S_{B,u},X\}-i[G_{B,u},X],
 \qquad \Phi_{B,u}(T)=e^{T\mathcal K_{B,u}}.
 \label{qg:truncated-generator}
\end{equation}
These maps are supported on $B^{\le u+1}$, while $O_{B,u}^{\pm1},S_{B,u},G_{B,u}$ are
supported on $C_u(B)$.

\begin{proposition}
\label{qg:chain-locality}
With the same $\lambda$ and a constant $C_d<\infty$,
\begin{align}
 \|\mathcal K_B\|_\diamond,\ \|\mathcal K_{B,u}\|_\diamond
 &\le q_\beta:=C_de^{2\beta J_{\mathrm{int}}}(1+\beta J_{\mathrm{int}})^4
                       \log(2+\beta J_{\mathrm{int}}),\label{qg:chain-norm}\\
 \|\mathcal K_B-\mathcal K_{B,u}\|_\diamond
 &\le q_\beta e^{-\lambda u/2},\label{qg:chain-tail}\\
 \|\Phi_B(T)-\Phi_{B,u}(T)\|_\diamond
 &\le q_\beta T e^{-\lambda u/2}.
 \label{qg:channel-diamond}
\end{align}
The same bounds hold after tensoring the maps with the identity on an
arbitrary finite-dimensional reference system.
\end{proposition}
\begin{proof}
The preceding two lemmas give
\begin{align*}
 \|\mathcal K_B\|_\diamond
 &\le\|\mathcal T_B\|_\diamond+\|S_B\|+2\|G_B\|\le q_\beta,\\
 \|\mathcal K_B-\mathcal K_{B,u}\|_\diamond
 &\le\|\mathcal T_B-\mathcal T_{B,u}\|_\diamond
       +\|S_B-S_{B,u}\|+2\|G_B-G_{B,u}\|
 \le q_\beta e^{-\lambda u/2}.
\end{align*}
The same norm bound holds for $\mathcal K_{B,u}$. Duhamel's formula and diamond-norm
contractivity yield
\begin{align*}
 \Phi_B(T)-\Phi_{B,u}(T)
 &=\int_0^T e^{(T-t)\mathcal K_B}
       (\mathcal K_B-\mathcal K_{B,u})e^{t\mathcal K_{B,u}}\,dt,\\
 \|\Phi_B(T)-\Phi_{B,u}(T)\|_\diamond
 &\le T\|\mathcal K_B-\mathcal K_{B,u}\|_\diamond
 \le q_\beta T e^{-\lambda u/2}.
\end{align*}
\end{proof}

\noindent We next convert the diamond-norm approximation from
\Cref{qg:chain-locality} into a transport-norm bound. This will allow
\Cref{q1:influence-section} to transfer influence bounds for the localized
channel $\Phi_{B,u}(T)$ to those of the original channel $\Phi_B(T)$.

\begin{proposition}
\label{qg:chain-transport}
With $\lambda$ from \Cref{qg:chain-locality}, for every
Hermitian $X$ and $|B|\le\ell$,
\begin{equation}
 \|[\Phi_B(T)-\Phi_{B,u}(T)](X)\|_{W_1}
 \le K_\beta T(\ell+u+2)e^{-\lambda u/2}\|X\|_1,
 \label{qg:channel-transport}
\end{equation}
where $K_\beta=C_de^{2\beta J_{\mathrm{int}}}(1+\beta J_{\mathrm{int}})^6 \log(2+\beta
J_{\mathrm{int}})$. 
\end{proposition}
\begin{proof}
We telescope the localized channels:
\[
 \Phi_B(T)-\Phi_{B,u}(T)
 =\sum_{m\ge u}[\Phi_{B,m+1}(T)-\Phi_{B,m}(T)].
\]
This sum is finite because the localized maps eventually coincide with the exact map.
Each difference is supported on $B^{\le m+2}$, a set of at most $\ell+4(m+2)$ sites, and
annihilates the partial trace over that set. By \Cref{lem:transport} and
\eqref{qg:channel-diamond},
\begin{align*}
 \|[\Phi_B(T)-\Phi_{B,u}(T)](X)\|_{W_1}
 &\le\sum_{m\ge u}[\ell+4(m+2)]
       \|\Phi_{B,m+1}(T)-\Phi_{B,m}(T)\|_\diamond\|X\|_1\\
 &\le Cq_\beta T\|X\|_1
                 \sum_{m\ge u}(\ell+m+2)e^{-\lambda m/2}.
\end{align*}
The geometric sums give
\begin{align*}
 \sum_{m\ge u}(\ell+m+2)e^{-\lambda m/2}
 &=e^{-\lambda u/2}\left[\frac{\ell+u+2}{1-e^{-\lambda/2}}
                    +\frac{e^{-\lambda/2}}{(1-e^{-\lambda/2})^2}\right]\\
 &\le C(1+\beta J_{\mathrm{int}})
             (\ell+u+1+\beta J_{\mathrm{int}})e^{-\lambda u/2}\\
 &\le C(1+\beta J_{\mathrm{int}})^2(\ell+u+2)e^{-\lambda u/2}.
\end{align*}

\end{proof}

\subsection{Influence of a block update}
\label{q1:influence-section}

Our next goal is to bound $D_{\Phi_B(T)}(j)$, the influence of the block
update $\Phi_B(T)$ at site $j$, defined in \eqref{eq:influence}.
These are the coefficients summed in the quantum DS condition \eqref{eq:qds}.
The estimates below will allow us to show in \Cref{q1:counting} that contraction
in block interiors outweighs the possible growth near block boundaries.
The approximation bound \eqref{qg:channel-transport} from
\Cref{qg:chain-transport} reduces the task to estimating the influence of the
localized channel $\Phi_{B,u}(T)$. We proceed in two steps:
\Cref{q1:two-collars} controls correlations between the ends of the block
Gibbs state $\sigma_\beta^B$ and its interior; \Cref{q1:influence} uses this
to bound the localized channel's influence according to the position of $j$.

\subsubsection{From scalar covariances to a trace-norm bound}

The following lemma converts the scalar clustering bound of
\Cref{qc:clustering} into a trace-norm estimate for correlations between
the ends of a block and its interior.

For an interval $B\subseteq\Lambda$, recall the Gibbs state $\sigma_\beta^B$
of $H_B$ from \eqref{eq:restricted-hamiltonian}. We write
$\sigma_\beta^B(O):=\Tr_{\HH_B}(\sigma_\beta^B O)$ for the expectation of
$O\in\BB(\HH_B)$, and $(\sigma_\beta^B)_U=\Tr_{B\setminus U}\sigma_\beta^B$
for the marginal on $U\subseteq B$, following the convention of
Section~\ref{sec:prelim}. With this notation, \Cref{qc:clustering} gives
\begin{equation}
 \begin{gathered}
 |\sigma_\beta^B(AY)-\sigma_\beta^B(A)\sigma_\beta^B(Y)|
 \le C_{\rm cl}\|A\|\|Y\|f(\dist(A,Y)),
 \\
 f(v)=e^{-c_{\rm cl}\sqrt v},\quad
 c_{\rm cl}\ge c e^{-4\beta J_{\mathrm{int}}}/(1+\beta J_{\mathrm{int}}),
 \end{gathered}
 \label{q1:clustering}
\end{equation}
where $C_{\rm cl}$ is a numerical constant.

\begin{lemma}[From interval covariances to a joint boundary marginal]
\label{q1:two-collars}
Suppose $\sigma_\beta^B$ satisfies \eqref{q1:clustering} with a nonincreasing
profile $f$. Let $A,D,C\subseteq B$ be intervals in this
left-to-right order, with either end interval allowed to be empty. If $D$ is at distance
at least $v$ from both nonempty end intervals, then
\begin{equation}
 \|(\sigma_\beta^B)_{A\cup C\cup D}-(\sigma_\beta^B)_{A\cup C}\otimes(\sigma_\beta^B)_D\|_1
 \le 3C_{\rm cl}d^{2(|A|+|C|)}f(v).
 \label{q1:trace-clustering}
\end{equation}
\end{lemma}
\begin{proof}
For norm-one operators $L_A,Y_D,R_C$, write
\begin{align*}
 \sigma_\beta^B(L_AY_DR_C)-\sigma_\beta^B(L_AR_C)\sigma_\beta^B(Y_D)
 ={}&\sigma_\beta^B(L_AY_DR_C)-\sigma_\beta^B(L_A)\sigma_\beta^B(Y_DR_C)\\
 &+\sigma_\beta^B(L_A)[\sigma_\beta^B(Y_DR_C)-\sigma_\beta^B(Y_D)\sigma_\beta^B(R_C)]\\
 &-[\sigma_\beta^B(L_AR_C)-\sigma_\beta^B(L_A)\sigma_\beta^B(R_C)]\sigma_\beta^B(Y_D).
\end{align*}
The first covariance is between $A$ and the interval spanning $D\cup C$; the other two
are also between separated intervals. Each is bounded by $C_{\rm cl}f(v)$. With an empty
end interval the same conclusion follows by taking its observable to be the scalar one.

We now turn this scalar covariance bound into the trace-norm bound. Set
$\Delta:=(\sigma_\beta^B)_{A\cup C\cup D}-(\sigma_\beta^B)_{A\cup C}\otimes(\sigma_\beta^B)_D$
and order its tensor
factors as $A,C,D$. Choose orthonormal bases $\{|a\rangle_A\}$ and $\{|c\rangle_C\}$,
and write their product basis as $|p\rangle=|a\rangle_A\otimes|c\rangle_C$, with
$p=(a,c)$. Expanding only the $A,C$ factors gives
 \[
 \Delta=\sum_{p,p'}|p\rangle\langle p'|\otimes\Delta_{pp'},\qquad
 \Delta_{pp'}:=(\langle p|\otimes\one_D)\Delta(|p'\rangle\otimes\one_D).
 \]
Thus each coefficient $\Delta_{pp'}$ is an operator on $\HH_D$. For $p=(a,c)$ and
$p'=(a',c')$, take $L_A=|a'\rangle\langle a|$ and $R_C=|c'\rangle\langle c|$. Both have
operator norm one, and $L_A\otimes R_C=|p'\rangle\langle p|$. For any operator $Y_D$
with $\norm{Y_D}\le1$, the definition of the marginals gives
 \begin{align*}
 \Tr(Y_D\Delta_{pp'})
 &=\Tr\bigl[(|p'\rangle\langle p|\otimes Y_D)\Delta\bigr]\\
 &=\sigma_\beta^B(L_A Y_D R_C)-\sigma_\beta^B(L_A R_C)\sigma_\beta^B(Y_D).
 \end{align*}
This is precisely the covariance bounded above. The coefficients $\Delta_{pp'}$ need not
be Hermitian, so we use trace-norm duality over all operators $Y_D$:
 \[
 \nOne{\Delta_{pp'}}
 =\sup_{\norm{Y_D}\le1}|\Tr(Y_D\Delta_{pp'})|
 \le3C_{\rm cl}f(v).
 \]
Each of $p,p'$ ranges over at most $d^{|A|+|C|}$ basis vectors, giving at most
$d^{2(|A|+|C|)}$ coefficients. Since each matrix unit $|p\rangle\langle p'|$ has trace
norm one and the trace norm multiplies on tensor products, the triangle inequality
yields
 \begin{align*}
 \nOne\Delta
 &\le\sum_{p,p'}\nOne{|p\rangle\langle p'|\otimes\Delta_{pp'}}
 =\sum_{p,p'}\nOne{\Delta_{pp'}}\\
 &\le3C_{\rm cl}d^{2(|A|+|C|)}f(v),
 \end{align*}
which proves \eqref{q1:trace-clustering}.
\end{proof}

\subsubsection{A uniform estimate on localized influences}

Write $D_{B,u}(j)$ and $D_B(j)$ for the influences of $\Phi_{B,u}(T)$ and
$\Phi_B(T)$ defined in \eqref{eq:influence}. We now use \Cref{q1:two-collars}
to estimate the influence of the localized channel. The definition requires a
bound for every Hermitian operator $X$ with $\Tr_jX=0$. The following lemma
separates three cases: a site in the interior $I_u(B)$, in the boundary
neighbourhood $C_u(B)$, or outside $B^{\le u+1}$. The first Duhamel identity
in \eqref{eq:duhamelidentities} gives interior contraction. For $j\in C_u(B)$,
we combine clustering with the second Duhamel identity, applied to the local
generator $\mathcal K_{B,u}$ from \eqref{qg:truncated-generator}. Finally,
locality gives influence one outside the support. The key point is that the
boundary influence is bounded independently of $|B|$, apart from a term
that decays exponentially in $T$. This will allow the interior contraction
to dominate when summing over large blocks.

\begin{lemma}[Uniform localized influences]
 \label{q1:influence}
With $M$ from \Cref{q1:boundary-operators}, set
\begin{equation}
 A_u:=8(u+1)+4\sum_{v=0}^\infty
 \min\{1,3M^8C_{\rm cl}d^{4(u+1)}f(v)\}.
 \label{q1:Ar}
\end{equation}
For every $T\ge0$, interval $B\subseteq[N]$, integer $u\ge0$, and site $j$,
\begin{equation}
 D_{B,u}(j)\le
 \begin{cases}
 e^{-M^{-4}T},&j\in I_u(B),\\
 A_u+4|B|e^{-M^{-4}T},&j\in C_u(B),\\
 1,&j\notin B^{\le u+1}.
 \end{cases}
 \label{q1:influence-bound}
\end{equation}
For the profile in \eqref{q1:clustering}, we also have
\begin{equation}
 A_u\le C_d\bigl[(u+1)+(1+\beta J_{\mathrm{int}})^2e^{8\beta
J_{\mathrm{int}}}(u+\beta J_{\mathrm{int}}+1)^2\bigr].
 \label{q1:Ar-bound}
\end{equation}
\end{lemma}
\begin{proof}
Fix $B,u$ and a Hermitian operator $X$ with $\Tr_jX=0$.
We apply \Cref{lem:trajectories} to the localized generator
$\mathcal K_{B,u}$ from \eqref{qg:truncated-generator}. Explicitly, define
\begin{align*}
 \Phi_t&:=e^{t\mathcal K_{B,u}}=\Phi_{B,u}(t),\qquad
 Z_t:=\exp[-t(S_{B,u}/2+iG_{B,u})],\\
 \mathcal N_t(Y)&:=Z_tYZ_t^*,\qquad
 \mathcal J_T:=\Phi_T-\mathcal N_T.
\end{align*}
Thus $\mathcal N_t$ is the semigroup generated by
$\mathcal K_{B,u}-\mathcal T_{B,u}$. The lower bound
$S_{B,u}\ge M^{-4}\one$ in \eqref{qg:chain-loss} and Hermiticity of $G_{B,u}$
give, by \Cref{lem:trajectories},
\[
 \|Z_t\|^2\le e^{-M^{-4}t},\qquad
 \|\mathcal N_t\|_\diamond\le e^{-M^{-4}t}.
\]
The same lemma shows that $\Phi_t$ is a channel and $\mathcal J_T$ is
completely positive, and gives the two Duhamel identities
\begin{equation}
 \mathcal J_T=\int_0^T\Phi_{T-t}\mathcal T_{B,u}\mathcal N_t\,dt
 =\int_0^T\mathcal N_{T-t}\mathcal T_{B,u}\Phi_t\,dt.
 \label{q1:first-last-jump}
\end{equation}
Moreover, $O_{B,u}^{\pm1}$ and $S_{B,u}$ are supported on $C_u(B)$ by
\Cref{q1:boundary-operators,q1:jump-locality}, as is $G_{B,u}$ by definition.
The formula for $Z_t$ gives the same support for $Z_t$.

\proofstep{Interior influence: $j\in I_u(B)$}
Since $O_{B,u}^{-1}$ and $Z_t$ act trivially on $j$,
\[
 \Tr_B(O_{B,u}^{-1}\mathcal N_t(X)O_{B,u}^{-*})=0.
\]
Thus $\mathcal T_{B,u}\mathcal N_t(X)=0$, and the first Duhamel identity
in \eqref{q1:first-last-jump} gives $\mathcal J_T(X)=0$.
Also $\Tr_j\mathcal N_T(X)=0$. The one-site identity in
\Cref{lem:transport}, applied to $\mathcal N_T(X)=\Phi_T(X)$, therefore gives
\[
 \|\Phi_T(X)\|_{W_1}=\tfrac12\|\mathcal N_T(X)\|_1
 \le\tfrac12 e^{-M^{-4}T}\|X\|_1.
\]
Dividing by $\|X\|_1/2$ and taking the supremum over $X$ proves the first case.

\proofstep{Boundary influence: $j\in C_u(B)$}
Write $(\sigma_\beta^B)_{I_u(B)}=\Tr_{B\setminus I_u(B)}\sigma_\beta^B$, and let
\[
 \mathcal R_{I_u(B)}(Y)=(\sigma_\beta^B)_{I_u(B)}\otimes\Tr_{I_u(B)}Y,
 \qquad
 \widehat\Phi_T=\mathcal J_T+\mathcal R_{I_u(B)}\mathcal N_T.
\]
For $I_u(B)=\varnothing$, take $\mathcal R_{I_u(B)}=\id$.
Then $\widehat\Phi_T$ is completely positive and supported on $B^{\le u+1}$;
trace preservation follows from that of $\mathcal R_{I_u(B)}$ and the channel
$\mathcal J_T+\mathcal N_T=\Phi_T$. The difference
\[
 (\Phi_T-\widehat\Phi_T)(X)
 =(\id-\mathcal R_{I_u(B)})\mathcal N_T(X)
\]
has zero partial trace over $I_u(B)$ and trace norm at most
$2e^{-M^{-4}T}\|X\|_1$. Hence \Cref{lem:transport} gives
\begin{equation}
 \|[\Phi_T-\widehat\Phi_T](X)\|_{W_1}
 \le2|I_u(B)|e^{-M^{-4}T}\|X\|_1
 \le2|B|e^{-M^{-4}T}\|X\|_1.
 \label{q1:completion}
\end{equation}

It remains to bound $\widehat\Phi_T(X)$ independently of $|B|$. Suppose first that
$I_u(B)\ne\varnothing$. The inner boundary region $C_{\rm in}=B\cap C_u(B)$ is the union
of at most two end intervals of total size at most $2(u+1)$. Let $D\subseteq I_u(B)$ be
a remaining interval at distance at least $v$ from each nonempty end interval; see
\Cref{fig:influenceregions}. Set $\zeta_u=3C_{\rm cl}d^{4(u+1)}$.
Lemma~\ref{q1:two-collars} gives

 \begin{equation}
 \|(\sigma_\beta^B)_{C_{\rm in}\cup D}
 -(\sigma_\beta^B)_{C_{\rm in}}\otimes(\sigma_\beta^B)_D\|_1
 \le\zeta_u f(v).
 \label{eq:qbp-inner-clustering}
\end{equation}
 \begin{figure}[H]
 \centering
 \begin{tikzpicture}[x=0.35cm,y=0.75cm,font=\small,
 region/.style={line width=0.5pt},
 guide/.style={gray!45,densely dashed,line width=0.4pt}]

 \draw[gray!65,line width=0.6pt] (1,0)--(31,0);
 \foreach \site in {1,...,31}
 \fill[black!75] (\site,0) circle[radius=1.2pt];
 \node[anchor=east] at (-0.4,0) {$\Lambda$};
 \node[above=4pt] at (1,0) {$1$};
 \node[above=4pt] at (7,0) {$a$};
 \node[above=4pt] at (25,0) {$b$};
 \node[above=4pt] at (31,0) {$N$};
 \draw[guide] (6.5,0.22)--(6.5,-4.12);
 \draw[guide] (25.5,0.22)--(25.5,-4.12);
 \node[anchor=east] at (-0.4,-0.75) {$B^{\le u+1}$};
 \draw[region,draw=blue!45!black,fill=blue!12]
 (3.5,-0.89) rectangle (28.5,-0.61);
 \node[anchor=east] at (-0.4,-1.35) {$B=[a,b]$};
 \draw[region,draw=blue!60!black,fill=blue!35]
 (6.5,-1.49) rectangle (25.5,-1.21);
 \node[anchor=east] at (-0.4,-1.95) {$C_u(B)$};
 \draw[region,draw=violet!75!black,fill=violet!24]
 (3.5,-2.09) rectangle (9.5,-1.81);
 \draw[region,draw=violet!75!black,fill=violet!24]
 (22.5,-2.09) rectangle (28.5,-1.81);
 \node[anchor=east] at (-0.4,-2.55) {$C_{\rm in}$};
 \draw[region,draw=orange!65!black,fill=orange!45]
 (6.5,-2.69) rectangle (9.5,-2.41);
 \draw[region,draw=orange!65!black,fill=orange!45]
 (22.5,-2.69) rectangle (25.5,-2.41);
 \node[anchor=east] at (-0.4,-3.15) {$I_u(B)$};
 \draw[region,draw=teal!65!black,fill=teal!22]
 (9.5,-3.29) rectangle (22.5,-3.01);
 \node[anchor=east] at (-0.4,-3.75) {$D$};
 \draw[region,draw=teal!65!black,fill=teal!65]
 (12.5,-3.89) rectangle (19.5,-3.61);

 \draw[gray!60,densely dotted] (9,-2.69)--(9,-4.43);
 \draw[gray!60,densely dotted] (13,-3.89)--(13,-4.43);
 \draw[gray!60,densely dotted] (19,-3.89)--(19,-4.43);
 \draw[gray!60,densely dotted] (23,-2.69)--(23,-4.43);
 \draw[<->,line width=0.5pt] (9,-4.32)--node[below=2pt] {$\ge v$} (13,-4.32);
 \draw[<->,line width=0.5pt] (19,-4.32)--node[below=2pt] {$\ge v$} (23,-4.32);
 \end{tikzpicture}
 \caption{Regions in the proof of
\Cref{q1:influence}, shown for a block away from
 the ends of the chain. All rows refer to the same sites.
 The dashed lines mark the boundary bonds of $B$;
 $C_{\rm in}=B\cap C_u(B)$, $I_u(B)=B\setminus C_u(B)$, and
 $B^{\le u+1}=B\cup C_u(B)$. The interval $D\subseteq I_u(B)$ is at distance
 at least $v$ from each component of $C_{\rm in}$.}
 \label{fig:influenceregions}
 \end{figure}
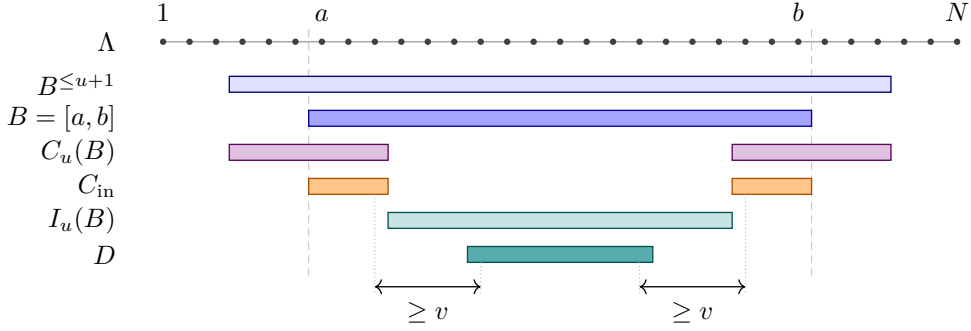

\noindent The second Duhamel identity in \eqref{q1:first-last-jump} reads
\begin{equation}
 \mathcal J_T(X)=\int_0^T
 L_{T-t}[\sigma_\beta^B\otimes Q_t(X)]L_{T-t}^*\,dt,
 \label{q1:last-jump}
\end{equation}
where
\[
 L_{T-t}=Z_{T-t}O_{B,u},\qquad
 Q_t(X)=\Tr_B(O_{B,u}^{-1}\Phi_t(X)O_{B,u}^{-*}).
\]
Here $L_{T-t}$ is supported on $C_u(B)$, since both $O_{B,u}$ and $Z_{T-t}$
have this support, as shown at the start of the proof. Their product
therefore acts trivially on $I_u(B)$. Moreover,
\[
 \|L_{T-t}\|^2\le M^2e^{-M^{-4}(T-t)},\qquad
 \|Q_t(X)\|_1\le M^2\|X\|_1.
\]
The second bound follows from $\|O_{B,u}^{-1}\|\le M$ and trace-norm
contractivity of the channel $\Phi_t$ and partial trace.
Now, since $L_{T-t}$ acts trivially on $I_u(B)$,
\begin{align*}
 &\Tr_{B^{\le u+1}\setminus D}
 \bigl(L_{T-t}[\sigma_\beta^B\otimes Q_t(X)]L_{T-t}^*\bigr)\\
 &\quad=\Tr_{C_u(B)}\bigl(L_{T-t}
 [\Tr_{I_u(B)\setminus D}(\sigma_\beta^B)\otimes Q_t(X)]L_{T-t}^*\bigr)\\
 &\quad=\Tr_{C_u(B)}\bigl(L_{T-t}
 [(\sigma_\beta^B)_{C_{\rm in}\cup D}\otimes Q_t(X)]L_{T-t}^*\bigr).
\end{align*}
Here $Q_t(X)$ acts on $B^c$; the reduced operator acts on
$D\cup(B^{\le u+1})^c$, with tensor factors ordered by site labels.

Replacing
the Gibbs marginal by its product approximation in \eqref{eq:qbp-inner-clustering} gives
\begingroup
\medmuskip=2mu
\thickmuskip=3mu
\begin{align*}
 \Tr_{C_{{u}}(B)}\bigl(L_{T-t}
 [(\sigma_\beta^B)_{C_{\rm in}}\otimes(\sigma_\beta^B)_D
 \otimes Q_t(X)]L_{T-t}^*\bigr)&=(\sigma_\beta^B)_D\otimes
 \Tr_{C_{u}(B)}\bigl(L_{T-t}
 [(\sigma_\beta^B)_{C_{\rm in}}\otimes Q_t(X)]L_{T-t}^*\bigr)\\
 &=(\sigma_\beta^B)_D\otimes\Tr_{B^{\le u+1}}
 (\mathcal N_{T-t}\mathcal T_{B,u}\Phi_t(X)).
\end{align*}
\endgroup
The first equality uses that $L_{T-t}$ acts trivially on $D$. For the second, tracing
the original integrand over $I_u(B)$ replaces $\sigma_\beta^B$ by its marginal on
$C_{\rm in}$. Multiplying the clustering estimate \eqref{eq:qbp-inner-clustering} by the
preceding norm bounds and integrating gives
\begin{align}
 &\|\Tr_{B^{\le u+1}\setminus D}\mathcal J_T(X)
 -(\sigma_\beta^B)_D\otimes
 \Tr_{B^{\le u+1}}\mathcal J_T(X)\|_1\notag\\
 &\quad\le M^4\zeta_u f(v)\|X\|_1
 \int_0^T e^{-M^{-4}(T-t)}\,dt\nonumber\\
 &\le M^8\zeta_u f(v)\|X\|_1.
 \label{q1:replacement-error}
\end{align}

We now show that the difference estimated in \eqref{q1:replacement-error}
is exactly $\Tr_{B^{\le u+1}\setminus D}\widehat\Phi_T(X)$.
First, the definition of $\mathcal R_{I_u(B)}$ and $D\subseteq I_u(B)$ give
\begin{align*}
 &\Tr_{B^{\le u+1}\setminus D}(\mathcal R_{I_u(B)}\mathcal N_T)(X)\\
 &\quad=\Tr_{I_u(B)\setminus D}((\sigma_\beta^B)_{I_u(B)})
 \otimes\Tr_{C_u(B)}\Tr_{I_u(B)}\mathcal N_T(X)\\
 &\quad=(\sigma_\beta^B)_D\otimes\Tr_{B^{\le u+1}}\mathcal N_T(X).
\end{align*}
Next, $\Tr_jX=0$ and $j\in C_u(B)\subseteq B^{\le u+1}$ imply
$\Tr_{B^{\le u+1}}X=0$. Since the channel $\Phi_T$ is supported on
$B^{\le u+1}$, it preserves the marginal on $(B^{\le u+1})^c$.
Thus $\mathcal N_T=\Phi_T-\mathcal J_T$ gives
\begin{align*}
 \Tr_{B^{\le u+1}}\mathcal N_T(X)
 &=\Tr_{B^{\le u+1}}\Phi_T(X)-\Tr_{B^{\le u+1}}\mathcal J_T(X)\\
 &=\Tr_{B^{\le u+1}}X-\Tr_{B^{\le u+1}}\mathcal J_T(X)\\
 &=-\Tr_{B^{\le u+1}}\mathcal J_T(X).
\end{align*}
Substituting both identities into
$\widehat\Phi_T=\mathcal J_T+\mathcal R_{I_u(B)}\mathcal N_T$ yields
\begin{align*}
 \Tr_{B^{\le u+1}\setminus D}\widehat\Phi_T(X)
 ={}&\Tr_{B^{\le u+1}\setminus D}\mathcal J_T(X)
 +(\sigma_\beta^B)_D\otimes\Tr_{B^{\le u+1}}\mathcal N_T(X)\\
 ={}&\Tr_{B^{\le u+1}\setminus D}\mathcal J_T(X)
 -(\sigma_\beta^B)_D\otimes\Tr_{B^{\le u+1}}\mathcal J_T(X).
\end{align*}
The last line is precisely the operator bounded in \eqref{q1:replacement-error}.
Applying that estimate gives
\begin{equation}
 \|\Tr_{B^{\le u+1}\setminus D}\widehat\Phi_T(X)\|_1
 \le M^8\zeta_u f(v)\|X\|_1.
 \label{q1:marginal}
\end{equation}

The reduced map on the left is itself a channel, so it also contracts the trace norm of
Hermitian inputs. Combining the two bounds gives
\begin{equation}
 \|\Tr_{B^{\le u+1}\setminus D}\widehat\Phi_T(X)\|_1
 \le\min\{1,M^8\zeta_u f(v)\}\|X\|_1.
 \label{q1:clipped-marginal}
\end{equation}

We now turn these marginal bounds into a transport estimate.
Enumerate the sites of $C_u(B)$ as $c_1,\ldots,c_{|C_u(B)|}$ in any order.
List the sites of $I_u(B)$ as $i_1,\ldots,i_{|I_u(B)|}$ by taking its
leftmost site, then its rightmost remaining site, then its leftmost
remaining site, and so on. This order works whether $C_{\rm in}$ occupies
one or both ends of $B$. Since $\Tr_{B^{\le u+1}}\widehat\Phi_T(X)=0$,
the telescoping estimate \eqref{eq:telescoping} in \Cref{lem:transport},
with all the $c$ sites followed by all the $i$ sites, gives
\begin{align*}
 \|\widehat\Phi_T(X)\|_{W_1}
 &\le\sum_{k=0}^{|C_u(B)|-1}
 \|\Tr_{\{c_1,\ldots,c_k\}}\widehat\Phi_T(X)\|_1\\
 &\quad+\sum_{k=0}^{|I_u(B)|-1}
 \|\Tr_{C_u(B)\cup\{i_1,\ldots,i_k\}}\widehat\Phi_T(X)\|_1\\
 &\le |C_u(B)|\|X\|_1+\sum_{k=0}^{|I_u(B)|-1}
 \|\Tr_{C_u(B)\cup\{i_1,\ldots,i_k\}}\widehat\Phi_T(X)\|_1.
\end{align*}
For $k=0$, the enumerated sets are empty and $\Tr_\varnothing=\id$.
The last inequality uses trace-norm contractivity of $\widehat\Phi_T$ and
partial trace to bound every term of the first sum by $\|X\|_1$.
In the second sum, the remaining interval is
$D=I_u(B)\setminus\{i_1,\ldots,i_k\}$, so the set being traced out is
$B^{\le u+1}\setminus D$. After $k$ removals, at least $\lfloor k/2\rfloor$
sites have been removed from each end of $I_u(B)$; hence $D$ is at distance
at least $\lfloor k/2\rfloor$ from each component of $C_{\rm in}$.
Equation~\eqref{q1:clipped-marginal} therefore bounds the $k$th term by
$\min\{1,M^8\zeta_u f(\lfloor k/2\rfloor)\}\|X\|_1$.
Each value of $\lfloor k/2\rfloor$ occurs at most twice, giving
\[
 \|\widehat\Phi_T(X)\|_{W_1}
 \le\left[|C_u(B)|+2\sum_{v=0}^\infty
 \min\{1,M^8\zeta_u f(v)\}\right]\|X\|_1
 \le\tfrac12 A_u\|X\|_1.
\]
Here we used $|C_u(B)|\le4(u+1)$ and $\zeta_u=3C_{\rm cl}d^{4(u+1)}$.
If $I_u(B)=\varnothing$, only the first sum is needed, since
$B^{\le u+1}=C_u(B)$. Finally, combine this bound with \eqref{q1:completion},
divide by $\|X\|_1/2$, and take the supremum over $X$.
This proves the middle case of \eqref{q1:influence-bound}.

\proofstep{Outside influence: $j\notin B^{\le u+1}$}
The channel $\Phi_T=\Phi_{B,u}(T)$ is supported on $B^{\le u+1}$ and
therefore acts trivially on $j$. Since $\Tr_jX=0$, this gives
$\Tr_j\Phi_T(X)=0$. The one-site identity in \Cref{lem:transport} and
trace-norm contractivity imply
\[
 \|\Phi_T(X)\|_{W_1}=\tfrac12\|\Phi_T(X)\|_1\le\tfrac12\|X\|_1.
\]
Taking the supremum in the definition of $D_{B,u}(j)$ proves the last case;
in fact, \Cref{lem:transport} gives $D_{B,u}(j)=1$.

Finally, consider the profile in \eqref{q1:clustering}.
For $a,c>0$, put $b=\max\{0,\log a\}$. A decreasing-function integral
comparison gives
\begin{align*}
 \sum_{v=0}^\infty\min\{1,ae^{-c\sqrt v}\}
 &\le1+\int_0^\infty\min\{1,e^{b-c\sqrt t}\}\,dt\\
 &=1+c^{-2}(b^2+2b+2).
\end{align*}
For the equality, split at $t=b^2/c^2$ and substitute $z=c\sqrt t$ in the tail.
Apply this with $a=3M^8C_{\rm cl}d^{4(u+1)}$ and $c=c_{\rm cl}$.
Then $b\le C_d(u+\beta J_{\mathrm{int}}+1)$ and $c_{\rm cl}^{-2}\le
C(1+\beta J_{\mathrm{int}})^2e^{8\beta J_{\mathrm{int}}}$.
Substitution in \eqref{q1:Ar} proves \eqref{q1:Ar-bound}.
\end{proof}

\subsection{Quantum DS for the block updates}
\label{q1:counting}

We now combine the influence estimates of \Cref{q1:influence} with the
transport approximation \eqref{qg:channel-transport} to prove the quantum
DS condition for the original channels $\Phi_B(T)$. We recall that each
site belongs to $m=\ell$ labelled blocks. For any integer radius $u\ge0$,
at most $4(u+1)$ blocks place it within distance $u$ of a boundary bond.
Thus, for large blocks, the contraction in their interiors outweighs the
boundary contributions. The localized channels provide the intermediate
influence bounds; we then sum the approximation errors to obtain the
following result.

\begin{proposition}[Quantum DS for the smoothed heat-bath block updates]
\label{q1:balance-proposition}
For every fixed $\beta,J_{\mathrm{int}},d$, there are an integer $\ell\ge2$
and a time $T>0$, depending only on these parameters, such that the family
$(\Phi_B(T))_{B\in\mathscr W_\ell}$ satisfies the quantum DS condition of
\Cref{def:quantumDS} with $\kappa=1/2$, namely
\begin{equation}
 \sum_{B\in\mathscr W_\ell}[D_B(j)-1]\le-\frac\ell2,
 \qquad j\in[N].
 \label{q1:strict-DS}
\end{equation}
The parameters satisfy \eqref{q1:parameters}. In particular, the same
$\ell,T$ work for
every chain length $N$ and every choice of the on-site terms $V_i$.
\end{proposition}
\begin{proof}
\proofstep{Counting the blocks}
Fix $j\in[N]$. Recall that $B_a=[a,a+\ell-1]\cap[N]$, with
$2-\ell\le a\le N$. Then $j\in B_a$ exactly when
$j-\ell+1\le a\le j$, so precisely $\ell$ labelled blocks contain $j$.
Write
\[
 \delta_B(j)=\min_{q\in\partial B}\dist(j,\{q,q+1\}),
\]
with $\delta_B(j)=\infty$ when $B=[N]$. A boundary bond labelled by $q$
can occur only for the block labels $a=q+1$ and $a=q-\ell+1$.
For each integer distance $v\ge0$, at most two bonds are at that distance
from $j$, namely $q=j+v$ and $q=j-v-1$ when these lie in $[1,N-1]$.
Hence at most four labelled blocks $B$ have $\delta_B(j)=v$.
Recall from \eqref{qg:collar-def} that
$C_u(B)=[N]\cap\bigcup_{q\in\partial B}[q-u,q+u+1]$ consists of the
sites within distance $u$ of a boundary bond of $B$. Thus $j\in C_u(B)$
is equivalent to $\delta_B(j)\le u$. Summing over $v=0,\ldots,u$
shows that at most $4(u+1)$ labelled blocks satisfy $j\in C_u(B)$.

\proofstep{Summing the influences}
Fix an integer $u\ge0$ and a time $T\ge0$, and put $\eta=e^{-M^{-4}T}$.
For each localized channel, \Cref{q1:influence} and $|B|\le\ell$ give
\[
 D_{B,u}(j)-1\le
 \begin{cases}
  -1+\eta,&j\in I_u(B),\\
  A_u+4\ell\eta-1,&j\in C_u(B),\\
  0,&j\notin B^{\le u+1}.
 \end{cases}
\]
Recall that $I_u(B)=B\setminus C_u(B)$ and $B^{\le u+1}=B\cup C_u(B)$.
In the interior and outside cases, the bound is exactly
$(-1+\eta)\one_{\{j\in B\}}$. In the boundary case $j\in C_u(B)$,
adding $A_u+4\ell\eta$ to this quantity gives at least
$A_u+4\ell\eta-1$, since $\eta\ge0$. Thus all three cases satisfy
\[
 D_{B,u}(j)-1\le(-1+\eta)\one_{\{j\in B\}}
       +(A_u+4\ell\eta)\one_{\{j\in C_u(B)\}}.
\]
The indicators count precisely $\ell$ blocks with $j\in B$ and at most
$4(u+1)$ blocks with $j\in C_u(B)$. Summing therefore gives
\begin{align*}
 \sum_{B\in\mathscr W_\ell}(D_{B,u}(j)-1)
 &\le(-1+\eta)\ell+4(u+1)(A_u+4\ell\eta)\\
 &=-\ell+[1+16(u+1)]\ell\eta+4(u+1)A_u.
\end{align*}
Now, for blocks $B$ with $\delta_B(j)\le u$,
\eqref{qg:channel-transport} and \eqref{eq:influence} give directly
\[
 D_B(j)\le D_{B,u}(j)+2K_\beta T(\ell+u+2)e^{-\lambda u/2}.
\]
For blocks with $u<\delta_B(j)<\infty$, use radius
$v=\delta_B(j)-1$ instead. Then $j\notin C_v(B)$: if $j\in B$, it lies
in $I_v(B)$; otherwise, $j\notin B^{\le v+1}$. Since $0<\eta\le1$,
\Cref{q1:influence} therefore gives
\begin{align*}
 D_{B,v}(j)-1&\le(-1+\eta)\one_{\{j\in B\}}\le0,\\
 D_B(j)-1&\le(-1+\eta)\one_{\{j\in B\}}
       +2K_\beta T(\ell+v+2)e^{-\lambda v/2}.
\end{align*}
For $B=[N]$, the localized channel is exact and its influence is
$e^{-T}\le\eta$, so $D_B(j)-1\le-1+\eta$ without an approximation error.

We now bound the total change in the influence coefficients when these
localized channels are replaced by $\Phi_B(T)$. The same approximation
bound applies to the absolute difference of the influences, and hence
\begin{align*}
 &\sum_{\substack{B\in\mathscr W_\ell\\\delta_B(j)\le u}}
       |D_B(j)-D_{B,u}(j)|
  +\sum_{\substack{B\in\mathscr W_\ell\\u<\delta_B(j)<\infty}}
       |D_B(j)-D_{B,\delta_B(j)-1}(j)|\\
 &\quad\le2K_\beta T(\ell+u+2)e^{-\lambda u/2}
       \underbrace{\#\{B\in\mathscr W_\ell:\delta_B(j)\le u\}}_{\le4(u+1)}\\
 &\qquad+2K_\beta T\sum_{v=u}^\infty(\ell+v+2)e^{-\lambda v/2}
       \underbrace{\#\{B\in\mathscr W_\ell:\delta_B(j)=v+1\}}_{\le4}\\
 &\quad\le8K_\beta T\left[(u+1)(\ell+u+2)e^{-\lambda u/2}
              +\sum_{v=u}^\infty(\ell+v+2)e^{-\lambda v/2}\right]\\
 &\quad=8K_\beta T e^{-\lambda u/2}
       [(u+1+g_0)(\ell+u+2)+g_1],
\end{align*}
where
\[
 g_0=\sum_{n=0}^\infty e^{-\lambda n/2}=\frac1{1-e^{-\lambda/2}},\qquad
 g_1=\sum_{n=0}^\infty ne^{-\lambda n/2}
                  =\frac{e^{-\lambda/2}}{(1-e^{-\lambda/2})^2}.
\]
The indicator bound thus applies with radius $u$ for $\delta_B(j)\le u$
and radius $\delta_B(j)-1$ for $u<\delta_B(j)<\infty$.
Adding the approximation errors and using the block counts gives
\begin{equation}
 \begin{split}
 \sum_{B\in\mathscr W_\ell}[D_B(j)-1]
 \le{}&(-1+\eta)\sum_{B\in\mathscr W_\ell}\one_{\{j\in B\}}
 +(A_u+4\ell\eta)\sum_{B\in\mathscr W_\ell}\one_{\{j\in C_u(B)\}}\\
 &+8K_\beta T e^{-\lambda u/2}[(u+1+g_0)(\ell+u+2)+g_1]\\
 \le{}&(-1+\eta)\ell+4(u+1)(A_u+4\ell\eta)\\
 &+8K_\beta T e^{-\lambda u/2}[(u+1+g_0)(\ell+u+2)+g_1]\\
 ={}&-\ell+[1+16(u+1)]\ell e^{-M^{-4}T}+4(u+1)A_u\\
 &+8K_\beta T e^{-\lambda u/2}[(u+1+g_0)(\ell+u+2)+g_1].
 \end{split}
 \label{q1:balance}
\end{equation}
These counts also apply when $N<\ell$, since distinct labels are retained
when blocks coincide with the whole chain.

\proofstep{Choosing the block length and time}
For a constant $A_d\ge1$ to be fixed below, set
\begin{equation}
 u=\lceil A_d(1+\beta J_{\mathrm{int}})^2\rceil,\qquad
 \ell=\lceil32(u+1)A_u\rceil,\qquad
 T=M^4\log\bigl(8[1+16(u+1)]\bigr).
 \label{q1:choices}
\end{equation}
Recall that $M=e^{\beta J_{\mathrm{int}}/2}$ by
\Cref{q1:boundary-operators}. These choices give
\[
 [1+16(u+1)]\ell e^{-M^{-4}T}=\frac\ell8,
 \qquad 4(u+1)A_u\le\frac\ell8.
\]
It remains to bound the localization error in \eqref{q1:balance} by
$\ell/4$. Since $A_u\ge8(u+1)$, we have $\ell\ge256(u+1)^2$.
Also $g_0\le C(1+\beta J_{\mathrm{int}})$ and
$g_1\le C(1+\beta J_{\mathrm{int}})^2$. Thus
\[
 E_u:=\frac{8K_\beta T}{\ell}e^{-\lambda u/2}
       [(u+1+g_0)(\ell+u+2)+g_1]
 \le CK_\beta T(u+1)e^{-\lambda u/2}.
\]
Using $\lambda=c/(1+\beta J_{\mathrm{int}})$ and
$K_\beta\le C_de^{2\beta J_{\mathrm{int}}}(1+\beta J_{\mathrm{int}})^7$
from \Cref{qg:chain-transport}, we obtain
\begin{align*}
 E_u
 &\le C_d(A_d+1)e^{4\beta J_{\mathrm{int}}}
       (1+\beta J_{\mathrm{int}})^9
       [1+\log(A_d+1)+\log(1+\beta J_{\mathrm{int}})]
       e^{-cA_d(1+\beta J_{\mathrm{int}})/2}\\
 &\le C_d(A_d+1)[1+\log(A_d+1)]
       e^{-(cA_d/2-14)(1+\beta J_{\mathrm{int}})}.
\end{align*}
The second line uses
$(1+\beta J_{\mathrm{int}})^9[1+\log(1+\beta J_{\mathrm{int}})]
\le Ce^{10(1+\beta J_{\mathrm{int}})}$.
Choose $A_d$ so that $cA_d/2>14$ and the last bound is at most $1/4$
at $\beta J_{\mathrm{int}}=0$; it is then at most $1/4$ for every
$\beta J_{\mathrm{int}}\ge0$. Substitution in \eqref{q1:balance} gives
\[
 \sum_{B\in\mathscr W_\ell}[D_B(j)-1]
 \le-\ell+\frac\ell8+\frac\ell8+\frac\ell4=-\frac\ell2.
\]
The channels preserve $\sigma_{\beta,H}$ by \Cref{qg:KMS}. Since
$m=\ell$, this proves the quantum DS condition with $\kappa=1/2$.

Finally, \eqref{q1:Ar-bound} and $u\le C_d(1+\beta J_{\mathrm{int}})^2$
give
\begin{align*}
 \ell
 &\le C_d(u+1)(1+\beta J_{\mathrm{int}})^2e^{8\beta J_{\mathrm{int}}}
       (u+\beta J_{\mathrm{int}}+1)^2
 \le C_d(1+\beta J_{\mathrm{int}})^8e^{8\beta J_{\mathrm{int}}},\\
 T&\le C_de^{2\beta J_{\mathrm{int}}}\log(2+\beta J_{\mathrm{int}}).
\end{align*}
This proves \eqref{q1:parameters}.
\end{proof}

\begin{proof}[Proof of Theorem~\ref{q1:main}]
\Cref{q1:balance-proposition} gives the parameter bounds
\eqref{q1:parameters} and the quantum DS condition with $\kappa=1/2$.
Applying \Cref{thm:qds} with $a_B=1/\ell$ gives uniqueness of the
stationary state and \eqref{q1:mixing}.
\end{proof}

\section{Stability under local perturbations}
\label{qp:section}
{%
In this section we prove that a strict quantum DS condition persists
under sufficiently small local perturbations of a reference Hamiltonian.
We then apply the general theorem to the Ising model at temperatures where
single-site conditions fail
(\Cref{qp:examples}).
}
{%
Fix a finite-range interaction geometry of polynomial volume growth
with exponent $D\ge1$, bounded local dimension $d$, and bounded-size,
bounded-diameter labelled blocks $B$ with rates $a_B>0$. Consider a perturbed Hamiltonian
\begin{equation}
 H^{(0)}=\sum_X h_X^{(0)},\qquad
 H=H^{(0)}+Q,\qquad Q=\sum_X Q_X,\qquad
 \delta=\sup_i\sum_{X\ni i}\|Q_X\|.
 \label{qp:local-norm}
\end{equation}
We also fix the reference intersite strength to be bounded by
$J_{\mathrm{int}}^{(0)}$, while on-site terms may have arbitrary finite
norms.
}
{ Choose a neighbourhood $R_B$
containing $B$ and every prescribed support crossing $B$, for both
the reference and the perturbation. We assume that, for every $r\ge0$,
\begin{equation}
 |R_B^{\le r}|\le V(1+r)^D,\qquad
 \sup_j\sum_{B:j\in R_B^{\le r}}a_B\le \Xi(1+r)^D.
 \label{qp:density}
\end{equation}
Here $V,\Xi$ are independent of the system size. For example, on a
finite lattice $\Lambda\subseteq\ZZ^D$ with the $\ell^\infty$ distance,
fix integers $L_0,\ell\ge1$
and take the nonempty blocks
\[
 B_x=\Lambda\cap\bigl(L_0x+\{0,\ldots,\ell L_0-1\}^D\bigr),
 \qquad x\in\ZZ^D,\qquad a_{B_x}=\ell^{-D}.
\]
Each block has at most $(\ell L_0)^D$ sites, and each site belongs
to exactly $\ell^D$ labelled blocks. Choosing $R_B$ as a fixed-radius
neighbourhood of $B$ containing all crossing supports gives
\eqref{qp:density}, with $V,\Xi$ depending only on $D,L_0,\ell$
and that radius. 
Write
\begin{equation}
 b=\sup_B|B|,\qquad
 w=\beta\sup_B\sum_{X\text{ crossing }B}\|h_X^{(0)}\|.
 \label{qp:parameters}
\end{equation}
Note that the sum defining $w$ excludes on-site terms. All blocks, rates, and
internal evolution times are fixed when varying the Hamiltonian.
}

{%
\begin{theorem}[Stability of the quantum DS condition]
\label{qp:commuting-main}
\label{qp:general-main}
Fix $\beta>0$, a volume-independent upper bound $T>0$, block times
$0<T_B\le T$, and $\kappa_0>0$.
With $D_\Phi(j)$ defined in \eqref{eq:influence}, suppose the QBP channels
$\Phi_B(H^{(0)})=e^{T_B\mathcal K_B(H^{(0)})}$ satisfy the quantum DS condition
with rates $a_B$ and constant $\kappa_0$:
\begin{equation}
 \min_j\sum_Ba_B[1-D_{\Phi_B(H^{(0)})}(j)]\ge\kappa_0>0.
 \label{qp:reference-certificate}
\end{equation}
Then there exists $c\in(0,1]$, independent of the volume, such that whenever
$\delta\le\delta_*:=c/\beta$, the QBP channels
$\Phi_B(H)=e^{T_B\mathcal K_B(H)}$ satisfy the quantum DS condition
with constant $\kappa_0/2$.
The constant $c$ depends only on the fixed geometric and intersite-interaction
bounds, $\beta$, the block parameters and rates, $T$, and $\kappa_0$. In
particular, with these data fixed, it is independent of the reference
on-site norms. 
The times $T_B$ are kept fixed under perturbation.
\end{theorem}
}

{%
\subsection{Proof of Theorem~\ref{qp:general-main}}
\label{qp:proofs}
The proof of Theorem \ref{qp:general-main} follows the same outline as the proof of Theorem \ref{q1:main}.
First, \Cref{qp:lipschitz} bounds the change of a block generator on
a finite region, while \Cref{qp:boundaries} controls the effect of
enlarging that region. Duhamel's formula and \Cref{lem:transport}
turn these estimates into bounds on the channel influences.
Summing over blocks then shows that sufficiently
small $\delta$ gives the quantum DS condition with constant $\kappa_0/2$,
uniformly in the system size.
}

{For $U\supseteq R_B$, write $O_{B,U}(H)$,
$\mathcal T_{B,U}(H)$, $S_{B,U}(H)$ and $\mathcal K_{B,U}(H)$
for the objects in \eqref{qg:ODE} and \eqref{qg:generator} constructed from
$H_U$, using the interpolation
\[
 H_U-(1-s)W_{\partial B}(H),\qquad 0\le s\le1.
\]
Here $W_{\partial B}(H)$ is the crossing interaction
\eqref{eq:crossing-interaction} of $H$. Since $U$ contains $B$ and
every support crossing $B$,
\[
 H_U=H_B+H_{U\setminus B}+W_{\partial B}(H).
\]
The crossing interaction is therefore unchanged. The block Hamiltonian
also remains $H_B$, so the reset in \eqref{qg:reset} uses the same
Gibbs state $\sigma_\beta^B$ and acts on $U$ as
\[
 \mathcal R_B(H)(X)=\sigma_\beta^B\otimes\Tr_B X,
 \qquad X\in\BB(\HH_U).
\]
The reference operators $O_{B,U}(H^{(0)})$,
$\mathcal T_{B,U}(H^{(0)})$, $S_{B,U}(H^{(0)})$ and
$\mathcal K_{B,U}(H^{(0)})$ are defined by replacing $H$ with
$H^{(0)}$ throughout this construction, including the Gibbs state
in the reset.}

{%
\begin{lemma}[Local Lipschitz bound]
\label{qp:lipschitz}
For $\beta\delta\le1$, $n=|U|$, and an arbitrary
reference $H^{(0)}$, there is a constant $C_{\rm loc}$ such that
\begin{equation}
 \|\mathcal K_{B,U}(H)-\mathcal K_{B,U}(H^{(0)})\|_\diamond
 \le C_{\rm loc}\,\beta\delta n^2,
 \label{qp:local-lipschitz}
\end{equation}
where $C_{\rm loc}$ depends only on
$d,b,w$.
\end{lemma}

\begin{proof}
\proofstep{(i) Estimates for the reset and the maps $\mathcal F,\mathcal G$}
Recall that $\mathcal F_K$ is defined by \eqref{qg:F} and that
$\mathcal G_K$ has the integral representation \eqref{qg:G-integral}.
For any Hermitian $K,V$, differentiating the exponential and its
normalization by Duhamel's formula gives
\cite[Sec.~V.A]{Kliesch2014Locality}
\[
 \frac{d}{dz}\sigma_{\beta,K+zV}
 =-\beta\int_0^1
  \sigma_{\beta,K+zV}^{1-t}
  \bigl(V-\Tr(\sigma_{\beta,K+zV}V)I\bigr)
  \sigma_{\beta,K+zV}^{t}\,dt.
\]
The Schatten H\"older inequality bounds its trace norm by $2\beta\|V\|$.
Since $\|H_B-H_B^{(0)}\|\le |B|\delta$, the resets in
\eqref{qg:reset} satisfy
\[
 \|\mathcal R_B(H)-\mathcal R_B(H^{(0)})\|_\diamond
 =\|\sigma_{\beta,H_B}-\sigma_{\beta,H_B^{(0)}}\|_1
 \le2|B|\beta\delta.
\]
For $0\le s\le1$, the difference of the two interpolating Hamiltonians
has coefficient $s$ on crossing perturbation terms and coefficient
one on the remaining perturbation terms. Hence
\[
 \begin{aligned}
 \big\|H_U-H_U^{(0)}-(1-s)
       \bigl(W_{\partial B}(H)-W_{\partial B}(H^{(0)})\bigr)\big\|\le\sum_{X\subseteq U}\|Q_X\|\le n\delta.
 \end{aligned}
\]
For every operator $A$ on $U$ and every $t\in\RR$, the Duhamel identity gives
\[
 \begin{aligned}
 \big\|\tau_t^{H_U-(1-s)W_{\partial B}(H)}(A)
       -\tau_t^{H_U^{(0)}-(1-s)W_{\partial B}(H^{(0)})}(A)\big\|\le2|t|n\delta\|A\|.
 \end{aligned}
\]
The crossing interactions satisfy
\[
 \begin{gathered}
 \|W_{\partial B}(H)-W_{\partial B}(H^{(0)})\|\le |B|\delta,
 \qquad \beta\|W_{\partial B}(H^{(0)})\|\le w,\\
 \beta\|W_{\partial B}(H)\|\le w+b\beta\delta\le w+b.
 \end{gathered}
\]
Integrating in \eqref{qg:F}, using
$\int |t|p_\beta(t)\,dt\le\beta$ from \eqref{eq:qbp-kernel-tail}, gives
\begin{equation}
 \begin{aligned}
 &\big\|\mathcal F_{H_U-(1-s)W_{\partial B}(H)}(W_{\partial B}(H))
       -\mathcal F_{H_U^{(0)}-(1-s)W_{\partial B}(H^{(0)})}
          (W_{\partial B}(H^{(0)}))\big\|\\
 &\quad\le\|W_{\partial B}(H)-W_{\partial B}(H^{(0)})\|
       +2\beta n\delta\|W_{\partial B}(H^{(0)})\|\\
 &\quad\le (|B|+2nw)\delta.
 \end{aligned}
 \label{qp:F-lipschitz}
\end{equation}
At $s=1$, the same Duhamel estimate, applied at positive and negative
times in \eqref{qg:G-integral}, gives
\begin{equation}
 \|[\mathcal G_{H_U}-\mathcal G_{H_U^{(0)}}](A)\|
 \le\frac{\beta}{4}\|H_U-H_U^{(0)}\|\|A\|
 \le\frac{\beta}{4}n\delta\|A\|,
 \label{qp:G-lipschitz}
\end{equation}
since
\[
 \int_0^\infty \frac{t}{\sinh(2\pi t/\beta)}\,dt=\frac{\beta^2}{16}.
\]
\Cref{qg:multiplier} gives
$\|\mathcal G_{H_U}\|_{\rm cb}\le4+n\log_2d$ on $U$.

\proofstep{(ii) Estimate for the block generator}
Since $B\subseteq U$, we have $|B|\le n$.
By \eqref{qp:F-lipschitz}, the difference of the ODE coefficients in
\eqref{qg:ODE} has norm at most
\[
 \frac{(|B|+2nw)\beta\delta}{4}
 \le\frac{1+2w}{4}n\beta\delta.
\]
Variation of constants in \eqref{qg:ODE}, using
\eqref{qg:O-interval}, gives
\[
 \begin{aligned}
 \|O_{B,U}(H)^{\pm1}-O_{B,U}(H^{(0)})^{\pm1}\|
 &\le\frac{1+2w}{4}n\beta\delta
       \int_0^1 e^{(w+b)(1-v)/4}e^{(w+b)v/4}\,dv\\
 &=\frac{1+2w}{4}e^{(w+b)/4}n\beta\delta.
 \end{aligned}
\]
Using \eqref{qg:generator} and the bound, valid for any operators $A,A'$,
\[
 \|\operatorname{Ad}_A-\operatorname{Ad}_{A'}\|_\diamond
 \le(\|A\|+\|A'\|)\|A-A'\|,
\]
and including the reset difference gives
\begin{align*}
 \|\mathcal T_{B,U}(H)-\mathcal T_{B,U}(H^{(0)})\|_\diamond
 &\le e^{w+b}\bigl[(1+2w)n+2|B|\bigr]\beta\delta\\
 &\le e^{w+b}(3+2w)n\beta\delta.
\end{align*}
Taking adjoints and evaluating at $\one$ gives the same bound for
$S_{B,U}(H)-S_{B,U}(H^{(0)})$. Thus
\begin{equation}
 \begin{aligned}
 \|\mathcal T_{B,U}(H)-\mathcal T_{B,U}(H^{(0)})\|_\diamond
 &\le e^{w+b}(3+2w)n\beta\delta,\\
 \|S_{B,U}(H)-S_{B,U}(H^{(0)})\|
 &\le e^{w+b}(3+2w)n\beta\delta.
 \end{aligned}
 \label{qp:jump-continuity}
\end{equation}
Both $\|S_{B,U}(H)\|$ and
$\|S_{B,U}(H^{(0)})\|$ are at most $e^{w+b}$. In particular, the coherent corrections satisfy
\begin{equation}
 \begin{aligned}
 &\|\mathcal G_{H_U}(S_{B,U}(H))
          -\mathcal G_{H_U^{(0)}}(S_{B,U}(H^{(0)}))\|\\
 &\quad\le\|\mathcal G_{H_U}\|_{\rm cb}
                   \|S_{B,U}(H)-S_{B,U}(H^{(0)})\|
       +\tfrac\beta4\|H_U-H_U^{(0)}\|\|S_{B,U}(H^{(0)})\|\\
 &\quad\le e^{w+b}\left[(4+n\log_2d)(3+2w)+\frac14\right]n\beta\delta.
\end{aligned}
 \label{qp:coherent-continuity}
\end{equation}
Combining \eqref{qp:jump-continuity} and \eqref{qp:coherent-continuity}
in the generator decomposition \eqref{qg:generator} gives
\begin{align*}
 &\|\mathcal K_{B,U}(H)-\mathcal K_{B,U}(H^{(0)})\|_\diamond\\
 &\qquad\le\|\mathcal T_{B,U}(H)-\mathcal T_{B,U}(H^{(0)})\|_\diamond
       +\|S_{B,U}(H)-S_{B,U}(H^{(0)})\|\\
 &\qquad\quad+2\|\mathcal G_{H_U}(S_{B,U}(H))
       -\mathcal G_{H_U^{(0)}}(S_{B,U}(H^{(0)}))\|\\
 &\qquad\le e^{w+b}\left[2(5+n\log_2d)(3+2w)+\frac12\right]n\beta\delta\\
 &\qquad\le e^{w+b}\left[2(5+\log_2d)(3+2w)+\frac12\right]n^2\beta\delta.
\end{align*}
The last bound therefore gives
\eqref{qp:local-lipschitz} with
\begin{equation}
 C_{\rm loc}=e^{w+b}\left[2(5+\log_2d)(3+2w)+\frac12\right].
 \label{qp:local-constant}
\end{equation}
\end{proof}
}

{%
\noindent We next control the error from restricting the block generators
for $H$ and $H^{(0)}$ to $R_B^{\le r}$, uniformly in the volume.

\begin{lemma}[Differences of localized generators]
\label{qp:boundaries}
Under the assumptions of \Cref{qp:lipschitz}, put $H_0=H^{(0)}$ and
$H_1=H$. For $\beta\delta\le1$, there are $C_s,\mu>0$, independent of the
volume, such that for every $a\in\{0,1\}$ and every integer $r\ge0$,
\begin{equation}
 \|\mathcal K_{B,R_B^{\le r+1}}(H_a)
       -\mathcal K_{B,R_B^{\le r}}(H_a)\|_\diamond
 \le C_s(r+1)^{D+1}e^{-\mu r}.
 \label{qp:boundary-envelope}
\end{equation}
The constants depend only on $\beta,d,b,w,V,J_{\mathrm{int}}^{(0)}$,
the interaction-range bound in \eqref{eq:labelled-interactions}, and the
volume-growth constants $C_{\rm vol},D$ defined in
Section~\ref{sec:interactions}, uniformly for $\beta\delta\le1$ and
independently of the on-site norms.
\end{lemma}

\begin{proof}
\proofstep{(i) Restriction estimates for $\mathcal F$ and $\mathcal G$}
{By \eqref{qp:local-norm} and \eqref{qp:parameters}, for
$\beta\delta\le1$ we have
\begin{equation}
 \begin{aligned}
 \sup_i\sum_{\substack{X\ni i\\|X|\ge2}}\|h_X^{(0)}\|
 &\le J_{\mathrm{int}}^{(0)},\\
 \sup_i\sum_{\substack{X\ni i\\|X|\ge2}}\|h_X^{(0)}+Q_X\|
 &\le J_{\mathrm{int}}^{(0)}+\delta
 \le J_{\mathrm{int}}^{(0)}+\frac1\beta,\\
 \beta\|W_{\partial B}(H^{(0)})\|&\le w,\\
 \beta\|W_{\partial B}(H)\|&\le w+|B|\beta\delta\le w+b.
 \end{aligned}
 \label{qp:common-interaction-bounds}
\end{equation}
The constants in \eqref{qg:lr-eq} can therefore be chosen independently
of the on-site norms.}
Let $A$ be supported on a nonempty set
$S\subseteq\Omega\subseteq\Omega'\subseteq\Lambda$, and let $q\ge0$
satisfy $S^{\le q}\cap\Omega'\subseteq\Omega$.
Apply the Lieb--Robinson bound \eqref{qg:lr-eq} from
\Cref{qg:lr} within $\Omega'$, with the common choice
$v_{\rm LR}=C_{\rm LR}(J_{\mathrm{int}}^{(0)}+1/\beta)>0$
allowed by \eqref{qp:common-interaction-bounds}. Together with unitary
norm preservation, this gives
\[
 \|\tau_t^{H_{\Omega'}}(A)-\tau_t^{H_\Omega}(A)\|
 \le |S|\|A\|\min\bigl\{2,
 C_{\rm LR}(e^{v_{\rm LR}|t|}-1)e^{-\mu_{\rm LR}q}\bigr\}.
\]
Splitting \eqref{qg:G-integral} at
$t=\mu_{\rm LR}(1+q)/(2v_{\rm LR})$ gives
\[
 \begin{aligned}
 &\|\mathcal G_{H_{\Omega'}}(A)-\mathcal G_{H_\Omega}(A)\|\\
 &\quad\le 2C_{\rm LR}|S|\|A\|e^{-\mu_{\rm LR}q}
 \int_0^{\mu_{\rm LR}(1+q)/(2v_{\rm LR})}
 \frac{e^{v_{\rm LR}t}-1}{\beta\sinh(2\pi t/\beta)}\,dt\\
 &\qquad+4\|A\|
 \int_{\mu_{\rm LR}(1+q)/(2v_{\rm LR})}^{\infty}
 \frac{dt}{\beta\sinh(2\pi t/\beta)}.
 \end{aligned}
\]
The factor $2$ accounts for positive and negative times; in the
second integral, each difference of two unitary evolutions has
norm at most $2\|A\|$.
Since $\beta\sinh(2\pi t/\beta)\ge2\pi t$ and
$e^{v_{\rm LR}t}-1\le v_{\rm LR}t\,e^{v_{\rm LR}t}$, the first
integral satisfies
\[
 \begin{aligned}
 \int_0^{\mu_{\rm LR}(1+q)/(2v_{\rm LR})}
 \frac{e^{v_{\rm LR}t}-1}{\beta\sinh(2\pi t/\beta)}\,dt
 &\le\frac{v_{\rm LR}}{2\pi}
 \int_0^{\mu_{\rm LR}(1+q)/(2v_{\rm LR})}e^{v_{\rm LR}t}\,dt\\
 &=\frac{e^{\mu_{\rm LR}(1+q)/2}-1}{2\pi}.
 \end{aligned}
\]
Multiplying by $e^{-\mu_{\rm LR}q}$ bounds the corresponding
contribution by $C|S|\|A\|e^{-\mu_{\rm LR}q/2}$.
For the second integral,
\[
 \begin{aligned}
 \int_{\mu_{\rm LR}(1+q)/(2v_{\rm LR})}^{\infty}
 \frac{dt}{\beta\sinh(2\pi t/\beta)}
 &=\frac1{2\pi}\log\coth
   \left(\frac{\pi\mu_{\rm LR}(1+q)}{2\beta v_{\rm LR}}\right)\\
 &\le\frac{e^{-\pi\mu_{\rm LR}(1+q)/(\beta v_{\rm LR})}}
 {\pi\bigl(1-e^{-\pi\mu_{\rm LR}/(\beta v_{\rm LR})}\bigr)}.
 \end{aligned}
\]
The inequality follows from $\log(1+x)\le x$ for $x\ge0$ and $q\ge0$.
The denominator is positive and independent of $q$, so the second
contribution is at most
$C\|A\|e^{-\pi\mu_{\rm LR}q/(\beta v_{\rm LR})}$.
For $\mathcal F$, splitting \eqref{qg:F} at the same time and using
$\int_{\mathbb R}p_\beta(t)\,dt=1$ and \eqref{eq:qbp-kernel-tail} gives
\[
 \begin{aligned}
 \|\mathcal F_{H_{\Omega'}}(A)-\mathcal F_{H_\Omega}(A)\|
 \le C_{\rm LR}|S|\|A\|e^{-\mu_{\rm LR}q}
       \bigl(e^{\mu_{\rm LR}(1+q)/2}-1\bigr)+2\|A\|e^{-\pi\mu_{\rm LR}(1+q)/(\beta v_{\rm LR})}.
 \end{aligned}
\]
Combining the two estimates gives
\[
 \begin{aligned}
 \|\mathcal F_{H_{\Omega'}}(A)-\mathcal F_{H_\Omega}(A)\|
 +\|\mathcal G_{H_{\Omega'}}(A)-\mathcal G_{H_\Omega}(A)\|
 &\le C|S|\|A\|
 \left(e^{-\mu_{\rm LR}q/2}
       +e^{-\pi\mu_{\rm LR}q/(\beta v_{\rm LR})}\right)\\
& \le C|S|\|A\|e^{-\mu_1q},
 \end{aligned}
\]
where $\mu_1=\min\{\mu_{\rm LR}/2,
\pi\mu_{\rm LR}/(\beta v_{\rm LR})\}>0$ and $C$ depends only on
$\beta,J_{\mathrm{int}}^{(0)}$ and the geometric parameters listed in
the lemma.

\proofstep{(ii) The QBP generator}
First consider $H$. Since $\beta\|W_{\partial B}(H)\|\le w+b$ by
\eqref{qp:common-interaction-bounds}, part~(i), applied along the
restricted interpolation with source $W_{\partial B}(H)$ supported
on $R_B$ and buffer $r$, together with variation of constants in
\eqref{qg:ODE} and \eqref{qg:O-interval}, gives
\[
 \begin{aligned}
 \|O_{B,R_B^{\le r}}(H)^{\pm1}\|&\le e^{(w+b)/4},\\
 \|O_{B,R_B^{\le r+1}}(H)^{\pm1}
       -O_{B,R_B^{\le r}}(H)^{\pm1}\|&\le Ce^{-\mu_1r}.
 \end{aligned}
\]
The bounds on $O_{B,R_B}(H)^{\pm1}$ give
$\|S_{B,R_B}(H)\|\le C$.
The conjugation estimate used in the proof of \Cref{qp:lipschitz},
with the same reset for both restrictions, then gives
\[
 \begin{aligned}
 \|\mathcal T_{B,R_B^{\le r+1}}(H)
       -\mathcal T_{B,R_B^{\le r}}(H)\|_\diamond+\|S_{B,R_B^{\le r+1}}(H)-S_{B,R_B^{\le r}}(H)\|
 \le Ce^{-\mu_1r}.
 \end{aligned}
\]
For $m\ge1$, the difference
$S_{B,R_B^{\le m}}(H)-S_{B,R_B^{\le m-1}}(H)$ is supported on
$R_B^{\le m}$ and has norm at most $Ce^{-\mu_1m}$, after enlarging $C$.
Telescoping these differences and using linearity gives
\[
 \begin{aligned}
 &\mathcal G_{H_{R_B^{\le r+1}}}
       \bigl(S_{B,R_B^{\le r+1}}(H)\bigr)
  -\mathcal G_{H_{R_B^{\le r}}}\bigl(S_{B,R_B^{\le r}}(H)\bigr)\\
 &\quad=\mathcal G_{H_{R_B^{\le r+1}}}
       \bigl(S_{B,R_B^{\le r+1}}(H)-S_{B,R_B^{\le r}}(H)\bigr)\\
 &\qquad+\bigl[\mathcal G_{H_{R_B^{\le r+1}}}
                   -\mathcal G_{H_{R_B^{\le r}}}\bigr]
       \bigl(S_{B,R_B}(H)\bigr)\\
 &\qquad+\sum_{m=1}^{r}
       \bigl[\mathcal G_{H_{R_B^{\le r+1}}}
                   -\mathcal G_{H_{R_B^{\le r}}}\bigr]
       \bigl(S_{B,R_B^{\le m}}(H)-S_{B,R_B^{\le m-1}}(H)\bigr).
 \end{aligned}
\]
Since $(R_B^{\le m})^{\le r-m}\subseteq R_B^{\le r}$ for $m\le r$,
\Cref{qg:multiplier} and part~(i) give
\begin{align*}
 &\bigl\|\mathcal G_{H_{R_B^{\le r+1}}}
       \bigl(S_{B,R_B^{\le r+1}}(H)-S_{B,R_B^{\le r}}(H)\bigr)\bigr\|
 \le C(r+2)^D e^{-\mu_1r},\\
 &\bigl\|\bigl[\mathcal G_{H_{R_B^{\le r+1}}}
                   -\mathcal G_{H_{R_B^{\le r}}}\bigr]
       \bigl(S_{B,R_B}(H)\bigr)\bigr\|
 \le Ce^{-\mu_1r},\\
 &\bigl\|\bigl[\mathcal G_{H_{R_B^{\le r+1}}}
                   -\mathcal G_{H_{R_B^{\le r}}}\bigr]
       \bigl(S_{B,R_B^{\le m}}(H)-S_{B,R_B^{\le m-1}}(H)\bigr)\bigr\|\\
 &\qquad\le C(m+1)^D e^{-\mu_1m}e^{-\mu_1(r-m)}
 =C(m+1)^D e^{-\mu_1r},\qquad 1\le m\le r.
\end{align*}
Thus
\[
 \begin{aligned}
 \bigl\|\mathcal G_{H_{R_B^{\le r+1}}}
       \bigl(S_{B,R_B^{\le r+1}}(H)\bigr)
  -\mathcal G_{H_{R_B^{\le r}}}\bigl(S_{B,R_B^{\le r}}(H)\bigr)\bigr\|&\le C(r+2)^D e^{-\mu_1r}
       +Ce^{-\mu_1r}\sum_{m=0}^r(m+1)^D\\
& \le C(r+2)^{D+1}e^{-\mu_1r}.
 \end{aligned}
\]
Together with the estimates for
$\mathcal T_{B,R_B^{\le r+1}}(H)-\mathcal T_{B,R_B^{\le r}}(H)$ and
$S_{B,R_B^{\le r+1}}(H)-S_{B,R_B^{\le r}}(H)$, this gives
\eqref{qp:boundary-envelope} for $H$, after adjusting the constants.
The same bounds apply directly to $H^{(0)}$, with the same constants
$C_s,\mu$, by \eqref{qp:common-interaction-bounds}.

\end{proof}

}

{%
\begin{proof}[Proof of \Cref{qp:general-main}]
We use the constant $C_{\rm loc}$ from \Cref{qp:lipschitz}
and the constants $C_s,\mu$ from \Cref{qp:boundaries}.
For every integer $R\ge0$ and $\beta\delta\le1$, Duhamel's formula
and diamond-norm contractivity give
\[
 \bigl\|e^{T_B\mathcal K_{B,R_B^{\le R}}(H)}
       -e^{T_B\mathcal K_{B,R_B^{\le R}}(H^{(0)})}\bigr\|_\diamond
 \le TC_{\rm loc}\beta\delta|R_B^{\le R}|^2.
\]
With $H_a$ as in \Cref{qp:boundaries}, that lemma also gives
\[
 \sum_{a=0}^1
 \bigl\|e^{T_B\mathcal K_{B,R_B^{\le r+1}}(H_a)}
       -e^{T_B\mathcal K_{B,R_B^{\le r}}(H_a)}\bigr\|_\diamond
 \le 2TC_s(r+1)^{D+1}e^{-\mu r}.
\]
For channels supported on $U$, \Cref{lem:transport} and the
normalization in \eqref{eq:influence} imply
\[
 |D_\Phi(j)-D_\Psi(j)|\le2|U|\|\Phi-\Psi\|_\diamond.
\]
If $j\notin U$, both influences are one, so their difference is zero.
Consequently,
\[
 \begin{aligned}
 \left|D_{e^{T_B\mathcal K_{B,R_B^{\le R}}(H)}}(j)
       -D_{e^{T_B\mathcal K_{B,R_B^{\le R}}(H^{(0)})}}(j)\right|\le2TC_{\rm loc}\beta\delta|R_B^{\le R}|^3
       \mathbf1_{\{j\in R_B^{\le R}\}},
 \end{aligned}
\]
and
\[
 \begin{aligned}
 \sum_{a=0}^1
 \left|D_{e^{T_B\mathcal K_{B,R_B^{\le r+1}}(H_a)}}(j)
       -D_{e^{T_B\mathcal K_{B,R_B^{\le r}}(H_a)}}(j)\right|\le4TC_s|R_B^{\le r+1}|(r+1)^{D+1}e^{-\mu r}
       \mathbf1_{\{j\in R_B^{\le r+1}\}}.
 \end{aligned}
\]
Telescope both influence sequences from radius $R$ all the way up to the exact
channels. In a finite system the restrictions eventually agree with
the full generators on the components meeting $R_B$. If $j$ lies
outside those components, all relevant influences are one. We thus obtain
\begin{align}
 &D_{\Phi_B(H)}(j)\notag\\
 &\quad\le D_{\Phi_B(H^{(0)})}(j)
 +2TC_{\rm loc}\beta\delta|R_B^{\le R}|^3
       \mathbf1_{\{j\in R_B^{\le R}\}}+4TC_s\sum_{\substack{r\ge R:\,j\in R_B^{\le r+1}}}
        |R_B^{\le r+1}|(r+1)^{D+1}e^{-\mu r}.
 \label{qp:site-comparison}
\end{align}
Next, we sum \eqref{qp:site-comparison} with weights $a_B$. The two bounds in
\eqref{qp:density} give
\[
 2TC_{\rm loc}\beta\delta\sum_{B:j\in R_B^{\le R}}a_B|R_B^{\le R}|^3
 \le2TC_{\rm loc}V^3\Xi\,\beta\delta(1+R)^{4D},
\]
and
\[
 \begin{aligned}
 4TC_s\sum_{r\ge R}(r+1)^{D+1}e^{-\mu r}
          \sum_{B:j\in R_B^{\le r+1}}a_B|R_B^{\le r+1}|
 &\le4TC_sV\Xi\sum_{r\ge R}(r+2)^{2D}(r+1)^{D+1}e^{-\mu r}\\
 &\le4TC_sV\Xi\sum_{r\ge R}(r+2)^{3D+1}e^{-\mu r}.
 \end{aligned}
\]
Thus, writing
\begin{equation}
 \mathfrak e(\delta,R)
 =2TC_{\rm loc}V^3\Xi\,\beta\delta(1+R)^{4D}
  +4TC_sV\Xi\sum_{r\ge R}(r+2)^{3D+1}e^{-\mu r},
 \label{qp:error}
\end{equation}
the reference condition \eqref{qp:reference-certificate} gives
\[
 \min_j\sum_Ba_B[1-D_{\Phi_B(H)}(j)]
 \ge\min_j\sum_Ba_B[1-D_{\Phi_B(H^{(0)})}(j)]-\mathfrak e(\delta,R)
 \ge\kappa_0-\mathfrak e(\delta,R).
\]
Since
$(r+2)^{3D+1}e^{-\mu r/2}\le[2+2(3D+1)/\mu]^{3D+1}$ for $r\ge0$,
\[
 \sum_{r\ge R}(r+2)^{3D+1}e^{-\mu r}
 \le\frac{[2+2(3D+1)/\mu]^{3D+1}}{1-e^{-\mu/2}}e^{-\mu R/2}.
\]
Choose
\[
 R=\left\lceil\frac2\mu\log\left(
 1+\frac{16TC_sV\Xi[2+2(3D+1)/\mu]^{3D+1}}
 {\kappa_0(1-e^{-\mu/2})}\right)\right\rceil.
\]
Then
\[
 4TC_sV\Xi\sum_{r\ge R}(r+2)^{3D+1}e^{-\mu r}
 \le\frac{\kappa_0}{4}.
\]
Choose $c\in(0,1]$ small enough that
\[
 2TC_{\rm loc}V^3\Xi\,c(1+R)^{4D}
 \le\frac{\kappa_0}{4}.
\]
The displayed choice of $R$ depends only on $D,\mu,T,C_s,V,\Xi,\kappa_0$,
so $c$ depends only on $D,\mu,T,C_{\rm loc},C_s,V,\Xi,\kappa_0$.
For $\delta\le\delta_*:=c/\beta$, we have $\beta\delta\le c\le1$ and
$\mathfrak e(\delta,R)\le\kappa_0/2$, proving the quantum DS condition
with constant $\kappa_0/2$.
\Cref{qp:lipschitz,qp:boundaries} give the stated parameter dependence
and uniformity in the volume and the reference on-site norms.
\end{proof}
}

{%
\subsection{Example: transverse-field Ising model on a torus}
\label{qp:examples}
This example illustrates stability around an interacting classical reference in
two dimensions. Consider the transverse-field Ising Hamiltonian on the $L\times L$
square torus, with $N=L^2$ spins and zero longitudinal field,
\begin{equation}
 H(g)=-\frac{J_{\mathrm{int}}}{4}\sum_{\langle i,j\rangle}Z_iZ_j
      -g\sum_iX_i,\qquad J_{\mathrm{int}}>0,\quad g\in\RR.
 \label{qp:ising-model}
\end{equation}
Here $X_i,Z_i$ are Pauli operators and $\langle i,j\rangle$ runs
over unordered nearest-neighbour pairs. The reference is
$H^{(0)}=H(0)$ and the perturbation is $Q=-g\sum_iX_i$, so
$\delta=|g|$ in \eqref{qp:local-norm}, although $\|Q\|=N|g|$ is extensive.
At $\beta J_{\mathrm{int}}=1.4$, the classical single-site Dobrushin
test fails: its influence sum is $2\tanh(\beta J_{\mathrm{int}}/2)$,
so contraction requires $\beta J_{\mathrm{int}}<\log3$.
In this normalization, the infinite-volume square-lattice critical
inverse temperature is $2\log(1+\sqrt2)/J_{\mathrm{int}}$~\cite{Onsager1944}.
For every $\beta$ below this value, the classical reference Gibbs measures
satisfy uniform strong spatial mixing on unions of sufficiently large
coarse squares~\cite{MOS1994,Martinelli}.
Since $\beta J_{\mathrm{int}}=1.4<2\log(1+\sqrt2)$,
\Cref{qp:classical-qds} therefore supplies sufficiently large blocks
and bounded update times for which the smoothed heat-bath
channels satisfy the quantum DS condition with constant $1/2$. By \Cref{qp:general-main} with $\kappa_0=1/2$,
the perturbed channels satisfy the quantum DS condition with constant
$1/4$ whenever $|g|\le\delta_*$, for some $\delta_*>0$
independent of $L$.
For the resulting generator $\mathcal L_g$, \Cref{thm:qds} gives
\begin{equation}
 \sup_\rho\|e^{t\mathcal L_g}(\rho)-\sigma_{\beta,H(g)}\|_1
 \le 2Ne^{-t/4},
 \qquad t_{\mathrm{mix}}(\mathcal L_g,\varepsilon)\le4\log(2N/\varepsilon),
 \label{qp:ising-mixing}
\end{equation}
where the supremum includes arbitrary quantum initial states.
At fixed physical and block parameters, \Cref{qi:theorem} yields preparation
to trace-norm error $\varepsilon$ with $N\operatorname{polylog}(N/\varepsilon)$
gates and classical operations.

Efficient classical algorithms for ferromagnetic transverse-field Ising equilibrium properties
and partition functions are already known~\cite{Bravyi2015Stoquastic,BravyiGosset2017Ferromagnets}.
Other results include high-temperature mixing of classical worldline
heat-bath updates~\cite{CrossonSlezak2025} and spectral-width
algorithms~\cite{SinclairVuong2026TFIM}. Vuong \cite{Vuong2026TFIM} also gives
full quantum Gibbs-state preparation.
For $lL\times L$ tori, one can check that the condition for efficiently classical algorithms of~\cite{Vuong2026TFIM} yields $2\beta J_{\mathrm{int}}\le1$ in the $g\to0$ limit, and hence does not cover this small-field example. Thus, as far as we know, our result is the only one to provide efficient preparation for our range of parameters. Note, however, that the results in~\cite{Vuong2026TFIM} covers arbitrary graphs and here we are restricted to those with polynomial growth.

Dissipative ground-state mixing near product references~\cite{ZhanEtAl2026Ground}
and thermal mixing for weak interactions~\cite{SMB2025} concern different
perturbative settings. 

\begin{remark}[Classical mixing and partition functions]
Already at $g=0$, efficient partition function estimation and sampling complexity may differ: on the zero-field ferromagnetic square
torus, worst-case single-site Glauber mixing is rapid above the classical
critical temperature but takes $\exp(\Theta(L))$ time at fixed temperatures
below it~\cite{GheissariSinclairLowTemperature}.
Nevertheless, the ferromagnetic partition function admits a fully polynomial
randomized approximation scheme at every temperature~\cite{JerrumSinclair1993}.
Thus efficient partition-function estimation need not imply rapid Glauber
mixing from arbitrary initial states; efficient sampling from suitable
initializations remains however possible~\cite{GheissariSinclairLowTemperature}.
\end{remark}
}

\section{Conclusion}
We constructed Gibbs-preserving, KMS-reversible block dynamics by
combining local Gibbs resets with quantum belief propagation.
The QBP identity at inverse temperature $\beta/2$ preserves the algebraic identities
of the reset map, while real-time locality controls the effect of
truncating the dynamics to a neighbourhood of the block. The quantum
DS condition compares contraction inside the updated blocks with
the influence of the updates on their boundaries.

For one-dimensional systems, decay of correlations and the locality
estimates imply logarithmic mixing at every finite temperature, with single-exponential temperature dependence in the normalization.
Under the stated locality and geometric assumptions, a strict
quantum DS condition for smoothed heat-bath block channels, uniform in the volume,
persists throughout a system-size-independent neighbourhood in the
local interaction norm. In particular, the unperturbed Hamiltonians may be noncommuting.
For classical references with uniformly bounded interactions, uniform strong spatial mixing implies the quantum DS
condition by \Cref{qp:classical-qds}.
These include systems for which the single-site
criterion fails.

At fixed physical parameters, these dynamics prepare the Gibbs state
with a number of gates and classical operations that is nearly
linear in the system size. The implementation uses local approximations
to the same generators and explicit encodings of the QBP operators.

It would be interesting to improve the complexity of our $1$D algorithm as suggested in Remark~\ref{qi:temperature-remark} to have an end-to-end result that is exponential in a polynomial of the inverse temperature. Finally, the main open question left by our work is how to obtain entropy contraction or a modified log-Sobolev constant from the quantum DS conditions. 

\section*{Acknowledgements}
D.S.F. acknowledges financial support from the Novo Nordisk Foundation
(Grant No. NNF20OC0059939 Quantum for Life) and from the ERC grant GIFNEQ
101163938. CR is supported by France 2030 under the French National
Research Agency award number ``ANR-22-PNCQ-0002''.

\paragraph{Use of artificial intelligence.}
OpenAI Codex with ChatGPT 6 Astra assisted with exploring and checking
proof strategies, and with preparing the manuscript. The current
organisation and exposition of the proofs are due to the authors. The
authors are responsible for the mathematical claims, the final
presentation, and any remaining errors.

\AddToHookNext{env/thebibliography/begin}{%
  \setlength{\itemsep}{3pt}%
}
\bibliographystyle{unsrt}
\bibliography{main_qbp,main_qbp111_supplement}
\appendix

\section{Proofs of the clustering estimates}
\label{app:clustering}

This appendix proves \Cref{qc:clustering}.
We first estimate how adding one positive interaction term
changes the Gibbs operator, then combine many cuts by inclusion--exclusion.

\subsection{A positive cut and its localization}

We use positive-exponent Gibbs operators only in this equilibrium
argument. The change of sign is explicit below. Define
\begin{equation}
 f_\beta(t)=\frac{2}{\pi\beta}
 \log\left(\frac{e^{\pi|t|/\beta}+1}{e^{\pi|t|/\beta}-1}\right).
 \label{qc:kernel}
\end{equation}
This kernel is $p_{2\beta}$ from \Cref{lem:qbp-kernels}. It has unit mass and Fourier transform
\[
 \widehat f_\beta(\omega)
 =\frac{\tanh(\beta\omega/2)}{\beta\omega/2},
 \qquad \widehat f_\beta(0)=1.
\]
The elementary expansion
$\log((1+z)/(1-z))=2\sum_{k\ge0}z^{2k+1}/(2k+1)$ gives the useful
tail formula
\begin{equation}
 \int_{|t|>T}f_\beta(t)\,dt
 =\frac{8}{\pi^2}\sum_{k\ge0}
    \frac{e^{-(2k+1)\pi T/\beta}}{(2k+1)^2}
 \le e^{-\pi T/\beta}.
 \label{qc:kernel-tail}
\end{equation}

\begin{lemma}[Adding a positive interaction across one cut]
\label{qc:positive-cut}
Suppose $K(s)=K(0)+s h$, $0\le s\le1$, is a nearest-neighbor chain
Hamiltonian whose edge terms have norm at most $4J_{\mathrm{int}}$, whose on-site
terms are arbitrary, and where $h\ge0$ is supported on one edge with
$\norm{h}\le4J_{\mathrm{int}}$. Define
\[
 D_s=\frac{\beta}{2}\int_{\RR}f_\beta(t)
                  e^{itK(s)}h e^{-itK(s)}\,dt,
 \qquad F_s'=D_sF_s,\qquad F_0=I.
\]
Then
\begin{equation}
 e^{\beta K(s)}=F_s e^{\beta K(0)}F_s^*,
 \quad 0\le D_s\le2\beta J_{\mathrm{int}}I,
 \quad \norm{F_s}\le e^{2\beta J_{\mathrm{int}}s},
 \quad F_s^*F_s\ge I,
 \quad \norm{F_s^{-1}}\le1.
 \label{qc:cut-identities}
\end{equation}
Let $I_r$ be the interval of $2r$ sites centered on the selected edge.
Define $Q_s$ by
\[
 D_s^{I_r}=\frac{\beta}{2}\int_{\mathbb R}f_\beta(t)
                  e^{itK(s)_{I_r}}h e^{-itK(s)_{I_r}}\,dt,
 \qquad Q_s'=D_s^{I_r}Q_s,\qquad Q_0=I.
\]
The identities and inequalities in
\eqref{qc:cut-identities} also hold for $Q_s$. Moreover,
\begin{equation}
 \norm{F_1-Q_1}\le C_0e^{-\mu_{\mathrm{cut}} r},
 \qquad C_0=C\beta J_{\mathrm{int}} e^{2\beta J_{\mathrm{int}}},
 \qquad \mu_{\mathrm{cut}}\ge\frac{c}{1+\beta J_{\mathrm{int}}},
 \label{qc:cut-locality}
\end{equation}
with universal $C,c>0$. These estimates hold uniformly when any
collection of other edges is deleted or multiplied by parameters in
$[0,1]$.
\end{lemma}

\begin{proof}
In an eigenbasis of $K(s)$, the derivative of $e^{\beta K(s)}$ has
matrix entries
\[
 \left(\frac{d}{ds}e^{\beta K(s)}\right)_{ab}
 =h_{ab}\frac{e^{\beta E_a}-e^{\beta E_b}}{E_a-E_b},
\]
with the continuous value at equal eigenvalues. The Fourier transform
in \eqref{qc:kernel} identifies this derivative with
$D_s e^{\beta K(s)}+e^{\beta K(s)}D_s$. Uniqueness of this linear
matrix differential equation proves the
stated identity.
The positivity and unit mass of $f_\beta$ imply
$0\le D_s\le(\beta\norm{h}/2)I\le2\beta J_{\mathrm{int}}I$. Thus
$\norm{F_s}\le e^{2\beta J_{\mathrm{int}}s}$, while
\[
 \frac{d}{ds}(F_s^*F_s)=2F_s^*D_sF_s\ge0
\]
proves $F_s^*F_s\ge I$ and hence $\norm{F_s^{-1}}\le1$.

For localization, put the on-site terms into their interaction picture.
The on-site propagator is a tensor product of unitaries; conjugation
therefore preserves the support and norm of every edge interaction.
The time-dependent nearest-neighbor commutator recursion then gives,
for the present two-site observable and interval restriction,
\begin{equation}
 \norm{e^{itK(s)}h e^{-itK(s)}
       -e^{itK(s)_{I_r}}h e^{-itK(s)_{I_r}}}
 \le CJ_{\mathrm{int}}\min\{1,e^{v_0J_{\mathrm{int}}|t|-r}\},
 \label{qc:lr}
\end{equation}
where $v_0$ is universal. To recall why no support-size factor appears,
one expands the commutator recursion in paths of overlapping edges
starting at the fixed support of $h$. A path reaching an omitted
boundary edge has at least $r-1$ steps; the bounded number of extensions
at each step and the time-ordered integrals give an exponential bound
$C\norm{h}\,e^{v_0J_{\mathrm{int}}|t|-r}$. Duhamel comparison of the two evolutions
involves only the two edges crossing the interval boundary. The trivial
bound $2\norm{h}$ supplies the first entry in the minimum. This
argument applies unchanged to time-dependent interactions and deleted
edges, and explains the claimed uniformity.

For $\beta J_{\mathrm{int}}>0$, split the QBP integral at
$T=r/(2v_0J_{\mathrm{int}})$. Using \eqref{qc:lr} on the first part and
\eqref{qc:kernel-tail} on the second gives
\[
 \begin{aligned}
 \norm{D_s-D_s^{I_r}}
 &\le C\beta J_{\mathrm{int}}\bigl(e^{-r/2}+e^{-\pi r/(2v_0\beta
J_{\mathrm{int}})}\bigr)
 \\
 &\le C\beta J_{\mathrm{int}} e^{-\mu_{\mathrm{cut}} r},
 \\
 \mu_{\mathrm{cut}}&=\min\left\{\frac{1}{2},\frac{\pi}{2v_0\beta
J_{\mathrm{int}}}\right\}
 \ge\frac{c}{1+\beta J_{\mathrm{int}}}.
\end{aligned}
\]
For $0\le u\le s\le1$, set $U_F(s,u)=F_sF_u^{-1}$.
It satisfies
\[
 \partial_s U_F(s,u)=D_sU_F(s,u),\qquad U_F(u,u)=I,
 \qquad \|U_F(s,u)\|\le e^{2\beta J_{\mathrm{int}}(s-u)}.
\]
Variation of constants and \eqref{qc:cut-identities} therefore give
\begin{align*}
 F_1-Q_1
 &=\int_0^1 U_F(1,s)(D_s-D_s^{I_r})Q_s\,ds,\\
 \|F_1-Q_1\|
 &\le C\beta J_{\mathrm{int}} e^{-\mu_{\mathrm{cut}}r}
       \int_0^1e^{2\beta J_{\mathrm{int}}(1-s)}e^{2\beta J_{\mathrm{int}}s}\,ds
 =C\beta J_{\mathrm{int}} e^{2\beta J_{\mathrm{int}}}e^{-\mu_{\mathrm{cut}}r}.
\end{align*}
This gives \eqref{qc:cut-locality}.
\end{proof}

\subsection{Many cuts and inclusion--exclusion}

\begin{proof}[Proof of \Cref{qc:clustering}]
The case $J_{\mathrm{int}}=0$ is a product state. Assume $\beta
J_{\mathrm{int}}>0$. By homogeneity, we
first treat Hermitian $A,B$ with norm at most one. Relabel the
intervals so that $A$ is to the left of $B$. Put
\[
 K=-H+(N-1)J_{\mathrm{int}}I
   =\sum_i(-V_i)+\sum_i(J_{\mathrm{int}}I-W_i).
\]
Its edge terms are positive with norm at most $2J_{\mathrm{int}}$, and its normalized
positive-exponent Gibbs state is precisely $\sigma_\beta$.
Double the chain and regard the two copies of each site as a single
site. Write 
\[
 \mathbb K=K\otimes I+I\otimes K,
 \qquad \mathbb Z=\Tr(e^{\beta\mathbb K}),
 \qquad
 \Psi=(A\otimes I)(B\otimes I-I\otimes B).
\]
The doubled edge terms are positive with norm at most $4J_{\mathrm{int}}$.
Because the observable supports are disjoint, $\Psi$ is Hermitian
and $\norm{\Psi}\le2$. Direct expansion shows that
\begin{equation}
 \Cov_{\sigma_\beta}(A,B)
 =\frac{\Tr(\Psi e^{\beta\mathbb K})}{\mathbb Z}.
 \label{qc:double}
\end{equation}

Choose $n_{\rm cut}$ edges in the gap between the supports, with disjoint
boundary neighbourhoods
$I_1,\ldots,I_{n_{\rm cut}}$ of $2r$ sites that lie entirely in
that gap.
Let $h_j$ be the corresponding doubled edge terms, and define
\[
 \mathbb K_0=\mathbb K-\sum_{j=1}^{n_{\rm cut}} h_j,
 \qquad
 \mathbb K_\lambda=\mathbb K_0+
                  \sum_{j=1}^{n_{\rm cut}}\lambda_jh_j,
 \qquad \lambda\in\{0,1\}^{n_{\rm cut}}.
\]
If $\lambda_j=0$, the missing edge separates $A$ from $B$ in both
copies. The Gibbs operator factorizes across this cut. Its right
factor is invariant under swapping the two copies, whereas
$B\otimes I-I\otimes B$ is antisymmetric under the same swap.
Consequently,
\[
 \Tr(\Psi e^{\beta\mathbb K_\lambda})=0
 \quad\hbox{unless}\quad \lambda=\boldsymbol1.
\]
Inclusion--exclusion therefore gives the exact identity
\begin{equation}
 \Cov_{\sigma_\beta}(A,B)
 =\frac{\Tr(\Psi\Gamma)}{\mathbb Z},
 \qquad
 \Gamma=\sum_{\lambda\in\{0,1\}^{n_{\rm cut}}}
          (-1)^{n_{\rm cut}-\sum_j\lambda_j}e^{\beta\mathbb K_\lambda}.
 \label{qc:alternating}
\end{equation}

For each $\lambda$, add the selected edges in the order $n_{\rm cut},n_{\rm
cut}-1,\ldots,1$.
The exact QBP identities give
\[
 e^{\beta\mathbb K_\lambda}
 =P_\lambda e^{\beta\mathbb K_0}P_\lambda^*,
 \qquad
 P_\lambda=F_{1,\lambda}\cdots F_{n_{\rm cut},\lambda}.
\]
Here $F_{j,\lambda}=I$ if $\lambda_j=0$; otherwise it is the operator from \Cref{qc:positive-cut} obtained by adding $h_j$ to the Hamiltonian with the later-indexed edges already
added. Since the selected
edges have disjoint neighbourhoods, restriction
to $I_j$
removes all selected edges except the $j$th.
Thus the operator
obtained by restriction to $I_j$ depends only on $\lambda_j$:
write $Q_j(0)=I$, $Q_j(1)=Q_j$, and
\[
 \widetilde P_\lambda=Q_1(\lambda_1)\cdots Q_{n_{\rm cut}}(\lambda_{n_{\rm cut}}).
\]
The factors $Q_j$ have disjoint supports. With
$\eta_r=C_0e^{-\mu_{\mathrm{cut}} r}$, \Cref{qc:positive-cut} and
telescoping of the products imply
\begin{align}
 \norm{P_\lambda},\norm{\widetilde P_\lambda}&\le e^{2\beta J_{\mathrm{int}}n_{\rm
cut}},
 &\norm{P_\lambda^{-1}}&\le1,\notag\\
 \norm{P_\lambda-\widetilde P_\lambda}
 &\le n_{\rm cut}\eta_r e^{2\beta J_{\mathrm{int}}(n_{\rm cut}-1)}.
 \label{qc:products}
\end{align}
For $E_\lambda=\widetilde P_\lambda P_\lambda^{-1}$ this gives
$\norm{E_\lambda-I}\le n_{\rm cut}\eta_r e^{2\beta J_{\mathrm{int}}(n_{\rm cut}-1)}$
and
$\norm{E_\lambda}\le e^{2\beta J_{\mathrm{int}}n_{\rm cut}}$. The elementary inequality
\[
 \nOne{ESE^*-S}
 \le\norm{E-I}(1+\norm{E})\Tr S,
 \qquad S\ge0,
\]
follows by writing $ESE^*-S=(E-I)SE^*+S(E^*-I)$.
Apply it to $S=e^{\beta\mathbb K_\lambda}$ to obtain
\begin{equation}
 \nOne{\widetilde P_\lambda e^{\beta\mathbb K_0}
              \widetilde P_\lambda^*-e^{\beta\mathbb K_\lambda}}
 \le\varepsilon_{\mathrm{cut}}\Tr(e^{\beta\mathbb K_\lambda}),
 \qquad
 \varepsilon_{\mathrm{cut}}=2n_{\rm cut}C_0e^{-\mu_{\mathrm{cut}} r}e^{4\beta
J_{\mathrm{int}}n_{\rm cut}}.
 \label{qc:relative-error}
\end{equation}
Since the missing edge terms are positive,
$\mathbb K_\lambda\le\mathbb K$, and hence
$\Tr(e^{\beta\mathbb K_\lambda})\le\mathbb Z$.
This follows from trace monotonicity: along a positive perturbation
$L+tP$, the derivative of $\Tr(e^{\beta(L+tP)})$ is
$\beta\Tr(Pe^{\beta(L+tP)})\ge0$.
Let $\widetilde\Gamma$ be the alternating sum in
\eqref{qc:alternating} with each $e^{\beta\mathbb K_\lambda}$
replaced by $\widetilde P_\lambda e^{\beta\mathbb K_0}\widetilde P_\lambda^*$.
Summing \eqref{qc:relative-error} gives
\begin{equation}
 \nOne{\Gamma-\widetilde\Gamma}\le2^{n_{\rm cut}}\varepsilon_{\mathrm{cut}}\mathbb Z.
 \label{qc:sum-error}
\end{equation}

It remains to bound the localized alternating sum. Since every $Q_j$
commutes with $\Psi$ and all the $Q_j$ have disjoint supports,
cyclicity and the distributive law give
\begin{equation}
 \Tr(\Psi\widetilde\Gamma)
 =\Tr\!\left[
   \Psi\prod_{j=1}^{n_{\rm cut}}(Q_j^*Q_j-I)e^{\beta\mathbb K_0}\right].
 \label{qc:factorization}
\end{equation}
Put $q=1-e^{-4\beta J_{\mathrm{int}}}$. From
$I\le Q_j^*Q_j\le e^{4\beta J_{\mathrm{int}}}I$ we obtain
\[
 0\le Q_j^*Q_j-I\le qQ_j^*Q_j,
 \qquad
 0\le\prod_j(Q_j^*Q_j-I)\le q^{n_{\rm cut}}\prod_jQ_j^*Q_j.
\]
The product inequalities hold because the factors act on disjoint
supports. For the Hermitian operator $\Psi$, commutation with these
positive products and $\norm{\Psi}\le2$ imply
\[
 \abs{\Tr(\Psi\widetilde\Gamma)}
 \le2q^{n_{\rm cut}}\Tr\!\left[\prod_jQ_j^*Q_j e^{\beta\mathbb K_0}\right]
 \le2q^{n_{\rm cut}}(1+\varepsilon_{\mathrm{cut}})\mathbb Z.
\]
The last step is \eqref{qc:relative-error} for
$\lambda=\boldsymbol1$. Combining this with \eqref{qc:sum-error}
proves the explicit multi-cut estimate
\begin{equation}
 \begin{gathered}
 \abs{\Cov_{\sigma_\beta}(A,B)}
 \le2q^{n_{\rm cut}}(1+\varepsilon_{\mathrm{cut}})+2^{n_{\rm
 cut}+1}\varepsilon_{\mathrm{cut}},\\
 q=1-e^{-4\beta J_{\mathrm{int}}},\quad
 \varepsilon_{\mathrm{cut}}=2n_{\rm cut}C_0e^{-\mu_{\mathrm{cut}} r}e^{4\beta
J_{\mathrm{int}}n_{\rm cut}}.
 \end{gathered}
 \label{qc:cut-bound}
\end{equation}

We now choose the
number and spacing of the cuts so that the
covariance prefactor is independent of temperature. Let
$\log_+u=\max\{0,\log u\}$ and choose
\begin{equation}
 k=1+\left\lceil
  \frac{4\beta J_{\mathrm{int}}+\log2+3+\log_+C_0}{\mu_{\mathrm{cut}}}
 \right\rceil.
 \label{qc:kappa}
\end{equation}
Because $C_0=C\beta J_{\mathrm{int}}e^{2\beta J_{\mathrm{int}}}$ and
$\mu_{\mathrm{cut}}^{-1}\le C(1+\beta J_{\mathrm{int}})$, we have
$1\le k\le C(1+\beta J_{\mathrm{int}})^2$. Write $a=\max\supp(A)$,
$b=\min\supp(B)$, so $R=b-a$, and set
\[
 n_{\rm cut}=\left\lfloor\sqrt{\frac{R-1}{2(k+1)}}\right\rfloor,
 \qquad r=\lceil k n_{\rm cut}\rceil.
\]
When $n_{\rm cut}\ge1$, the boundary neighbourhoods
\[
 I_j=\{a+2(j-1)r+1,\ldots,a+2jr\},
 \qquad 1\le j\le n_{\rm cut},
\]
are disjoint and lie in the gap because
$2rn_{\rm cut}\le2(k+1)n_{\rm cut}^2\le R-1$.
Choose the central edge of each
boundary neighbourhood. Our choice of $k$ ensures
$C_0e^{-n_{\rm cut}\log_+C_0}\le1$ and therefore
\[
 2^{n_{\rm cut}+1}\varepsilon_{\mathrm{cut}}\le4n_{\rm cut}e^{-3n_{\rm cut}},
 \qquad 2q^{n_{\rm cut}}\varepsilon_{\mathrm{cut}}\le4n_{\rm cut}e^{-3n_{\rm cut}}.
\]
With $\zeta=e^{-4\beta J_{\mathrm{int}}}\le1$, $q^{n_{\rm cut}}\le e^{-\zeta n_{\rm
cut}}$, and
$n_{\rm cut} e^{-3n_{\rm cut}}\le e^{-\zeta n_{\rm cut}}$, \eqref{qc:cut-bound} yields
\[
 \abs{\Cov_{\sigma_\beta}(A,B)}\le10e^{-\zeta n_{\rm cut}}.
\]
Taking the integer part
changes the estimate by at most a universal factor: for $R\ge1$,
\[
 n_{\rm cut}\ge\sqrt{\frac{R}{2(k+1)}}-2.
\]
Thus the desired estimate for Hermitian observables follows from
$k+1\le C(1+\beta J_{\mathrm{int}})^2$. If $n_{\rm cut}=0$, the elementary covariance
bound $2$
is covered by the same universal prefactor. Decomposing arbitrary $A,B$ into their Hermitian
real and imaginary
parts increases the bound by
at most a factor of four. Homogeneity then proves
\eqref{qc:covariance}.
Every step is unchanged by restricting the chain to a subinterval.
\end{proof}

\section{Implementation of the smoothed heat-bath dynamics}
\label{app:implementation}

We prove \Cref{qi:theorem} by reducing the global evolution to a
sequence of local block updates. Locality restricts each update to
a neighbourhood of logarithmic radius, where we construct block
encodings of the QBP operators and apply Lindblad simulation.

\subsection{Input model and local approximation}
\label{app:implementation:model}

Throughout this section we assume that the Hamiltonian terms, interaction
graph, inverse temperature, update rates, and prescribed block durations
$0<T_B\le T$ are supplied to sufficiently high precision. We simulate
the update on $B$ for time $T_B$ and use the common upper bound $T$
in all accuracy and cost estimates. We first localize the block updates:

\begin{lemma}
\label{qi:certified-restrictions}
For an integer $u\ge r$, construct $\mathcal K_{B,u}$ by replacing
$H$ with $H_{B^{\le u}}$ in $\mathcal K_B$, retaining the reset
$\mathcal R_B$. Then $\Phi_{B,u}=e^{T_B\mathcal K_{B,u}}$ satisfies
\begin{equation}
 \|\Phi_B-\Phi_{B,u}\|_\diamond
 \le TC_{\rm loc}e^{-\mu u}.
 \label{qi:tail}
\end{equation}
One may take $\mu=\mu_{\rm LR}/[16(1+\beta v_{\rm LR})]$.
The constant $C_{\rm loc}$ depends only on
$\beta J_{\rm int},d,r,D,C_{\rm vol}$ and $b=\sup_B|B|$,
not on $N,T$ or the on-site norms.
For chains and $u\ge\max\{r,b\}$, one may instead use
$\mu=\lambda/2$ and $C_{\rm loc}=e^{\lambda/2}q_\beta$ from
\Cref{qg:chain-locality}.
\end{lemma}
\begin{proof}
The restriction estimate in the proof of \Cref{qp:boundaries}, applied to
$B^{\le v}$ with intersite strength $J_{\rm int}$ and crossing
norm at most $bJ_{\rm int}$, gives
$\|\mathcal K_{B,v+1}-\mathcal K_{B,v}\|_\diamond
\le C_s(v+1)^{D+1}e^{-2\mu v}$.
This step uses only the interaction bounds and locality.
Duhamel's formula and contractivity therefore yield
\[
 \|\Phi_B-\Phi_{B,u}\|_\diamond
 \le TC_s\sum_{v\ge u}(v+1)^{D+1}e^{-2\mu v}
 \le TC_{\rm loc}e^{-\mu u}.
\]
For a proper chain interval and $u\ge|B|$,
$C_{u-1}(B)=B^{\le u}$, so \eqref{qg:channel-diamond} gives the
stated chain constants. Finite-range chains are grouped as in
Section~\ref{q1:section}.
\end{proof}

\subsection{Encodings of the QBP operators}

To simulate a restricted update, we need encodings of its jump
operators and coherent correction. We first approximate the QBP
filter by a positive quadrature, then use it to encode the dressing
propagators and the commutator defining the correction.

An $(\alpha,\eta)$ block encoding of $A$ is a unitary $U$ with
$\|A-\alpha(\langle0|\otimes\one)U(|0\rangle\otimes\one)\|\le\eta$.
We use the sum and product rules of
\cite[Lemmas~52--53]{GSLW2019}, with separate ancillary registers.

Fix $R=B^{\le u}$, $u\ge r$, and put
$W=W_{\partial B}$ and $H_R(s)=H_R-(1-s)W$ for $s\in[0,1]$.
The restricted dressing solves
\begin{equation}
 A_R(s)=-\frac\beta4\mathcal F_{H_R(s)}(W),\qquad
 O_R'=A_RO_R,\qquad O_R(0)=\one.
 \label{qi:ode}
\end{equation}
Recall that
\begin{equation}
 p_\beta(t)=\frac4{\pi\beta}\log\coth\frac{\pi|t|}{\beta},
 \qquad \int_{\mathbb R}p_\beta(t)\,dt=1.
 \label{qi:kernel}
\end{equation}
Real-time Duhamel and \Cref{lem:qbp-kernels} give
\begin{equation}
 \sup_s\|A_R(s)\|\le\frac\beta4\|W\|,\qquad
 \sup_s\|A_R'(s)\|
 \le\frac\beta2\|W\|^2\int_{\mathbb R}|t|p_\beta(t)\,dt=O(1).
 \label{qi:odebounds}
\end{equation}
The local-term representations encode $H_R(s)$ and $W$ with
normalizations $O(|R|)$ and $O(1)$, respectively.

\begin{lemma}[Positive quadrature for the QBP filter]
\label{qi:filter}
Let $0<\xi<1$, $q=1+\log(1/\xi)$, and $|R|\le\operatorname{poly}(q)$.
There are nonnegative weights $w_\nu$ and real nodes $t_\nu$,
independent of $s$ and of the input operator $A$, such that
\begin{equation}
 \begin{gathered}
 \sum_\nu w_\nu\le1+\xi,\qquad
 \sup_{s\in[0,1]}\left\|
 \mathcal F_{H_R(s)}(A)
 -\sum_\nu w_\nu e^{it_\nu H_R(s)}A e^{-it_\nu H_R(s)}
 \right\|\le\xi\|A\|.
 \end{gathered}
 \label{qi:quadrature}
\end{equation}
Their number, description length, and classical construction cost
are polynomial in $q$, and $\max_\nu|t_\nu|=O(\beta q)$.
The same construction applies to $H_R=H_R(1)$.
\end{lemma}
\begin{proof}
Take a bound $\Lambda_R\ge1+\sup_s\|H_R(s)\|=O(|R|)$ and
truncate the integral to $\beta e^{-cq}\le|t|\le c\beta q$.
The omitted mass is $O(qe^{-cq})+O(e^{-2\pi cq})$.
Partition the positive integration interval into subintervals whose lengths are a small
fixed multiple of $\min\{t,\beta,(1+\Lambda_R)^{-1}\}$, shortening
the last subinterval as needed, and reflect the partition to negative times.
There are polynomially many subintervals. On each rescaled subinterval the
integrand is analytic in a fixed Bernstein ellipse, with kernel
bounded by $O(q)$ and
\[
 \|e^{izH_R(s)}Ae^{-izH_R(s)}\|
 \le e^{2\Lambda_R|\operatorname{Im}z|}\|A\|=O(\|A\|).
\]
Thus the Gaussian quadrature estimate
\cite[Theorem~4.5, Eq.~(4.14)]{Trefethen2008Gauss}, with $O(q)$
nodes per subinterval, gives the claimed error. Applying it to matrix
elements between arbitrary unit vectors gives the operator-norm
bound, and applying it to the scalar kernel gives the weight sum.
The nodes and weights are computable in polynomial time \cite[Section~1]{JohanssonMezzarobba2018}, which concludes the proof.
\end{proof}

Combining this quadrature with controlled Hamiltonian simulation
\cite{BCCKS} turns an encoding of $A$ with normalization
$\alpha_A$ into one of $\mathcal F_{H_R(s)}(A)$ with normalization
$2\alpha_A$. For absolute error $\eta$, take
$\xi=\eta/[c(1+\alpha_A)]$ and allocate the remaining error to
the input encoding, controlled evolutions and scalar weights.
The cost is polynomial in $1+\log(1/\eta)$ whenever the input cost
and normalization are. Reversible evaluation of the coefficients
of $H_R-W+sW$ makes the construction coherent in a binary register
for $s$.

\begin{lemma}[Block encodings of the QBP operators]
\label{qi:dressing}
For $0<\eta<1$ and $|R|\le\operatorname{poly}(1+\log(1/\eta))$,
$O_R(1)$ and $O_R(1)^{-1}$ have $(\alpha,\eta)$ block encodings
with constant $\alpha$ and gate and classical construction cost
$\operatorname{poly}(1+\log(1/\eta))$, uniformly over the local inputs.
\end{lemma}
\begin{proof}
The filter construction encodes $A_R(s)$ coherently with constant
normalization $\alpha_0\ge1$, and \eqref{qi:odebounds} bounds its
total variation on $[0,1]$ by $O(1)$.
Partition this interval into a fixed number $m\ge2\alpha_0$ of
equal parts and apply the propagator block-encoding lemma
\cite[Lemma~5]{FLT2023} on each part. The preceding construction
supplies its coefficient oracle, with
$O(1+\log(1/\eta))$ qubits in the binary time register.
The product of these encodings gives $O_R(1)$; reducing the local
errors by the fixed product-normalization factors gives absolute
error $\eta$. Both $m$ and the resulting normalization are
independent of $N$ and $\eta$.
For the inverse, apply the same result to $-A_R(1-s)$, whose
propagator $O_R(1-s)O_R(1)^{-1}$ equals $O_R(1)^{-1}$ at $s=1$.
The reversed-time generator has the same norm and variation bounds.
\end{proof}

The restricted jumps and loss are
\begin{equation}
 L_{B,a}^R=O_R(1)(\sigma_\beta^B)^{1/2}F_{B,a}O_R(1)^{-1},
 \qquad S_B^R=\sum_a(L_{B,a}^R)^*L_{B,a}^R.
 \label{qi:jumps}
\end{equation}
There are at most $d^{2|B|}$ jumps, and the matrices
$(\sigma_\beta^B)^{1/2}F_{B,a}$ have constant dimension and can be computed efficiently. The sum and
product rules therefore give constant-normalization encodings
of the jumps and $S_B^R$. It remains to encode
\begin{equation}
 G_B^R=\mathcal G_{H_R}(S_B^R).
 \label{qi:G}
\end{equation}

\begin{lemma}
\label{qi:correction}
Under the same region-size assumption, $G_B^R$ has an
$(\alpha,\eta)$ block encoding with normalization and gate and
classical construction cost polynomial in $1+\log(1/\eta)$.
The same holds for $\mathcal G_{H_R}(A)$ when the Hermitian input
$A$ has such an encoding at the required precision.
\end{lemma}
\begin{proof}
Let $\alpha_H=O(|R|)$ and $\alpha_A$ normalize encodings of
$H_R$ and $A$. The product and sum rules encode $C=i[H_R,A]$
with normalization $2\alpha_H\alpha_A$. By
\eqref{eq:qbp-G-source},
\[
 \mathcal G_{H_R}(A)=\frac\beta8\mathcal F_{H_R}(C).
\]
The filter construction gives normalization
$O(\beta\alpha_H\alpha_A)$ and the claimed cost. Averaging the encoding with
its adjoint makes the encoded block Hermitian while preserving
the error bound.
\end{proof}

\subsection{Proof of the implementation theorem}

The proof has two steps. Poissonization writes $e^{t\mathcal L}$ as an
average over a random number of random block updates, and a truncation
bounds this number by $K=O(1+Nt+\log(1/\varepsilon))$. Each update is
then localized to a neighbourhood of logarithmic radius and simulated
there to precision $\varepsilon/(2K)$, at polylogarithmic cost.
Throughout, write $\Theta:=1+\log(2N(1+t)/\varepsilon)$.

\begin{proof}[Proof of \Cref{qi:theorem}]

\proofstep{Poissonization and truncation}
Put $a=\sum_Ba_B$ and $P=a^{-1}\sum_Ba_B\Phi_B$, as in \Cref{thm:qds}.
Then $\mathcal L=a(P-\id)$ and
\[
 e^{t\mathcal L}
 =e^{-\lambda}\sum_{k\ge0}\frac{\lambda^k}{k!}P^k
 =\mathbb E[P^Z],\qquad
 \lambda=at,\quad Z\sim\operatorname{Pois}(\lambda).
\]
Thus $e^{t\mathcal L}$ is realized by drawing $Z$ and applying $Z$
independent block updates, each block $B$ being chosen with probability
$a_B/a$. Since the rates are bounded by $a_*$ and $M=O(N)$, we have
$\lambda\le Ma_*t=O(Nt)$. Set
\[
 K=\left\lceil2\lambda+\log\frac4\varepsilon\right\rceil
 =O(Nt+\Theta).
\]
Markov's inequality applied to $e^Z$ gives
$\Pr(Z>K)\le e^{(e-1)\lambda-K}\le\varepsilon/4$. Replacing the law of
$Z$ by its conditional law given $Z\le K$ therefore changes the channel
by at most $2\Pr(Z>K)\le\varepsilon/2$ in diamond norm.

\proofstep{One block update}
Let $\delta=\varepsilon/(2K)$, so that $\log(1/\delta)=O(\Theta)$.
Fix a block $B$ and choose the integer
\[
 u=\max\left\{r,
   \left\lceil\frac1\mu\log\frac{2TC_{\rm loc}}\delta\right\rceil\right\}
 =O(\Theta),
\]
so that $\|\Phi_B-\Phi_{B,u}\|_\diamond\le\delta/2$ by
\Cref{qi:certified-restrictions}. The localized update acts on
$R=B^{\le u}$, and polynomial volume growth gives
$|R|\le bC_{\rm vol}(1+u)^D=\operatorname{poly}(\Theta)$.
By \Cref{qi:dressing,qi:correction}, the at most $d^{2b}$ localized jump
operators \eqref{qi:jumps} and the coherent correction \eqref{qi:G} admit
block encodings whose normalizations and costs are
$\operatorname{poly}(\Theta)$, and we construct them to precision
$\delta/\operatorname{poly}(T,\nu)$, where $\nu$ is their total
normalization. The Lindblad simulation algorithm of
\cite[Theorem~11]{LW}, which accepts such inexact encodings, then
produces a circuit for $e^{T_B\mathcal K_{B,u}}$ with diamond error at most
$\delta/2$. Since $T_B\le T$ and the number of jumps are bounded, its cost,
including compilation into elementary gates, is $\operatorname{poly}(\Theta)$.
Each update is thus implemented to diamond error $\delta$, and
contractivity bounds the accumulated error over at most $K$ updates by
$K\delta=\varepsilon/2$.

Adding the two contributions gives
$\|\widehat{\mathcal E}_t-e^{t\mathcal L}\|_\diamond\le\varepsilon$.

Finally, for the chain construction the block matrices and the
$d^{2\ell}$ jump operators make $C'$ exponential in $\ell$, and $\ell$ in
\eqref{intro:parameters} is exponential in $\beta J_{\mathrm{int}}$;
this gives the doubly exponential prefactor.
\end{proof}

\begin{remark}[Temperature dependence in one dimension]
\label{qi:temperature-remark-appendix}
We believe that the doubly exponential temperature dependence of the
one-dimensional circuit-cost prefactor in \Cref{qi:theorem} can be
reduced to a single exponential in a polynomial of
$\beta J_{\mathrm{int}}$. The present bound comes from dense
representations of the block Gibbs states and separate encodings of
exponentially many jump operators. A possible approach is to prepare
block Gibbs purifications recursively, joining adjacent intervals
through QBP, and then use the purification circuit to encode all
dressed reset jumps collectively. The normalization at each joining
step is controlled by the interaction across a single bond. This
suggests a prefactor $\exp[O((1+\beta J_{\mathrm{int}})^2)]$, while
retaining nearly linear dependence on $N$ and polylogarithmic
dependence on $N/\varepsilon$ at fixed physical parameters, under the
input assumptions of \Cref{qi:theorem}. We leave a detailed complexity
analysis for future work.
\end{remark}

\end{document}